\documentclass[a4paper, 12pt, openright,twoside]{report}

\usepackage{cite}
\usepackage{indentfirst}
\usepackage{fancyhdr}
\usepackage[dvips]{graphicx}
\usepackage[latin1]{inputenc}
\usepackage{newlfont}
\usepackage{amsthm}
\usepackage{amssymb}
\usepackage{mathrsfs}
\usepackage{amsmath}
\usepackage{latexsym}
\usepackage{eucal}
\usepackage{eufrak}
    \theoremstyle{plain}
    \newtheorem{teorema}{Theorem}[section]
    
    \newtheorem{lemma}[teorema]{Lemma}
    \newtheorem{prop}[teorema]{Proposition}
    \theoremstyle{definition}
    
    \theoremstyle{remark}

\DeclareMathOperator{\card}{card}
\begin{document}
\begin{titlepage}
\begin{center}
{{\large{\textsc{Alma Mater Studiorum $\cdot$ Universit\`a di
Bologna}}}} \rule[0.1cm]{13.5cm}{0.1mm}
\rule[0.5cm]{13.5cm}{0.6mm}
{\small{\bf DOTTORATO DI RICERCA IN MATEMATICA\\
XXIII Ciclo \\
Settore scientifico disciplinare: MAT/07}}
\end{center}
\vspace{12mm}
\begin{center}
{\large{\bf A MEAN FIELD MODEL FOR}}\\
\vspace{3mm}
{\large{\bf THE COLLECTIVE BEHAVIOUR}}\\ 
\vspace{3mm}
{\large{\bf OF INTERACTING MULTI-SPECIES PARTICLES:}}\\
\vspace{3mm}
{\large{\bf MATHEMATICAL RESULTS AND APPLICATION}}\\ 
\vspace{3mm}
{\large{\bf TO THE INVERSE PROBLEM}}\\
\vspace{19mm} {\large{\bf Presentata da: MICAELA FEDELE}}
\end{center}
\vspace{15mm}
\noindent
\begin{minipage}[t]{0.47\textwidth}
{\large{\bf Coordinatore Dottorato:\\
Chiar.mo Prof.\\
ALBERTO \\ PARMEGGIANI}}
\end{minipage}
\hfill
\begin{minipage}[t]{0.47\textwidth}\raggedleft
{\large{\bf Relatore:\\
Chiar.mo Prof.\\
PIERLUIGI \\CONTUCCI}}
\end{minipage}
\vspace{15mm}
\begin{center}
{\large{\bf Esame finale anno 2011 }}
\end{center}
\vspace{40cm}
\end{titlepage}

\newpage
\pagenumbering{roman} \chapter*{Introduction}
\addcontentsline{toc}{chapter}{Introduction}
\rhead[\fancyplain{}{\bfseries
Introduction}]{\fancyplain{}{\bfseries\thepage}}
\lhead[\fancyplain{}{\bfseries\thepage}]{\fancyplain{}{\bfseries
Introduction}} 

The study of the normalized sum of random variables and its asymptotic behaviour has been
and continues to be a central chapter in probability and statistical mechanics. When those 
variables are independent the central limit theorem ensures that the sum with square-root normalization
converges toward a Gaussian distribution. The main topic of this thesis is the generalization of the central limit theorem for spin random variable whose interaction is described by a multi-species mean-field Hamiltonian. 
Ellis, Newman and Rosen in \cite{ellis1978limit} and \cite{ellis1980limit} describe the distribution of the suitable normalized sums of spins for mean-field model a la Curie-Weiss.
They have found the conditions, in terms of the interaction, that lead in the thermodynamic limit to a Gaussian behaviour and those who lead to a higher order exponential distribution. We prove analogous results for the multi-species mean-field model under the assumption that the Hamiltonian is a convex function of the magnetizations and the first non vanishing partial derivatives of the pressure functional are all the same order (homogeneity hypothesis).
The extension to non convex interactions or the complete classification of
the limiting distribution beyond the homogeneity hypothesis will be object of further
investigations.\\
The thesis handle other two problems concerning the multi-species mean-field model: the computation of the exact solution of the thermodynamic limit of the pressure and the solution of the inverse problem.

The so called Curie-Weiss model \cite{weiss1907hypothese} has been introduced in 1907 by Pierre Weiss in the attempt to describe Pierre Curie's experimental observations of the magnetic behaviour of some metals such as iron and nickel at different temperature (see \cite{curie1895proprietes}). These materials, after having been exposed to an external magnetic field develop a magnetization with the same sign of the field. Curie noted that when the field was then switched off the materials showed two different behaviours depending on the temperature at which the magnetization was induced. If the temperature was below a critical value the materials retained a degree of magnetization, called spontaneous magnetization, whereas they was not capable of doing this when the temperature was greater or equal to the critical value. As temperature approached the critical value from below the spontaneous magnetization vanished abruptly. 

The multi-species mean-field model is a generalization of the Curie-Weiss model in which the particles are partitioned into an arbitrary number of groups. Both the interaction constant and the external field take different values depending on the groups particles belong to. Such models have been introduced since the 50s to reproduce the phase transition of some materials called metamagnets. In particular in \cite{gorter1956transitions, motizuki1959metamagnetism, kincaid1975phase} and \cite{galam1998metamagnets} a bipartite mean-field model is used to approximate a two-sublattice with nearest neighbor and next-nearest neighbor exchanging interactions. The same model has also been used to study the loss of gibbsianness for a system that evolves according to a Glauber dynamics \cite{kulske2007spin}.\\
In recent times the general version of these models have been proposed in the attempt to describe the large scale behaviour of some socio-economic systems \cite{contucci2007modeling}, assuming that individual's decisions depends upon the decisions of others. Multi-species non-interacting spin models are at the basis of the so called discrete choice \cite{mcfadden2001economic} theory used by the Nobel prize McFadden to forecast the choice of some socio-economic agents. The extension of the discrete choice theory to the interacting case was first suggest in \cite{durlauf1999can} and \cite{brock2001discrete}.
The investigation of the model introduced in \cite{contucci2007modeling} has been pursued at a mathematical level in \cite{gallo2008bipartite}. It has been shown the existence of the thermodynamic limit of the pressure exploiting a monotonicity condition on the Gibbs state of the Hamiltonian (see \cite{bianchi2003thermodynamic}). The factorization of the correlation functions has been proved for almost every choice of parameters and the exact solution of the thermodynamic limit is computed whenever the Hamiltonian is a convex function of the magnetizations.
The phenomenological test of the model has been started in \cite{gallo2008parameter} and it is a topic of current investigations.

The thesis is organized as follows. In the first chapter we consider the general version of the mean-field model. At first we compute the exact solution of the thermodynamic limit following Thompson in \cite{thompson1988classical}. Then we present the results of Ellis, Newman and Rosen on the asymptotic behaviour of the distribution of the sum of spins. Finally we discuss in detail the Curie-Weiss model. In particular we show that the critical phases can be evaluated probabilistically analyzing the distribution of the sum of spins in the thermodynamic limit (see \cite{ellis1978statistics}).

Chapter two focus on the multi-species mean-field model. Firstly we compute the exact solution of the thermodynamic limit. We exploit a tails estimation on the number of configurations that share the same value of the vector of the magnetizations. This technique is the same used by Talagrand to compute the thermodynamic limit for the Curie-Weiss model \cite{talagrand2003spin}.\\
Secondly we analyze the limiting thermodynamic behaviour of the probability distribution of the random vector of the sums of spins of each species. In order to use the Ellis, Newman and Rosen method we consider multi-species model whose Hamiltonian is a convex function of the magnetizations. Under this assumption we show that, when the system reaches its thermodynamic limit, the probability distribution of the suitable normalized random vector of the sums of spins converges to a non-trivial random variable. The behaviour of this random variable crucially depends upon the nature of the maxima points of a function $f$ associated to the model. If the function $f$ admits a unique global maximum point and the determinant of its Hessian matrix computed in this point is different from zero, the suitable normalized random vector of the sums of spins converges to a multivariate Gaussian whose covariance matrix can be compute from the mean-field equations. If the function $f$ admits a unique global maximum point in which the determinant of its Hessian matrix vanishes and the order of the first partial derivatives of $f$ different from zero is homogeneous, the suitable normalized random vector of the sums of spins converges to an higher order exponential distribution. The order of this distribution is those of the first partial derivatives of $f$ different from zero. When the function $f$ has more than one maximum point we obtain a similar result whenever the random vector of the magnetizations is close enough to one of the maximum points. The proof of this last result is quite different respect those used to prove an analogous statement in \cite{ellis1980limit}. The reason is that for the multi-species model the Legendre transformation is not a useful tool. We follows the proof used in \cite{ellis1990limit} to compute a limit theorem with conditioning for the Curie-Weiss-Potts model. 

In the last chapter we solve the inverse problem for the Curie-Weiss model as well as for its multi-species generalization. The inverse problem is a procedure to obtain the parameters of a Boltzmann-Gibbs's distribution by knowing average value and correlations of the spins at the thermodynamic equilibrium. We solve it observing that the susceptibility matrix can be written as function of the interaction parameters as well as function of the average value and the correlations of the spins. Once the interaction parameters are computed, the field is obtained inverting the mean-field equations. The inverse problem for the mean-field approximation of the Ising model is solved in \cite{tanaka1998mean, roudi2009statistical} and \cite{roudi2009ising}.
We also show how to use this method to obtain the maximum likelihood estimator of the parameters of the two models from a suitable sample of empirical data.

\clearpage{\pagestyle{empty}\cleardoublepage} 

 \tableofcontents
\rhead[\fancyplain{}{\bfseries\leftmark}]{\fancyplain{}{\bfseries\thepage}}
\lhead[\fancyplain{}{\bfseries\thepage}]{\fancyplain{}{\bfseries
Index}} \clearpage{\pagestyle{empty}\cleardoublepage}

\pagenumbering{arabic}

\chapter{The Mean-Field Model}
\lhead[\fancyplain{}{\bfseries\thepage}]{\fancyplain{}{\bfseries\rightmark}}
In this chapter after defining the more general case of mean-field model we compute the limit for large $N$ of the pressure function and of the di\-stribution of the normalized sum of spins. As an example, we analyze the Curie-Weiss model.

\section{The model}
We consider a system composed of $N$ particles that interact with each other and with an external magnetic field. In particular the interaction between two spins is independent from their distance.  Such system is defined by the Hamiltonian:
\begin{equation}\label{Hamiltoniana.1}
H_{N}(\boldsymbol{\sigma})=-\frac{J}{2N}\sum_{i,j=1}^{N}\sigma_{i}\sigma_{j}-h\sum_{i=1}^{N}\sigma_{i}
\end{equation}\\
\noindent where $\sigma_{i}$ is the spin of the particle $i$, the parameter $J>0$ is the coupling constant and $h$ is the value of the magnetic field. The probability of a configuration of spins $\boldsymbol{\sigma}=(\sigma_{1},\dots,\sigma_{N})$ is given by the measure of Boltzmann-Gibbs:
\begin{equation}\label{misura.BG.1}
P_{N,J,h}\{\boldsymbol{\sigma}\}=\frac{\exp(- H_{N}(\boldsymbol{\sigma}))}{Z_{N}(J,h)}\prod_{i=1}^{N}d\rho(\sigma_{i})
\end{equation}
%\begin{equation}\label{misura.BG.2}
% P_{N,J,h}\{\boldsymbol{\sigma}\}=\frac{1}{Z_{N}(J,h)}\exp(- H_{N}(\boldsymbol{\sigma}))\pi_{N}P_{\rho}\{\boldsymbol{\sigma}\}
%\end{equation}
\noindent where $Z_{N}(J,h)$ is the canonical partition function:
\begin{equation*}
Z_{N}(J,h)=\int_{\mathbb{R}^{N}}\exp(-H_{N}(\boldsymbol{\sigma}))\prod_{i=1}^{N}d\rho(\sigma_{i})
\end{equation*}
%\begin{equation}\label{partizione2}
%Z_{N}(J,h)=\int_{\Omega_{N}}\exp(-\beta H_{N}(\boldsymbol{\sigma}))\pi_{N}P_{\rho}(d\boldsymbol{\sigma})
%\end{equation}
\noindent and $\rho$ is the distribution of a single spin in the absence of interaction with the other spins. In (\ref{misura.BG.1}) we do not write the inverse temperature $\beta$ because we consider it inclosed in the parameters of the Hamiltonian. We assume that $\rho$ is a non-degenerate Borel probability measure on $\mathbb{R}$ and satisfies\\
\begin{equation}\label{condizione.ro}
 \int_{\mathbb{R}}\exp\bigg(\frac{ax^{2}}{2}+bx\bigg)d\rho(x)<\infty\quad\quad\text{all}\;a,b\in\mathbb{R},\;a>0.
\end{equation}\\
It is easy to check that the measure 
\begin{equation*}
\bar{\rho}(x)=\frac{1}{2}\Big(\delta(x-1)+\delta(x+1)\Big)
\end{equation*}
where $\delta(x-x_{0})$, with $x_{0}\in\mathbb{R}$ denotes the unit point mass with support at $x_{0}$, verified the condition (\ref{condizione.ro}). The model defined by the Hamiltonian (\ref{Hamiltoniana.1}) and the distribution (\ref{misura.BG.1}) whit $\rho=\bar{\rho}$ is called model of Curie-Weiss. We explain this model in detail at the end of the chapter.\\
Considered the distribution (\ref{misura.BG.1}), it is possible to compute the expected value of any observable $\psi(\boldsymbol{\sigma})$ of interest 
%We call  its expectation value with respect to the Boltzmann-Gibbs measure (\ref{misura.BG.1})
\begin{equation}\label{stato.Gibbs}
\langle \psi(\boldsymbol{\sigma})\rangle_{BG}=\dfrac{\displaystyle{\int_{\mathbb{R}^{N}}}\psi(\boldsymbol{\sigma})\exp(- H_{N}(\boldsymbol{\sigma}))\prod_{i=1}^{N}d\rho(\sigma_{i})}{\displaystyle{\int_{\mathbb{R}^{N}}}\exp(-H_{N}(\boldsymbol{\sigma}))\prod_{i=1}^{N}d\rho(\sigma_{i})}.
\end{equation}
The expected value $\langle \psi(\boldsymbol{\sigma})\rangle_{BG}$ is called Gibbs state of the observable $\psi(\boldsymbol{\sigma})$.
The main observable of the mean-field model is the magnetization $m_{N}(\boldsymbol{\sigma})$ of a configuration $\boldsymbol{\sigma}$:
\begin{equation*}
m_{N}(\boldsymbol{\sigma})=\frac{1}{N}\sum_{i=1}^{N}\sigma_{i}.
\end{equation*}
We claim that since $\sum_{i,j}\sigma_{i}\sigma_{j}=(\sum_{i}\sigma_{i})^{2}$, the Hamiltonian (\ref{Hamiltoniana.1}) can be written as function of the magnetization:
\begin{equation}\label{Hamiltoniana.2}
H_{N}(\boldsymbol{\sigma})=-N\bigg(\frac{J}{2}m_{N}^{2}(\boldsymbol{\sigma})+hm_{N}(\boldsymbol{\sigma})\bigg).
\end{equation}
Rather then compute directly the Gibbs state (\ref{stato.Gibbs}) of $m_{N}(\boldsymbol{\sigma})$ we consider the pressure function associated to the model:
\begin{equation*}
p_{N}(J,h)=\frac{1}{N}\ln Z_{N}(J,h).
\end{equation*}
The reason is that the Gibbs state of the magnetization can be obtained differentiating $p_{N}(J,h)$ with respect to $h$
\begin{equation*}
\frac{\partial p_{N}}{\partial h}=\dfrac{\displaystyle{\int_{\mathbb{R}^{N}}} m_{N}(\boldsymbol{\sigma})\exp(- H_{N}(\boldsymbol{\sigma}))\prod_{i=1}^{N}d\rho(\sigma_{i})}{\displaystyle{\int_{\mathbb{R}^{N}}}\exp(-H_{N}(\boldsymbol{\sigma}))\prod_{i=1}^{N}d\rho(\sigma_{i})}=\langle m_{N}(\boldsymbol{\sigma})\rangle_{BG}.
\end{equation*}

\section{Thermodynamic limit}
A statistical mechanical model is well posed if the Hamiltonian is an intensive quantity of the number of spins. This property is verified if the pressure function $p_{N}(J,h)$, associated to the model, admits limit as $N\rightarrow\infty$. 
The existence of the thermodynamic limit of $p_{N}(J,h)$ for the model defined by Hamiltonian (\ref{Hamiltoniana.1}) and distribution (\ref{misura.BG.1}) follows from a subadditivity property of the logarithm of the partition function (see \cite{barra2008mean}).\\ 
In this section we compute the exact solution of this limit. We have to calculate the partition function $Z_{N}(J,h)$. 
Using the expression of the Hamiltonian (\ref{Hamiltoniana.2}) we can write the partition function in the form: 
\begin{equation*}
Z_{N}(J,h)=\int_{\mathbb{R}}\exp\bigg(N\bigg(\frac{J}{2}m^{2}+hm\bigg)\bigg)d\nu_{m_{N}}(m) 
\end{equation*}\\
\noindent where $\nu_{m_{N}}$ denotes the distribution of $m_{N}(\boldsymbol{\sigma})$ on $(\mathbb{R}^{N},\prod_{i=1}^{N}\rho(\sigma_{i}))$.\\

\noindent Since $J>0$, by a Gaussian transform the following identity holds:\\
\begin{equation}\label{trasformazione.gaussiana}
 \exp\bigg(\frac{N J}{2}m^{2}\bigg)=\bigg(\frac{N J}{2\pi}\bigg)^{\frac{1}{2}}\int_{\mathbb{R}}\exp\bigg(NJ\Big(xm-\frac{1}{2}x^{2}\Big)\bigg)dx.
\end{equation}\\
By (\ref{trasformazione.gaussiana}) we can write:
\begin{align*}
 Z_{N}(J,h)&=\bigg(\frac{N J}{2\pi}\bigg)^{\frac{1}{2}}\!\!\iint_{\mathbb{R}^{2}}\!\!\!\exp\bigg(N J\Big(xm-\frac{1}{2}x^{2}\Big)+N h m\bigg)d\nu_{m_{N}}(m)dx\nonumber\\
&=\bigg(\frac{N J}{2\pi}\bigg)^{\frac{1}{2}}\!\!\!\int_{\mathbb{R}}\!\exp\bigg(\!-\frac{N J}{2}x^{2}\bigg)\!\int_{\mathbb{R}}\exp\Big(N m(Jx+ h)\Big)d\nu_{m_{N}}(m)dx
\end{align*}
where
\begin{align*}
\int_{\mathbb{R}}\exp\Big(N m(Jx+h)\Big)d\nu_{m_{N}}(m) &=\int_{\mathbb{R}^{N}}\exp\bigg(\sum_{i=1}^{N}\sigma_{i}(Jx+h)\bigg)\prod_{i=1}^{N}d\rho(\sigma_{i})\nonumber\\
&=\prod_{i=1}^{N}\int_{\mathbb{R}}\exp\Big(\sigma_{i}(Jx+h)\Big)d\rho(\sigma_{i}).
\end{align*}
Thus integrating over the spins we obtain
\begin{equation*}
  Z_{N}(J,h)=\bigg(\frac{N J}{2\pi}\bigg)^{\frac{1}{2}}\int_{\mathbb{R}}\exp(Nf(x))dx
\end{equation*}
where
\begin{equation}\label{funzionale.pressione}
 f(x)=-\frac{ J}{2}x^{2}+\ln\int_{\mathbb{R}}\exp( s(Jx +h ))d\rho(s).
\end{equation}\\
The condition (\ref{condizione.ro}) on the measure $\rho$ assures that the integral of the expression (\ref{funzionale.pressione}) is finite for each $x\in\mathbb{R}$. Indeed, if $x<0$ choosing $a= J$ and $b= h$ in (\ref{condizione.ro}) we have
\begin{equation*}
\int_{\mathbb{R}}\!\exp( s(Jx +h))d\rho(s)\!<2\!\!\int_{\mathbb{R}}\!\!\exp\bigg(\frac{ J}{2} s^{2} + h s\bigg)d\rho(s)+\!\!\int_{2x}^{0}\!\!\exp( s(Jx +h ))d\rho(s)\!<\!\infty
\end{equation*}\\
If $x>0$ choosing $a= Jx$ and $b= h$ in (\ref{condizione.ro}) we have\\
\begin{equation*}
\int_{\mathbb{R}}\!\exp( s(Jx +h))d\rho(s)\!<\!2\!\!\int_{\mathbb{R}}\!\!\exp\bigg(\!\frac{ Jx}{2} s^{2} + h s\!\bigg)d\rho(s)+\!\!\int_{0}^{2}\!\!\exp( s(Jx +h ))d\rho(s)\!<\!\infty
\end{equation*}\newpage
We can state the following
\begin{prop}\label{proposizione.f}
	Let $f(x)$ be the function defined by (\ref{funzionale.pressione}). Then
	\begin{enumerate}
		\item $f(x)$ is a real analytic function that tends to $-\infty$ as $|x|\rightarrow\infty$;
		\item $f(x)$ admits a finite number of global maximum points;
		\item for any positive $N\in\mathbb{N}$
		\begin{equation}\label{proprieta.f}
		\int_{\mathbb{R}}\exp(Nf(x))dx<\infty
		\end{equation}
		\item if $\mu$ is a global maximum point of $f(x)$
		\begin{equation}\label{proprieta.f.2}
		\lim_{N\rightarrow\infty}\frac{1}{N}\ln\int_{\mathbb{R}}\exp(Nf(x))dx=f(\mu)
		\end{equation}
	\end{enumerate}	
\end{prop}\vspace{0.5cm}

\begin{proof}
	For complex $z$ and $L>0$ we have
	\begin{multline}\label{dim.1.passaggio.1}
	\bigg|\int_{\mathbb{R}}\exp( s(Jz+h))d\rho(s)\bigg|\leq\int_{|s|\leq L}\exp(| s(Jz+h)|)d\rho(s)\\+\int_{|s|> L}\exp(| Jsz|)		\exp( hs)d\rho(s)
	\end{multline}
	where 
	\begin{equation}\label{dim.1.passaggio.2}
	\int_{|s|\leq L}\exp(| s(Jz+h)|)d\rho(s)\leq\rho([-L,L])\exp( L|Jz+h|)
	\end{equation}
        and
	\begin{align}\label{dim.1.passaggio.3}
	\int_{|s|> L}\!\!\!\!\!\!\exp(| Jsz|)\exp( hs)d\rho(s) &\!\leq\!\!\int_{|s|>L}\!\!\!\!\!\exp\bigg(\frac{ J}{2}\Big(s^{2}+|z|^{2}\Big)\!		\bigg)\exp( h s)d\rho(s)\nonumber\\
	&\!=\!\exp\!\bigg(\frac{ J}{2}|z|^{2}\bigg)\!\!\int_{|s|>L}\!\!\!\!\!\exp\bigg(\frac{ J}{2}s^{2}+ h s\bigg)d\rho(s).
	\end{align}
	By inequalities (\ref{dim.1.passaggio.2}) and (\ref{dim.1.passaggio.3}) and the condition (\ref{condizione.ro}) on the measure $\rho$, the expression (\ref{dim.1.passaggio.1}) has order $o(\exp( J|z|^{2}/2))$ as $|z|\rightarrow\infty$, thus the function $f$ is real analytic and tends to $-\infty$ as $|x|\rightarrow\infty$.\\
	
	To prove the second statement we take a sequence $x_{l}$ of $\mathbb{R}$ such that 
	\begin{equation*}
	\lim_{l\rightarrow\infty}f(x_{l})=\sup_{x\in\mathbb{R}}f(x)=L\leq\infty.
	\end{equation*}
	Since $f(x)$ tends to $-\infty$ as $|x|\rightarrow\infty$, the sequence $x_{l}$ is bounded. Thus we can take a 			subsequence $x_{k_{l}}$ that tends to $x_{0}\in\mathbb{R}$ as $l\rightarrow\infty$. Hence by continuity of $f$ we 		have
	\begin{equation*}
	f(x_{0})=\lim_{l\rightarrow\infty}f(x_{k_{l}})=\sup_{x\in\mathbb{R}}f(x).
	\end{equation*}
	This shows that $f$ must have at least one maximum point. Since $f(x)\rightarrow-\infty$ as $|x|\rightarrow\infty$ the 	point $x_{0}$ and other possible global maximum points must belong to a compact set. The analyticity of $f$ ensures 	that inside this set the global maximum points are finite in number.\\
	
	We prove the statement (\ref{proprieta.f}) by induction on $N$. For $N=1$ we have:
	\begin{align}\label{dim.1.passaggio.4}
	\int_{\mathbb{R}}e^{f(x)}dx &=\iint_{\mathbb{R}^{2}}\exp\bigg(-\frac{ J}{2}x^{2}+ s(Jx+h)\bigg)d\rho(s)dx\nonumber\\
	&=\iint_{\mathbb{R}^{2}}\exp\bigg(-\frac{ J}{2}(x-s)^{2}\bigg)\exp\bigg(\frac{ J}{2}s^{2}+ h s\bigg)d\rho(s)dx\nonumber\\
	&=\bigg(\frac{2\pi}{ J}\bigg)^{\frac{1}{2}}\int_{\mathbb{R}}\exp\bigg(\frac{ J}{2}s^{2}+ h s\bigg)d\rho(s).
	\end{align}
	%where we use the identity 
	%\begin{equation}
	%(1/\sqrt{2\pi\sigma^{2}})\int_{\mathbb{R}}\exp(-(x-m)^{2}/2\sigma^{2})dx=1
	%\end{equation}
	Since the integral on the left-hand side of (\ref{dim.1.passaggio.4}) is finite by the condition (\ref{condizione.ro}) on the measure $\rho$ (with $a= J$ and $b= h$) the result (\ref{proprieta.f}) is proved for $N=1$.
	Supposed true the inductive hypothesis 
	\begin{equation}\label{ipotesi.induttiva}
	 \int_{\mathbb{R}}e^{(N-1)f(x)}dx<\infty
	\end{equation}
	and defined $F=\max\{f(x)|x\in\mathbb{R}\}$, the statement follows because
	\begin{equation}\label{dim.1.passaggio.5}
	\int_{\mathbb{R}}e^{Nf(x)}dx\leq e^{F}\int_{\mathbb{R}}e^{(N-1)f(x)}dx
	\end{equation}
where the right-hand side of (\ref{dim.1.passaggio.5}) is finite by the inductive hypothesis (\ref{ipotesi.induttiva}).\\

	To prove the statement (\ref{proprieta.f.2}) we write 
	\begin{equation*}
	\int_{\mathbb{R}}\exp(Nf(x))dx=e^{Nf(\mu)}I_{N}  
	\end{equation*}
	where 
	\begin{equation*}
	I_{N}=\int_{\mathbb{R}}\exp(N(f(x)-f(\mu)))dx\;.
	\end{equation*}\\
	Since $f(x)-f(\mu)\leq 0$, the integral $I_{N}$ is a decreasing function of $N$.\\
	Thus
	\begin{equation*}
	\ln\int\exp(Nf(x))dx\leq Nf(\mu)+\ln I_{1}
	\end{equation*}
	hence we obtain 
	\begin{equation}\label{dim.1.passaggio.6}
	\lim_{N\rightarrow\infty}\frac{1}{N}\ln\int\exp(Nf(x))dx\leq f(\mu).
	\end{equation}\\
	By continuity of the function $f$, given any $\epsilon>0$, there exists $\delta_{\epsilon}>0$ such that as $|x-\mu|<\delta_{\epsilon}$ we have $f(x)-f(\mu)>-\epsilon$.\\ 
	Thus
	\begin{equation*}
	I_{N}\geq\int_{\mu-\delta_{\epsilon}}^{\mu+\delta_{\epsilon}}\exp(N(f(x)-f(\mu)))dx>2\delta_{\epsilon}e^{-N\epsilon}
	\end{equation*}
	and in the limit
	\begin{equation}\label{dim.1.passaggio.7}
	\lim_{N\rightarrow\infty}\frac{1}{N}\ln\int\exp(Nf(x))dx\geq f(\mu)-\epsilon.
	\end{equation}
	Since $\epsilon$ is arbitrary the statement (\ref{proprieta.f.2}) follows from the inequalities (\ref{dim.1.passaggio.6}) and (\ref{dim.1.passaggio.7}).
\end{proof}\vspace{0.5cm}
The proposition \ref{proposizione.f} implies that in the thermodynamic limit
\begin{align*}
\lim_{N\rightarrow\infty}p_{N}(J,h)&=\lim_{N\rightarrow\infty}\bigg(\frac{1}{2N}\ln\bigg(\frac{ JN}{2\pi}\bigg)+\frac{1}{N}\ln\int_{\mathbb{R}}\exp(Nf(x))dx\bigg)\nonumber\\
&=\max_{x\in\mathbb{R}}f(x).
\end{align*} 
Differentiating $f(x)$ with respect to $x$ and looking for the values that va\-nishied the derivative we find the following condition:\\
\begin{equation}\label{condizione.estremale}
x=\dfrac{\displaystyle{\int_{\mathbb{R}}}s\exp\Big( s (Jx+h)\Big)d\rho(s)}{\displaystyle{\int_{\mathbb{R}}}\exp\Big( s (Jx+h)\Big)d\rho(s)}.
\end{equation}\\
Evidently a maximum point of $f$ must satisfy the condition (\ref{condizione.estremale}). Let $\mu$ be a global maximum point of $f$. If we differentiate the thermodynamic limit of $p_{N}$ with respect to $h$ we find:
%The extremal condition of $f$ is $\Phi'(x)= Jx$ if $\mu$ verify this condition
%called mean field equation.
\begin{align*}
\frac{\partial}{\partial h}\Big(\lim_{N\rightarrow\infty}p_{N}(J,h)\Big) &= - J\mu\frac{\partial \mu}{\partial h}+\Big( J\frac{\partial \mu}{\partial h}+1\Big)\dfrac{\displaystyle{\int_{\mathbb{R}}}s\exp( s (J\mu +h))d\rho(s)}{\displaystyle{\int_{\mathbb{R}}}\exp( s (J\mu+h))d\rho(s)}\notag\\
&=- J\mu\frac{\partial \mu}{\partial h}+ J\mu\frac{\partial \mu}{\partial h}+\mu\notag\\
&=\mu.
\end{align*}
Thus $\mu$ is precisely the magnetization of the system in the thermodynamic limit.

\section{Asymptotic behaviour of the sum of spins}
The study of the normalized sum of random variables and its asymptotic behaviour is
a central chapter in probability and statistical mechanics. When those 
variables are independent the central limit theorem ensures that the sum with square-root normalization
converges toward a Gaussian di\-stribution. Spins whose interaction is described by Hamiltonian (\ref{Hamiltoniana.1}) and distribution (\ref{misura.BG.1}) are not independent random variables, thus the central limit theorem can't help us to understand the behaviour of their sum
\begin{equation*}
S_{N}(\boldsymbol{\sigma})=\sum_{i=1}^{N}\sigma_{i}.
\end{equation*}
The generalization of the central limit theorem to this type of random variable was performed by Ellis, Newman and Rosen \cite{ellis1978limit, ellis1980limit}.\\ 
They found that the behaviour in the thermodynamic limit of the probability distribution of $S_{N}(\boldsymbol{\sigma})$ depends on the number and the type of the maxima points of the functional $f$ given by (\ref{funzionale.pressione}).\\

Before stating and proving these results  we have to clarify the meaning of type of a maximum point.\\ 

Let $\mu_{1},\dots,\mu_{P}$ be the global maxima points of the function $f$ defined in (\ref{funzionale.pressione}). For each $p$ there exists a positive integer $k_{p}$ and a negative real number $\lambda_{p}$ such that around $\mu_{p}$ we can write:\\
\begin{equation*}
f(x)=f(\mu_{p})+\lambda_{p}\frac{(x-\mu_{p})^{2k_{p}}}{(2k_{p})!}+o((x-\mu_{p})^{2k_{p}}).
%\quad\text{as $s\rightarrow \mu_{p}$}
\end{equation*}\\
The numbers $k_{p}$ and $\lambda_{p}$ are called, respectively, the type and the strength of the maximum point $\mu_{p}$. We define the maximal type $k^{*}$ of the function $f$ as the largest of the $k_{p}$. Define the function
\begin{equation*}
B(x;y)=f(x+y)-f(y).
\end{equation*} 
For each $p=1,\dots,P$ there exists $\delta_{p}>0$ sufficiently small such that for $|x|<\delta_{p} N^{1/2k}$ as $N\rightarrow\infty$\\
\begin{align}\label{proprieta.B}
NB\Big(\frac{x}{N^{1/2k}};\mu_{p}\Big)&=\frac{\lambda }{(2k)!}x^{2k}+o(1)P_{2k}(x)\nonumber\\ \\
NB\Big(\frac{x}{N^{1/2k}};\mu_{p}\Big)&\leq\frac{1}{2}\frac{\lambda }{(2k)!}x^{2k}+P_{2k+1}(x)\nonumber
\end{align}\\
where $P_{2k}(x)$ is a polynomial of $2k$ degree and $P_{2k+1}$ is a polynomial of $2k+1$ degree.
\newpage
Normalizing $S_{N}(\boldsymbol{\sigma})$ by the total number of spins we obtain the magnetization. Its behaviour in the thermodynamic limit is specified by the following
\begin{teorema}\label{teorema.1}
        Let $\mu_{1},\dots,\mu_{P}$ be the global maximum points of maximal type $k^{*}$ of the function $f(x)$ given by (\ref{funzionale.pressione}). Let $\lambda_{1},\dots,\lambda_{p}$ be respectively the strengths of the maximum points.
        Then as $N\rightarrow\infty$\\
	\begin{equation*}
	m_{N}(\boldsymbol{\sigma})\overset{\mathscr{D}}{\rightarrow}\dfrac{\sum\limits_{p=1}^{P}b_{p}\;\delta(x-\mu_{p})}{\sum\limits_{p=1}^{P}b_{p}}
	\end{equation*}
	where $b_{p}=\lambda_{p}^{-1/2k^{*}}$.
\end{teorema}
\vspace{0.3cm}
We claim that if $f$ admits only one global maximum point $\mu$ of maximal type the limiting distribution of the magnetization is a delta picked in $\mu$. In other world the variance of the magnetization vanishes for large $N$.
When $f$ has more global maximum points of maximal type this result holds only locally around each maximum point. For the proof see \cite{ellis1980limit}.\\
Thus it is important to determinate a suitable normalization of $S_{N}(\boldsymbol{\sigma})$ such that in the thermodynamic limit it converges to a well define random variable. If $f$ has a unique maximum point the problem is solved by the following 
\begin{teorema}\label{teorema.2}
	Suppose that the function $f$ given by (\ref{funzionale.pressione}) has a unique maximum point $\mu$ of type $k$ and strength $\lambda$. Then
	\begin{equation*}
	\bar{S}_{k}(\boldsymbol{\sigma})=\frac{S_{N}(\boldsymbol{\sigma})-N\mu}{N^{1-1/2k}}\overset{\mathscr{D}}{\rightarrow}\begin{cases}
	N\Big(0,(-\lambda)^{-1}-J^{-1}\Big) & \text{if $k=1$} \\\\
	\exp\bigg(\dfrac{\lambda}{(2k)!}x^{2k}\bigg) & \text{if $k>1$}
	\end{cases}
	\end{equation*}
	where $(-\lambda)^{-1}-J^{-1}>0$ for $k=1$.
\end{teorema}

Instead if $f$ has more than one maximum point we have the following local statement
\begin{teorema}\label{teorema.3}
	Assume that $\mu$ is either a nonunique global maximum point of the function $f$ given by (\ref{funzionale.pressione}). Let $k$ and $\lambda$ be respectively the type and the strength of $\mu$. Then there exists $A>0$ such that for all $a\in(0,A)$ if $m_{N}(\boldsymbol{\sigma})\in [\mu-a,\mu+a]$ then
	\begin{equation*}
	\bar{S}_{k}(\boldsymbol{\sigma})=\frac{S_{N}(\boldsymbol{\sigma})-N\mu}{N^{1-1/2k}}\overset{\mathscr{D}}{\rightarrow}\begin{cases}
	N\Big(0,(-\lambda)^{-1}- J^{-1}\Big) & \text{if $k=1$} \\\\
	\exp\bigg(\dfrac{\lambda s^{2k}}{(2k)!}\bigg) & \text{if $k\geq 2$}
	\end{cases}
	\end{equation*}
	where $(-\lambda)^{-1}- J^{-1}>0$ for $k=1$.
\end{teorema}

The result of theorem \ref{teorema.3} is valid also for local maximum points of the function $f$. For the proofs of these theorems it is useful to define the function 
\begin{equation}\label{funzione.fi}
\Phi_{\rho}(x)=\frac{1}{ J}\ln\int_{\mathbb{R}}\exp( s(Jx +h ))d\rho(s).
\end{equation}
We claim that since
\begin{equation}\label{funzione.fi.2}
\Phi_{\rho}(x)=\frac{1}{ J}f(x)+\frac{1}{2}x^{2}
\end{equation}
the first statement of proposition \ref{proposizione.f} ensures that the function $\Phi_{\rho}$ is real analytic. 
The second derivative of $\Phi_{\rho}$ is\\
\begin{align*}
\Phi_{\rho}''(x) &= J\bigg(\dfrac{\int_{\mathbb{R}}s^{2}\exp( s(Jx+h))d\rho(s)}{\int_{\mathbb{R}}\exp( s(Jx+h))d\rho(s)}-\bigg(\dfrac{\int_{\mathbb{R}}s\exp( s(Jx+h))d\rho(s)}{\int_{\mathbb{R}}\exp( s(Jx+h))d\rho(s)}\bigg)^{2}\bigg)\nonumber\\
&= J\bigg(\int_{\mathbb{R}}s^{2}d\rho_{x}(s)-\bigg(\int_{\mathbb{R}}s\;d\rho_{x}(s)\bigg)^{2}\bigg)= J\mathrm{Var}_{\rho_{x}}(Y)
\end{align*}\\
\noindent where $\mathrm{Var}_{\rho_{x}}(Y)$ denotes the variance of a random variable $Y$ whose distribution is 
\begin{equation}\label{misura.ro.x}
 \rho_{x}(s)=\dfrac{\exp( s(Jx+h))d\rho(s)}{\int_{\mathbb{R}}\exp( s(Jx+h))d\rho(s)}.
\end{equation}\\
\noindent Since $\rho$ is a nondegenerate measure, by definition of variance of a random variable, $\Phi_{\rho}''(x)>0$ for any $x\in\mathbb{R}$.\\

The proofs of the theorems \ref{teorema.1} and \ref{teorema.2} also need the following preliminary results:
\begin{lemma}\label{lemma.1}
	Suppose that for each $N$, $X_{N}$ and $Y_{N}$ are independent random variables such that $X_{N}\overset{\mathscr{D}}{\rightarrow}\nu$, where for all $a\in\mathbb{R}$
	\begin{equation*}
	\int e^{iax}d\nu(x)\neq 0.
	\end{equation*}
	Then $Y_{N}\overset{\mathscr{D}}{\rightarrow}\mu$ if and only if $X_{N}+Y_{N}\overset{\mathscr{D}}{\rightarrow}\nu*\mu$.\\
	Where $\nu*\mu$ indicates the convolution of two distribution, that is:
	\begin{equation*}
	\nu*\mu =\int_{-\infty}^{\infty}\nu(x-t)\mu(t)dt.
	\end{equation*}
\end{lemma}
\begin{proof}
 Weak convergence of measures is equivalent to pointwise convergence of characteristic functions.
\end{proof}
\begin{lemma}\label{lemma.2}
	Suppose that the random variable $W \sim N(0;J^{-1})$ is independent of $S_{N}(\boldsymbol{\sigma})$ for all $N\geq 1$. Then given $\gamma$ and $m$ real, the distribution of 
	\begin{equation*}
	\frac{W}{N^{1/2 -\gamma}}+\frac{S_{N}(\boldsymbol{\sigma})-Nm}{N^{1-\gamma}}
	\end{equation*}
	is given by
	\begin{equation}\label{tesi.lemma.2}
	\dfrac{\exp\Big(Nf\Big(\dfrac{s}{N^{\gamma}}+m\Big)\Big)ds}{\displaystyle{\int_{\mathbb{R}}}\exp\Big(Nf\Big(\dfrac{s}{N^{\gamma}}+m\Big)\Big)ds}.
	\end{equation}\\
	where the function $f$ is given by (\ref{funzionale.pressione}).
\end{lemma}
\begin{proof}
	Given $\theta$ real
	\begin{equation*}
	P\bigg\{\frac{W}{N^{1/2 -\gamma}}+\frac{S_{N}(\boldsymbol{\sigma})-Nm}{N^{1-\gamma}}\leq\theta\bigg\}=P\Big\{\sqrt{N}W+S_{N}(\boldsymbol{\sigma})\in E\Big\}
	\end{equation*}
	\noindent where $E=(-\infty,\theta N^{1-\gamma}+Nm]$. The distribution of $\sqrt{N}W+S_{N}(\boldsymbol{\sigma})$ is given by the convolution of the Gaussian $N(0,N J^{-1})$ with the distribution of $S_{N}(\boldsymbol{\sigma})$ 
	\begin{equation*}
	\frac{1}{Z_{N}(J,h)}\exp\bigg(\bigg(\frac{J}{2N}s^{2}+hs\bigg)\bigg)d\nu_{S}(s)
	\end{equation*}
	\noindent where $d\nu_{S}(s)$ denotes the distribution of $S_{N}(\boldsymbol{\sigma})$ on $(\mathbb{R}^{N},\prod_{i=1}^{N}\rho(\sigma_{i}))$.\\
	Thus we have:
%$N$-fold convolution of the measure $\rho$ with itself. Thus:
	\begin{multline*}
	P\Big\{\sqrt{N}W+S_{N}(\boldsymbol{\sigma})\in E\Big\}=\frac{1}{Z_{N}(J,h)}\bigg(\frac{ J}{2\pi N}\bigg)^{\frac{1}{2}}\\
	\times\int_{E}\exp\bigg(-\frac{ J}{2N}t^{2}\bigg)\int_{\mathbb{R}}\exp\bigg( s\bigg(\frac{J}{N}t+h\bigg)\bigg)d\nu_{S}(s)dt
	\end{multline*}
	where
	\begin{equation*}
	\int_{\mathbb{R}}\exp\bigg( s\bigg(\frac{J}{N}t+h\bigg)\bigg)d\nu_{S}(s)=\prod_{i=1}^{n}\int_{\mathbb{R}}\exp\bigg( \sigma_{i}\bigg(\frac{J}{N}t+h\bigg)\bigg)d\rho(\sigma_{i}).
	\end{equation*}\\
	\noindent If we make the following change of variable 
	\begin{equation*}
	x=\frac{t-Nm}{N^{1-\gamma}}
	\end{equation*}
	\noindent and we integrate over the spins, we obtain:
	\begin{equation}\label{pim}
	P\Big\{\sqrt{N}W+S_{N}(\boldsymbol{\sigma})\in E\Big\}=\frac{1}{Z_{N}(J,h)}\bigg(\frac{ JN^{1-2\gamma}}{2\pi}\bigg)^{\frac{1}{2}}\!\!\!\int_{-\infty}^{\theta}\!\!\!\exp\bigg(Nf\bigg(\frac{x}{N^{\gamma}}+m\bigg)\bigg)dx.
	\end{equation}
	Taking $\theta\rightarrow\infty$ the (\ref{pim}) gives an equation for $Z_{N}(J,h)$ which when substituted back yields the distribution (\ref{tesi.lemma.2}). The integral in (\ref{tesi.lemma.2}) is finite by (\ref{proprieta.f}).
\end{proof}
We remark that for $\gamma<1/2$ the random variable $W$ does not contribute to the limit of the distribution (\ref{tesi.lemma.2}) as $N\rightarrow\infty$.
\begin{lemma}\label{lemma.3}
	Defined $F=\max\{f(x)|x\in\mathbb{R}\}$, let $V$ be any closed (possibly unbounded) subset of $\mathbb{R}$ which contains no global maxima of $f(x)$. Then there exists $\epsilon>0$ so that 
	\begin{equation}\label{tesi.lemma.3}
	e^{-NF}\int_{V} e^{Nf(x)}dx=O(e^{-N\epsilon}) \quad\quad\quad N\rightarrow\infty.
	\end{equation}
\end{lemma}
\begin{proof}
	$V$ contains no global maxima of $f(x)$, thus:
	\begin{equation*}
	\sup_{x\in V}f(x)\leq\sup_{x\in\mathbb{R}}f(x)-\epsilon=F-\epsilon
	\end{equation*}
	hence
	\begin{align}\label{dim.3.passaggio}
	e^{-NF}\!\!\!\int_{V}\!\! e^{Nf(x)}dx &<e^{-NF}e^{(N-1)(F-\epsilon)}\int_{V} e^{f(x)}dx\nonumber\\
	&\leq e^{-NF}e^{N(F-\epsilon)}\bigg(e^{(F-\epsilon)}\int_{\mathbb{R}} e^{f(x)}dx\bigg)\nonumber\\
        &=e^{-NF}\bigg(e^{(F-\epsilon)}\bigg(\frac{2\pi}{ J}\bigg)^{\frac{1}{2}}\!\!\!\int_{\mathbb{R}}\exp\!\bigg(\frac{ J}{2}x^{2}+ h x\bigg)d\rho(x)\bigg)
	\end{align}
	The condition (\ref{condizione.ro}) on the measure $\rho$ (with $a= J$ and $b= h$) assures that the latter passage of (\ref{dim.3.passaggio}) is $O(e^{-N\epsilon})$ as $N\rightarrow\infty$. This proved the (\ref{tesi.lemma.3}).
\end{proof} 
At last we can prove the theorems \ref{teorema.1} and \ref{teorema.2}.\\
\noindent {\it Proof of Theorem \ref{teorema.1}}.
	By definition $m_{N}(\boldsymbol{\sigma})=S_{N}(\boldsymbol{\sigma})/N$, thus by lemmas \ref{lemma.1} and \ref{lemma.2} (with $\gamma=0$ and $m=0$) we know that
	\begin{equation*}
	\frac{W}{\sqrt{N}}+\frac{S_{N}(\boldsymbol{\sigma})}{N}\sim\dfrac{e^{Nf(x)}dx}{\displaystyle{\int_{\mathbb{R}}}e^{Nf(x)}dx}
	\end{equation*}
	where $W\sim N(0, J^{-1})$. We have to prove that for any bounded continuous function $\phi(x)$
	\begin{equation}\label{teo.1.passaggio.1}
	\dfrac{\displaystyle{\int_{\mathbb{R}}}e^{Nf(x)}\phi(x)dx}{\displaystyle{\int_{\mathbb{R}}}e^{Nf(x)}dx}\rightarrow\dfrac{\sum\limits_{p=1}^{P}\phi(\mu_{p})b_{p}}{\sum\limits_{p=1}^{P}b_{p}}.
	\end{equation}
	Consider $\delta_{1},\dots,\delta_{P}$ such that the conditions expressed in  (\ref{proprieta.B}) are satisfied we choose $\bar{\delta}=\min\{\delta_{p}\;|\;p=1,\dots,P\}$, decreasing it (if necessary) to assure that $0<\bar{\delta}<\min\{|\mu_{p}-\mu_{q}|:1\leq p\neq q\leq P\}$. Denoted by $V$ the closet set
	\begin{equation*}
	V=\mathbb{R}-\bigcup_{p=1}^{P}(\mu_{p}-\bar{\delta},\mu_{p}+\bar{\delta})
	\end{equation*}
	by lemma \ref{lemma.3} there exists $\epsilon>0$ such that as $N\rightarrow\infty$
	\begin{equation}\label{teo.1.passaggio.2}
	e^{-NF}\int_{V}e^{Nf(x)}\phi(x)dx=O(e^{-N\epsilon}). 
	\end{equation}
	For each $p=1,\dots,P$ we have
	\begin{align*}
	N^{1/2k^{*}}e^{-NF}&\int_{\mu_{p}-\bar{\delta}}^{\mu_{p}+\bar{\delta}}e^{Nf(x)}\phi(x)dx\nonumber\\ &=N^{1/2k^{*}}\int_{-\bar{\delta}}^{\bar{\delta}}\exp(N(f(u+\mu_{p})-f(\mu_{p})))\phi(u+\mu_{p})du\nonumber\\
	&=\!\!\!\!\int_{|w|<\bar{\delta} N^{1/2k^{*}}}\!\!\!\!\!\exp\bigg(NB\bigg(\frac{w}{N^{1/2k^{*}}};\mu_{p}\bigg)\bigg)\phi\bigg(\frac{w}{N^{1/2k^{*}}}+\mu_{p}\bigg)dw
	\end{align*}\\
	\noindent where the two equalities follow from suitable changes of variable. First we define $u=x-\mu_{p}$, and then $w=uN^{1/2k^{*}}$.
	Thus by (\ref{proprieta.B}) and the dominated convergence theorem
	\begin{align}\label{teo.1.passaggio.4}
	\lim_{N\rightarrow\infty}N^{1/2k^{*}}e^{-NF}\int_{\mu_{p}-\delta}^{\mu_{p}+\delta}&e^{Nf(x)}\phi(x)dx=\nonumber\\ &\phi(\mu_{p})\int_{\mathbb{R}}\exp\Big(\frac{\lambda_{p}}{(2k^{*})!}w^{2k^{*}}\Big)dw.
	\end{align}\\
	Since $\lambda_{p}<0$ the integral of (\ref{teo.1.passaggio.4}) is finite. Making the change of variable
	\begin{equation*}
	x=\frac{w}{(-\lambda_{p})^{1/2k^{*}}}
	\end{equation*}
	in the right-hand side of (\ref{teo.1.passaggio.4}) we obtain:
	\begin{multline}\label{teo.1.passaggio.5}
	\lim_{N\rightarrow\infty}N^{1/2k^{*}}e^{-NF}\int_{\mu_{p}-\delta}^{\mu_{p}+\delta}e^{Nf(x)}\phi(x)dx=\\\phi(\mu_{p})(-\lambda_{p})^{1/2k^{*}}\int_{\mathbb{R}}\exp\bigg(-\frac{x^{2k^{*}}}{(2k^{*})!}\bigg)dx.
	\end{multline}
	By (\ref{teo.1.passaggio.2}) and (\ref{teo.1.passaggio.5})
	\begin{multline}\label{teo.1.passaggio.6}
	\lim_{N\rightarrow\infty}N^{1/2k^{*}}e^{-NF}\int_{\mathbb{R}}e^{Nf(x)}\phi(x)dx=\\\sum_{p=1}^{P}\phi(\mu_{p})(-\lambda_{p})^{1/2k^{*}}\int_{\mathbb{R}}\exp\bigg(-\frac{x^{2k^{*}}}{(2k^{*})!}\bigg)dx.
	\end{multline}
	In a similar way for the denominator we have\\
	\begin{equation}\label{teo.1.passaggio.7}
	\lim_{N\rightarrow\infty}N^{1/2k^{*}}e^{-NF}\int_{\mathbb{R}}e^{Nf(x)}dx=\sum_{p=1}^{P}(-\lambda_{p})^{1/2k^{*}}\!\!\!\int_{\mathbb{R}}\exp\bigg(\!-\frac{x^{2k^{*}}}{(2k^{*})!}\bigg)dx.
	\end{equation}\\
	Now the statement (\ref{teo.1.passaggio.1}) follows from (\ref{teo.1.passaggio.6}) and (\ref{teo.1.passaggio.7}).
	\qed\\
	
\noindent {\it Proof of Theorem \ref{teorema.2}}.
	By lemma \ref{lemma.2} and \ref{lemma.3} (with $\gamma=1/2k$ and $m=\mu$) we know:
	\begin{equation*}
	\frac{W}{N^{1/2 -1/2k}}+\frac{S_{N}(\boldsymbol{\sigma})-N\mu}{N^{1-1/2k}} \sim\frac{\exp\bigg(Nf\bigg(\dfrac{x}{N^{1/2k}}+\mu\bigg)\bigg)dx}{\displaystyle{\int_{\mathbb{R}}}\exp\bigg(Nf\bigg(\dfrac{x}{N^{1/2k}}+\mu\bigg)\bigg)dx}
	\end{equation*}\\
	where $W\sim N(0,J^{-1})$. If $k>1$ to prove the result it suffices to verify that:\\
	\begin{equation}\label{teo.2.passaggio.1.bu}
	\frac{\displaystyle{\int_{\mathbb{R}}}\exp\bigg(Nf\bigg(\dfrac{x}{N^{1/2k}}+\mu\bigg)\bigg) \phi(x)dx}{\displaystyle{\int_{\mathbb{R}}}\exp\bigg(Nf\bigg(\dfrac{x}{N^{1/2k}}+\mu\bigg)\bigg)dx}\rightarrow\frac{\displaystyle{\int_{\mathbb{R}}}\exp\bigg(\dfrac{\lambda }{(2k)!}x^{2k}\bigg)\phi(x)dx}{\displaystyle{\int_{\mathbb{R}}}\exp\bigg(\dfrac{\lambda }{(2k)!x^{2k}}\bigg)dx}
	\end{equation}\\
	for any bounded continuous function $\phi:\mathbb{R}\rightarrow\mathbb{R}$. We pick $\delta>0$ such that it satisfies the conditions (\ref{proprieta.B}).
	By lemma \ref{lemma.3} there exists $\epsilon>0$ so that as $N\rightarrow\infty$
	\begin{equation}\label{teo.2.passaggio.2.bu}
	e^{-NF}\int_{|x|\geq\delta N^{1/2k}}\exp\bigg(Nf\bigg(\dfrac{x}{N^{1/2k}}+\mu\bigg)\bigg) \phi(x)dx=O(N^{1/2k}e^{-N\epsilon})
	\end{equation} 
	where $F=\max\{f(x)|x\in\mathbb{R}\}$. On the other hand as $|x|<\delta N^{1/2k}$ 
	\begin{multline*}
	e^{-NF}\int_{|x|<\delta N^{1/2k}}\!\!\exp\bigg(Nf\bigg(\dfrac{x}{N^{1/2k}}+\mu\bigg)\bigg)\phi(x)dx =\\
	e^{N(F-f(\mu))}\int_{|x|<\delta N^{1/2k}}\!\!\exp\bigg(NB\bigg(\dfrac{x}{N^{1/2k}};\mu\bigg)\bigg)\phi(x)dx.
	\end{multline*}
	By (\ref{proprieta.B}) and the dominated convergence theorem
	\begin{multline}\label{teo.2.passaggio.3.bu}
	\lim_{N\rightarrow\infty}e^{-NF}\int_{|x|<\delta N^{1/2k}}\exp\bigg(Nf\bigg(\dfrac{x}{N^{1/2k}}+\mu\bigg)\bigg)\phi(x)dx 
	=\\\int_{\mathbb{R}}\exp\bigg(\frac{\lambda}{(2k)!}x^{2k}\bigg)\phi(x)dx
	\end{multline}
	where the integral of the right-hand side of (\ref{teo.2.passaggio.3.bu}) is finite because $\lambda<0$.
	By (\ref{teo.2.passaggio.2.bu}) and (\ref{teo.2.passaggio.3.bu}) the statement  (\ref{teo.2.passaggio.1.bu}) follows for $k>1$.\\

	For $k=1$ in analogous way we prove that for any bounded continuous function $\phi:\mathbb{R}\rightarrow\mathbb{R}$:
	\begin{equation*}
	\dfrac{\displaystyle{\int_{\mathbb{R}}}\exp\bigg(Nf\bigg(\dfrac{x}{\sqrt{N}}+\mu\bigg)\bigg)\phi(x)dx}{\displaystyle{\int_{\mathbb{R}}}\exp\bigg(Nf\bigg(\dfrac{x}{\sqrt{N}}+\mu\bigg)\bigg)dx}\rightarrow\dfrac{\displaystyle{\int_{\mathbb{R}}}\exp\Big(\dfrac{\lambda}{2}x^{2}\Big)\phi(x)dx}{\displaystyle{\int_{\mathbb{R}}}\exp\Big(\dfrac{\lambda}{2}x^{2}\Big)dx}.
	\end{equation*}
	The Gaussian $N(0,(-\lambda)^{-1})$ obtained is the convolution of the limiting di\-stribution of the random variables $W$ and $\bar{S}_{1}(\boldsymbol{\sigma})$. Since $W\sim N(0, J^{-1})$, the random variable $\bar{S}_{1}(\boldsymbol{\sigma})$ as $N\rightarrow\infty$ has to converge to a Gaussian whose covariance is $(-\lambda)^{-1}-J^{-1}$. 
	To complete the proof we must check that 
	\begin{equation}\label{teo.2.passaggio.4.bu}
	(-\lambda)^{-1}-J^{-1}=\dfrac{\lambda+ J}{-\lambda J}>0
	\end{equation}
	where we claim that $\lambda=f''(\mu)$. Considering the function $\Phi_{\rho}$ defined in (\ref{funzione.fi}), by (\ref{funzione.fi.2}) we have $\lambda+J= J\Phi''(\mu)$. Since $\Phi_{\rho}''(\mu)>0$ and $\lambda<0$ the inequality (\ref{teo.2.passaggio.4.bu}) holds.
	\qed\\

To prove Theorem \ref{teorema.3} it is useful to consider the Legendre transformation of the function $\Phi_{\rho}$ defined in (\ref{funzione.fi})
\begin{equation}\label{legendre}
\Phi_{\rho}^{*}(y)=\sup_{x\in\mathbb{R}}\{xy-\Phi_{\rho}(x)\}.
\end{equation}
We claim that it is possible to define the function $\Phi_{\rho}^{*}$ because $\Phi_{\rho}''(x)>0$ for all $x\in\mathbb{R}$.\newpage We can state the following
\begin{lemma}\label{lemma.legendre}
Let $\Phi_{\rho}^{*}$ the function defined in (\ref{legendre}) and $\mu$ a maximum point of the function $f$ given by (\ref{funzionale.pressione}). Then 
\begin{enumerate}
\item There exists an open (possibly unbounded) interval $I$ containing $\mu$ such that $\Phi_{\rho}^{*}$ is finite, real analytic and convex (with $(\Phi_{\rho}^{*})''(x)>0$) on $I$ and $\Phi_{\rho}^{*}=+\infty$ on $\bar{I}^{C}$. 
\item Consider the random variable $U_{N}(\boldsymbol{\sigma})=m_{N}(\boldsymbol{\sigma})-\mu$. Denote by $\nu_{U}$ its distribution on $(\mathbb{R}^{N},\prod_{i=1}^{N}\rho_{\mu}(\sigma_{i}))$ where $\rho_{\mu}$ is given by (\ref{misura.ro.x}) with $x=\mu$. For any $u>0$
	\begin{equation}\label{tesi.legendre}
	P\{U_{N}(\boldsymbol{\sigma})>u\}\leq\exp\Big(-N J(\Phi_{\rho}^{*}(\mu+u)-\Phi_{\rho}^{*}(\mu)-(\Phi_{\rho}^{*})'(\mu)u)\Big).
	\end{equation}
	\item There exists a number $u_{0}>0$ such that for all $u\in(0,u_{0})$
	\begin{equation}\label{tesi.legendre.2}
	(\Phi_{\rho}^{*})'(\mu+u)-(\Phi_{\rho}^{*})'(\mu)=u+\xi(u)\quad\quad\xi(u)>0
	\end{equation}
	\end{enumerate}
\end{lemma}
\begin{proof}
	Since $\Phi_{\rho}''(x)>0$ for all $x\in\mathbb{R}$ the function $\Phi_{\rho}'$ is strictly increasing and hence admits inverse $(\Phi_{\rho}')^{-1}$. By (\ref{legendre}) the function $\Phi_{\rho}^{*}$ is bounded if and only if there exists a point $x_{0}\in\mathbb{R}$ such that $y=\Phi_{\rho}'(x_{0})$. This condition is verified when $y$ belongs to the image of the function $\Phi_{\rho}'$. In this case we have
	\begin{equation}\label{legendre.passaggio.1}
	\Phi_{\rho}^{*}(y)=yx_{0}-\Phi_{\rho}(x_{0})\quad
	(\Phi_{\rho}^{*})'(y)=(\Phi_{\rho}')^{-1}(y)\quad
	(\Phi_{\rho}^{*})''(y)=\frac{1}{\Phi_{\rho}''(x_{0})}
	\end{equation}
	Thus $\Phi_{\rho}^{*}$ is real analytic and convex in particular with $(\Phi_{\rho}^{*})''(y)>0$. By (\ref{funzione.fi.2}) and (\ref{condizione.estremale}) we have $\Phi_{\rho}'(\mu)=\mu$, hence $\mu$ is inside the image of $\Phi_{\rho}'$. On the other hand for $y$ in the complement of the closure of the image of $\Phi_{\rho}'$ we have $\Phi_{\rho}^{*}(y)=+\infty$. This shows that the first sentence of lemma \ref{lemma.legendre} is proved taken $I$ equal to the image of $\Phi_{\rho}'$.\\
	
	Let $\nu$ be any measure on $\mathscr{B}$.  %$E_{w,N}=\{(x_{1},\dots,x_{n})|\sum_{i=1}^{N}x_{i}>wN\}$. 
	Choose $Jy+h>0$, by the exponential Chebyshev's inequality we can write:
	\begin{align*}
	%\int_{E_{w,N}}\prod_{i=1}^{N}d\nu(x_{i}) 
	P\bigg\{\sum_{i=1}^{N}&x_{i}>Nw\bigg\}\nonumber\\ &=P\bigg\{\exp\bigg(\sum_{i=1}^{N}x_{i}\beta(Jy+h)\bigg)>\exp\Big(Nw\beta(Jy+h)\Big)\bigg\}\nonumber\\
	&\leq E\bigg[x_{i}(Jy+h)\bigg]^{N}\exp(-Nw(Jy+h))\nonumber\\
	&=\exp(-N w(Jy+h))\prod_{i=1}^{N}\int_{\mathbb{R}}\exp( x_{i}(Jy+h))d\nu(x_{i})\nonumber\\
	&\leq\exp\Big(-N hw-N J\Big(wy-\frac{1}{ J}\int_{\mathbb{R}}\exp( x_{i}(Jy+h))d\nu(x_{i})\Big)\Big)\nonumber\\
	&\leq\exp\Big(-N hw-N J\sup\Big\{wy-\Phi_{\nu}(y)|(Jy+h)>0\Big\}\Big)
	\end{align*}
	where $\Phi_{\nu}$ is given by (\ref{funzione.fi}) with $\rho=\nu$ and $E[\cdot]$ denotes the expectation value with respect to the measure $\rho$. By convexity of the function $\Phi_{\nu}$, whenever $w>\int_{\mathbb{R}}xd\nu(x)$ the superior value of $\{wy-\Phi_{\nu}(y)\;|\;y\in\mathbb{R}\}$ is reached for $Jy+h>0$. This shows that:
	\begin{equation*}
	P\bigg\{\sum_{i=1}^{N}x_{i}>Nw\bigg\}\leq\exp\Big(-N hw-N J\Phi_{\nu}^{*}(w)\Big) \quad\text{whenever}\; w>\int_{\mathbb{R}}xd\nu(x).
	\end{equation*}
	Since $\mu$ is a maximum point of the function $f$ by the condition (\ref{condizione.estremale}) and the definition of the measure $\rho_{\mu}$ (\ref{misura.ro.x}) with $x=\mu$  
	\begin{equation*}
	\int_{\mathbb{R}}xd\rho_{\mu}(x)=\dfrac{\int_{\mathbb{R}}x\exp( x(J\mu+h))d\rho(x)}{\int_{\mathbb{R}}\exp( x(J\mu+h))d\rho_{\mu}(x)}=\mu<\mu+u.
	\end{equation*}
	Thus
	\begin{align*}
	P\{U_{N}(\boldsymbol{\sigma})>u\}&=P\{S_{N}(\boldsymbol{\sigma})>N(\mu+u)\}\nonumber\\
	&\leq\exp(-N h(\mu+u)-N J\Phi_{\rho_{\mu}}^{*}(\mu+u))
	\end{align*}
	where
	\begin{align*}
	\Phi^{*}_{\rho_{\mu}}((\mu+u)) &=\sup_{y\in{\mathbb{R}}}\Big\{(\mu+u)y-\frac{1}{ J}\ln\int_{\mathbb{R}}\exp( s(Jy+h))d\rho_{\mu}(s)\Big\}\nonumber\\\nonumber\\
	&=\sup_{y\in{\mathbb{R}}}\Big\{(\mu+u)y-\frac{1}{ J}\ln\!\!\int_{\mathbb{R}}\!\!\exp\Big( s\Big(J\Big(y+\mu+\frac{h}{J}\Big)+h\Big)\Big)d\rho(s)
	\nonumber\\
	&\quad+\frac{1}{ J}\ln\int_{\mathbb{R}}\exp( s(J\mu+h))d\rho(s)\Big\}\nonumber\\\nonumber\\
	&=-\mu^{2}-\mu u-\frac{h}{J}(\mu+u)+\Phi_{\rho}(\mu)\nonumber\\
	&\quad+\sup_{y\in{\mathbb{R}}}\Big\{(\mu+u)\Big(y+\mu+\frac{h}{J}\Big)-\Phi_{\rho}\Big(y+\mu+\frac{h}{J}\Big)\Big\}\nonumber\\\nonumber\\
	&=\Phi_{\rho}^{*}(\mu+u)-\mu^{2}-\mu u-\frac{h}{J}(\mu+u)+\Phi_{\rho}(\mu).
	\end{align*}\\
	Since $(\Phi_{\rho}')^{-1}(\mu)=\mu$, by (\ref{legendre.passaggio.1}) we have 
	\begin{equation*}
	\Phi_{\rho}^{*}(\mu)=\mu^{2}-\Phi_{\rho}(\mu)\quad
	(\Phi_{\rho}^{*})'(\mu)=\mu.
	\end{equation*}
	Thus
	\begin{align}
	P\{U_{N}(\boldsymbol{\sigma})>u\}&\leq\exp\Big(-N h(\mu+u)-N J(\Phi_{\rho}^{*}(\mu+u)-\Phi_{\rho}^{*}(\mu)\nonumber\\
	&\quad-(\Phi_{\rho}^{*})'(\mu)u-\frac{h}{J}\Big(\mu+u\Big)\Big)\nonumber\\
	&=\exp(-N J(\Phi_{\rho}^{*}(\mu+u)-\Phi_{\rho}^{*}(\mu)-(\Phi_{\rho}^{*})'(\mu)u)).
	\end{align}
	This proves the statement (\ref{tesi.legendre}).\\

	Since $\mu$ is a maximum point of $f$ there exists $u_{0}>0$ such that $x>\Phi_{\rho}'(x)$ as $x\in(\mu,\mu+u_{0})$.
	Thus,$(\Phi_{\rho}^{*})'(\mu+u)>\mu+u$ is true for any $u\in(0,u_{0})$. Since $(\Phi_{\rho}^{*})'(\mu)=\mu$ the sentence (\ref{tesi.legendre.2}) is proved.
\end{proof}

\begin{lemma}[Transfer Principle]\label{trasferimento}
	Let $\nu_{U}$ be the distribution of the random variable $U_{N}(\boldsymbol{\sigma})=m_{N}(\boldsymbol{\sigma})-\mu$ on $(\mathbb{R}^{N},\prod_{i=1}^{N}d\rho_{\mu}(\sigma_{i}))$. There exists $\widehat{B}>0$  only depending on $\rho$ such that for each $B\in(0,\widehat{B})$ and for each $a\in(0,B/2)$ and each $r\in\mathbb{R}$, there exists $\bar{\delta}=\bar{\delta}(a,B)>0$ such that as $N\rightarrow\infty$:
	\begin{multline*}
	\int_{\mathbb{R}}\!\!\exp\bigg(irN^{\gamma}w-\frac{N J}{2}w^{2}\bigg)\!\!\int_{|u|\leq a}\!\!\!\!\!\exp(N Juw)d\nu_{U}(u)dw\\
        =\int_{|w|\leq B}\!\!\!\!\!\exp\bigg(irN^{\gamma}w-\frac{N J}{2}w^{2}\bigg)\!\!\int_{\mathbb{R}}\!\!\exp(N Juw)d\nu_{U}(u)dw+O(e^{-N\bar{\delta}}).
	\end{multline*}
\end{lemma}\vspace{0.2cm}

\begin{proof}
	We shall find $\widehat{B}>0$ such that for each $B\in(0,\widehat{B})$ and each $a\in(0,B/2)$, there exists $\bar{\delta}=\bar{\delta}(a,B)>0$ such that as $N\rightarrow\infty$\\
	\begin{equation}\label{trasferimento.1}
	\int_{|w|>B}\exp\bigg(-\frac{N J}{2}w^{2}\bigg)\int_{|u|\leq a}\exp(N Juw)d\nu_{U}(u)dw=O(e^{-N\bar{\delta}})
	\end{equation}
	and
	\begin{equation}\label{trasferimento.2}
	\int_{|w|\leq B}\exp\bigg(-\frac{N J}{2}w^{2}\bigg)\int_{|u|>a}\exp(N Juw)d\nu_{U}(u)dw=O(e^{-N\bar{\delta}}).
	\end{equation}\\
	We start with equality (\ref{trasferimento.1}). For any $B>0$ and any $a\in(0,B/2)$ we have
	\begin{align}\label{trasferimento.passaggio.1}
	\int_{|w|>B}\!\!\!\!\exp\Big(-\frac{N J}{2}w^{2}\Big)\int_{|u|\leq a}\!\!\!\!\exp&(N Juw)d\nu_{U}(u)dw\nonumber\\
	&\leq 2\int_{B}^{+\infty}\!\!\!\exp\bigg(\!-\!N J\Big(\frac{w^{2}}{2}-aw\Big)\bigg)dw\nonumber\\
	&\leq 2 \int_{B}^{+\infty}\!\!\!\exp\bigg(\!-\!N Jw\Big(\frac{B}{2}-a\Big)\bigg)dw.
	\end{align}
	As $N\rightarrow\infty$ the latter integral in (\ref{trasferimento.passaggio.1}) is $O(e^{-N\bar{\delta}_{1}})$ whit $\bar{\delta}_{1}=B(B/2-a)$, thus the equality (\ref{trasferimento.1}) is proved.

	In the proof of identity (\ref{trasferimento.2}) we exploit the following result 
	\begin{equation}\label{integrazione.parti}
	E[Y1_{\{a\leq Y\leq b\}}]\leq aP(Y\geq a)\int_{a}^{b}P(Y\geq t)dt
	\end{equation}
	where $Y$ is a random variable whose distribution is given by $\rho_{Y}$, $E[\cdot]$ denotes the expectation value with respect to the distribution $\rho_{Y}$ and $1_{\{a\leq Y\leq b\}}$ is the indicator function of the set ${\{a\leq Y\leq b\}}$.
	The inequality (\ref{integrazione.parti}) is obtained integrating by parts the left-hand side of the following:
	\begin{equation*}
	\int_{a}^{b}P(Y\geq t)dt=bP(Y\geq b)-aP(Y\geq a)-\int_{a}^{b}tP(Y\geq t)'dt
	\end{equation*}
	and observing that $P(Y\geq t)'=-\rho_{Y}(t)$.
	
	The left-hand side of equality (\ref{trasferimento.2}) is upper bounded by
	\begin{equation}\label{trasferimento.passaggio.2}
	2B\sup_{|w|\leq B}\;\int_{|u|>a}\exp\bigg(-N J\Big(\frac{w^{2}}{2}-uw\Big)\bigg)d\nu_{U}(u).
	\end{equation}
	The integral in (\ref{trasferimento.passaggio.2}) breaks up into one over $(a,+\infty)$ and another over $(-\infty,a)$. For the first using (\ref{integrazione.parti}), we obtain
	\begin{align}\label{trasferimento.passaggio.3}
	\sup_{|w|\leq B}&\int_{a}^{+\infty}\exp\bigg(-N J\Big(\frac{w^{2}}{2}-uw\Big)\bigg)d\nu_{U}(u)\nonumber\\&\leq\sup_{|w|\leq B}\exp\bigg(-N J\Big(\frac{w^{2}}{2}-wa\Big)\bigg)P\{U_{N}(\boldsymbol{\sigma})>a\}\nonumber\\
	&\quad + JNB\sup_{|w|\leq B}\int_{a}^{+\infty}\!\!\!\!\exp\bigg(\!-N J\Big(\frac{w^{2}}{2}-uw\Big)\bigg)P\{U_{N}(\boldsymbol{\sigma})>u\}du.
	\end{align}
	By (\ref{tesi.legendre}) we can bound $P\{U_{N}(\boldsymbol{\sigma})>u\}$, where $u\geq a$. In particular for $u\geq a$ it holds
	\begin{equation}\label{trasferimento.passaggio.4}
	\Phi_{\rho}^{*}(\mu+u)-\Phi_{\rho}^{*}(\mu)-(\Phi_{\rho}^{*})'(\mu)u\geq
	\begin{cases}
	u^{2}+\theta_{1} &\text{for}\;a\leq u\leq u_{0}\\
	u\theta_{2} &\text{for}\;u>u_{0}
	\end{cases}
	\end{equation}
	where $\theta_{1}=\int_{0}^{a}\xi(t)dt>0$ and $\theta_{2}=\xi(u_{0}/2)/2$.\\
	We consider an interval $I$ such that lemma \ref{lemma.legendre} is verified. For all $\mu+u\in\bar{I}^{c}$ the (\ref{trasferimento.passaggio.4}) holds since $\Phi_{\rho}^{*}(\mu+u)=+\infty$. For $\mu+u\in\bar{I}$, if $a\leq u\leq u_{0}$ by (\ref{tesi.legendre.2}) we have:
	\begin{align*}
	\Phi_{\rho}^{*}(\mu+u)-\Phi_{\rho}^{*}(\mu)-(\Phi_{\rho}^{*})'(\mu)u &=\int_{0}^{u}(\Phi_{\rho}^{*})'(\mu+t) -(\Phi_{\rho}^{*})'(\mu)dt\nonumber\\
	&=\int_{0}^{u}t+\xi(t)dt\nonumber\\
	&=\frac{u^{2}}{2}+\int_{0}^{u}\xi(t)dt\nonumber\\
	&\geq \frac{u^{2}}{2}+\theta_{1}.
	\end{align*}
	This prove the first line of (\ref{trasferimento.passaggio.4}). If $u>u_{0}$, for $u_{0}/2\leq t\leq u$
	\begin{equation*}
	(\Phi_{\rho}^{*})'(\mu+t) -(\Phi_{\rho}^{*})'(\mu)\geq(\Phi_{\rho}^{*})'\Big(\mu+\frac{u_{0}}{2}\Big)-(\Phi_{\rho}^{*})'(\mu)\geq\xi\Big(\frac{u_{0}}{2}\Big) 
	\end{equation*}
	thus if $u\geq u_{0}$
	\begin{align*}
	\int_{0}^{u}(\Phi_{\rho}^{*})'(\mu+t) -(\Phi_{\rho}^{*})'(\mu)dt &\geq \int_{u_{0}/2}^{u}(\Phi_{\rho}^{*})'(\mu+t) -(\Phi_{\rho}^{*})'(\mu)dt\nonumber\\
	&\geq \Big(u-\frac{u_{0}}{2}\Big)\xi\Big(\frac{u_{0}}{2}\Big)\nonumber\\
	&\geq u\theta_{2}.
	\end{align*}
	This proves the second line of (\ref{trasferimento.passaggio.4}).

	Choose $\hat{B}$ such that $0<\hat{B}<\theta_{2}$, for any $B\in(0,\hat{B})$ using (\ref{tesi.legendre}) and (\ref{trasferimento.passaggio.4}) we have
	\begin{align*}
	NB J\sup_{|w|\leq B}&\int_{a}^{+\infty}\exp\bigg(- JN\bigg(\frac{w^{2}}{2}-uw\bigg)\bigg)P\{U_{N}(\boldsymbol{\sigma})>u\}du\nonumber\\
	&\leq NB J\sup_{|w|\leq B}\int_{a}^{u_{0}}\exp\bigg(-N\bigg(\frac{ J}{2}w^{2}- 
	Juw+\frac{u^{2}}{2}+\theta_{1}\bigg)\bigg)du\nonumber\\
	&\quad+NB J\sup_{|w|\leq B}\int_{u_{0}}^{+\infty}\exp\bigg(-N\bigg(\frac{ J}{2}w^{2}- Juw+u\theta_{2}\bigg)\bigg)du.
	\end{align*}
	Since
	\begin{equation*}
	\int_{a}^{u_{0}}\!\!\exp\bigg(\!\!-N\bigg(\frac{ J}{2}w^{2}- 
	Juw+\frac{u^{2}}{2}+\theta_{1}\bigg)\bigg)du=e^{-N\theta_{1}}\!\int_{a}^{u_{0}}\!\!\!\exp\bigg(-\frac{N J}{2}\Big(u-w\Big)^{2}\bigg)du
	\end{equation*}
	and
	\begin{equation*}
	\int_{u_{0}}^{+\infty}\!\!\!\exp\bigg(\!-N\bigg(\frac{ J}{2}w^{2}- Juw+u\theta_{2}\bigg)\bigg)du=\frac{\exp\bigg(-N\Big(\dfrac{w^{2}}{2}+u_{0}(\theta_{2}-w)\Big)\bigg)}{N(\theta_{2}-w)}
	\end{equation*}\\
	we obtain
	\begin{multline*}
	NB J\sup_{|w|\leq B}\int_{a}^{+\infty}\exp\bigg(- JN\bigg(\frac{w^{2}}{2}-uw\bigg)\bigg)P\{U_{N}(\boldsymbol{\sigma})>u\}du\\=O(Ne^{-N\theta_{1}})+O(e^{-Nu_{0}(\theta_{2}-B)}).
	\end{multline*}
	Thus the last line of (\ref{trasferimento.passaggio.3}) is $O(e^{-N\bar{\delta}_{2}})$ where $\bar{\delta}_{2}=\min\{\theta_{1}/2,u_{0}(\theta_{2}-B))\}$. Concerning the term of (\ref{trasferimento.passaggio.3}) involving $P\{U_{N}(\boldsymbol{\sigma})>a\}$ we have
	\begin{multline*}
	\sup_{|w|\leq B}\exp\bigg(-N J\Big(\frac{w^{2}}{2}-wa\Big)\bigg)P\{U_{N}(\boldsymbol{\sigma})>a\}\\\leq\sup_{|w|\leq B}\exp\bigg(-N J\Big(\frac{w^{2}}{2}-wa+\frac{a^{2}}{2}+\theta_{1}\Big)\bigg)=O(e^{-N\theta_{1}}).
	\end{multline*}
	The integral over $(-\infty,a)$ is handled in the same way. Thus we have proved identity (\ref{trasferimento.1}) and (\ref{trasferimento.2}) with $\bar{\delta}=\min\{\bar{\delta}_{1},\bar{\delta}_{2}\}$.
\end{proof}
Now we can prove the theorem \ref{teorema.3}.\\
{\it Proof of Theorem \ref{teorema.3}}.
	Given $k>1$, to prove the statement we must find $A>0$ such that for each $r\in\mathbb{R}$
	and any $a\in(0,A)$ when the magnetization $m_{N}(\boldsymbol{\sigma})$ is inside $[\mu-a,\mu+a]$, the Gibbs value of the characteristic function of the random variable $\overline{S}_{k}(\boldsymbol{\sigma})$:
	\begin{equation}\label{teo.2.passaggio.1}
	\Big\langle e^{ir\overline{S}_{k}(\boldsymbol{\sigma})}\Big||m_{N}(\boldsymbol{\sigma})-\mu|\leq a\Big\rangle_{BG}=\dfrac{\displaystyle{\int_{|m_{N}(\boldsymbol{\sigma})-\mu|\leq a}}\!\!\!\!e^{ir\overline{S}_{k}(\boldsymbol{\sigma})}e^{- H_{N}(\boldsymbol{\sigma})}\prod_{i=1}^{N}d\rho(\sigma_{i})}{\displaystyle{\int_{|m_{N}(\boldsymbol{\sigma})-\mu|\leq a}}\!\!\!\!e^{- H_{N}(\boldsymbol{\sigma})}\prod_{i=1}^{N}d\rho(\sigma_{i})}
	\end{equation}
	tends as $N\rightarrow\infty$ to
	\begin{equation}\label{teo.2.isultato}
	\dfrac{\displaystyle{\int_{\mathbb{R}}}\exp(irs)\exp\Big(\dfrac{\lambda}{(2k)!}s^{2k}\Big)ds}{\displaystyle{\int_{\mathbb{R}}}\exp\Big(\dfrac{\lambda}{(2k)!}s^{2k}\Big)ds}.
	\end{equation}
	Defining
	\begin{equation*}
	\widetilde{H}_{N}(\boldsymbol{\sigma})=-\frac{J}{2}\bigg(\frac{S_{N}(\boldsymbol{\sigma})-N\mu}{\sqrt{N}}\bigg)^{2}
	\end{equation*}\\
	we can write (\ref{teo.2.passaggio.1}) as
	%\begin{equation}
	%\dfrac{\int\limits_{|m_{N}(\boldsymbol{\sigma})-\mu|\leq a}\exp(ir\overline{S}_{k}(\boldsymbol{\sigma}))\exp(- \widetilde{H}_{N}(\boldsymbol{\sigma}))\exp(\beta S_{N}(\boldsymbol{\sigma})(h+J\mu))\prod_{i=1}^{N}d\rho(\sigma_{i})}{\int\limits_{|m_{N}(\boldsymbol{\sigma})-\mu|\leq a}\exp(-\beta \widetilde{H}_{N}(\boldsymbol{\sigma}))\exp(\beta S_{N}(\boldsymbol{\sigma})(h+J\mu))\prod_{i=1}^{N}d\rho(\sigma_{i})}
	%\end{equation}
	\begin{equation*}
	\Big\langle e^{ir\overline{S}_{k}(\boldsymbol{\sigma})}\Big||m_{N}(\boldsymbol{\sigma})-\mu|\leq a\Big\rangle_{BG}=\dfrac{\displaystyle{\int_{|m_{N}(\boldsymbol{\sigma})-\mu|\leq a}}\!\!\!e^{ir\overline{S}_{k}(\boldsymbol{\sigma})}e^{- \widetilde{H}_{N}(\boldsymbol{\sigma})}\prod_{i=1}^{N}d\rho_{\mu}(\sigma_{i})}{\displaystyle{\int_{|m_{N}(\boldsymbol{\sigma})-\mu|\leq a}}\!\!\!e^{- \widetilde{H}_{N}(\boldsymbol{\sigma})}\prod_{i=1}^{N}d\rho_{\mu}(\sigma_{i})}
	\end{equation*}
	where $\rho_{\mu}$ is the measure defined by (\ref{misura.ro.x}) with $x=\mu$.\\ 
	Consider the random variable 
	\begin{equation*}
	U_{N}(\boldsymbol{\sigma})=\frac{S_{N}(\boldsymbol{\sigma})-N\mu}{N} 
	\end{equation*}
	and let $\nu_{U}$ be its distribution on $(\mathbb{R}^{N},\prod_{i=1}^{N}d\rho_{\mu}(\sigma_{i}))$. We can write:\\
	\begin{equation}\label{teo.2.passaggio.3}
	\Big\langle e^{ir\overline{S}_{k}(\boldsymbol{\sigma})}\Big||U_{N}(\boldsymbol{\sigma})|\leq a\Big\rangle_{BG}=\dfrac{\displaystyle{\int_{|u|\leq a}}\exp(irN^{\gamma}u)\exp\Big(\frac{N J}{2}u^{2}\Big)d\nu_{U}(u)}{\displaystyle{\int_{|u|\leq a}}\exp\Big(\frac{N J}{2}u^{2}\Big)d\nu_{U}(u)}.
	\end{equation}
	By identity (\ref{trasformazione.gaussiana}) with $m=u$ after the simplification of the term $\sqrt{N J/2\pi}$ the right-hand side of (\ref{teo.2.passaggio.3}) becomes
	\begin{equation}\label{teo.2.passaggio.4}
	\dfrac{\displaystyle{\int_{|u|\leq a}}\exp(irN^{\gamma}u)\displaystyle{\int_{\mathbb{R}}}\exp\Big(-\frac{N J}{2}x^{2}+N Jux\Big)d\nu_{U}(u)dx}{\displaystyle{\int_{\{|u|\leq a\}\times{\mathbb{R}}}}\exp\Big(-\frac{N J}{2}x^{2}+N Jux\Big)d\nu_{U}(u)dx}.
	\end{equation}
	Making the change of variable
	\begin{equation}\label{cambio.variabili.teo.2}
	w=x+\dfrac{ir}{ JN^{1-\gamma}}
	\end{equation}
	the (\ref{teo.2.passaggio.4}) becomes\\
	\begin{equation}\label{teo.2.passaggio.5}
	\dfrac{\exp\bigg(\dfrac{r^{2}}{2 JN^{1-2\gamma}}\bigg)\displaystyle{\int_{\mathbb{R}}}\exp\Big(irN^{\gamma}w-\frac{N J}{2}w^{2}\Big)\displaystyle{\int_{|u|\leq a}}\exp(N Juw)d\nu_{U}(u)dw}{\displaystyle{\int_{\mathbb{R}}}\exp\Big(-N\frac{ J}{2}w^{2}\Big)\displaystyle{\int_{|u|\leq a}}\exp(N Juw)d\nu_{U}(u)dw}.
	\end{equation}
	The change of variable (\ref{cambio.variabili.teo.2}) is justified by the analyticity of the integrand in (\ref{teo.2.passaggio.5}) as function of $w$ complex and the rapid decrease of this integrand to zero as $|Re(w)|\rightarrow\infty$ and $|Im(w)|\leq|r|N^{\gamma}$. Since $k>1$
	we have that 
	\begin{equation*}
	\exp\Big(\dfrac{r^{2}}{2 JN^{1-2\gamma}}\Big)\rightarrow 1\quad\text{as}\; N\rightarrow\infty
	\end{equation*}
	hence we can neglect this term for the rest of the proof.
	Using the transfer principle \ref{trasferimento} we can find $\hat{B}>0$ such that the (\ref{teo.2.passaggio.5}) can be written as 
	\begin{equation}\label{teo.2.passaggio.6}
	\dfrac{\displaystyle{\int_{|w|\leq \hat{B}}}\exp\Big(irN^{\gamma}w-\frac{N J}{2}w^{2}\Big)\displaystyle{\int_{\mathbb{R}}}\exp(N Juw)d\nu_{U}(u)dw}{\displaystyle{\int_{|w|\leq \hat{B}}}\exp\Big(-N\frac{ J}{2}w^{2}\Big)\displaystyle{\int_{\mathbb{R}}}\exp(N Juw)d\nu_{U}(u)dw}+O(e^{-N\bar{\delta}}).
	\end{equation}
	Making the change of variable $s=N^{\gamma}w$ and pick $\bar{B}=\min\{\delta,\hat{B}\}$, where $\delta$ is taken such that the conditions (\ref{proprieta.B}) are verified, we have for (\ref{teo.2.passaggio.6})
	\begin{equation*}
	\dfrac{\displaystyle{\int_{|s|\leq\bar{B}N^{\gamma}}}\exp(irs)\exp\bigg(- JN^{1-2\gamma}\frac{s^{2}}{2}\bigg)\displaystyle{\int_{\mathbb{R}}}\exp( JN^{1-\gamma}us)d\nu_{U}(u)ds}{\displaystyle{\int_{|s|\leq\bar{B}N^{\gamma}}}\exp\bigg(- JN^{1-2\gamma}\frac{s^{2}}{2}\bigg)\displaystyle{\int_{\mathbb{R}}}\exp( JN^{1-\gamma}us)d\nu_{U}(u)dt}+O(e^{-N\bar{\delta}}).
	\end{equation*}
	%$\Phi'(\mu)= J\mu$
	where:
	\begin{align*}
	\int_{\mathbb{R}}\exp( J&N^{1-\gamma}us)d\nu_{U}(u)\\
	&=\dfrac{\displaystyle{\int_{\mathbb{R}^{N}}}\exp\bigg(\frac{ J}{N^{\gamma}}s(S_{N}(\boldsymbol{\sigma})-N\mu)\bigg)\exp( S_{N}(\boldsymbol{\sigma})(J\mu+h))\prod_{i=1}^{N}d\rho(\sigma_{i})}{\displaystyle{\int_{\mathbb{R}^{N}}}\exp( S_{N}(\boldsymbol{\sigma})(J\mu+h))\prod_{i=1}^{N}d\rho(\sigma_{i})}\nonumber\\
	&=\exp\bigg(N J\bigg(\Phi\Big(\mu+\frac{s}{N^{\gamma}}\Big)-\Phi(\mu)-\mu\frac{s}{N^{\gamma}}\bigg)\bigg)
	\end{align*}
	Thus
	\begin{align*}
	\exp\bigg(- JN^{1-2\gamma}&\frac{s^{2}}{2}\bigg)\int_{\mathbb{R}}\exp( JN^{1-\gamma}us)d\nu_{U}(u)\nonumber\\
	&=\exp\bigg(N J\bigg(-\frac{s^{2}}{2N^{2\gamma}}+\Phi\Big(\mu+\frac{s}{N^{\gamma}}\Big)-\Phi(\mu)-\mu\frac{s}{N^{\gamma}}\bigg)\bigg)\nonumber\\
	&=\exp\bigg(N\bigg(f\Big(\mu+\frac{s}{N^{\gamma}}\Big)-f(\mu)\bigg)\bigg)\nonumber\\
	&=\exp\bigg(NB\bigg(\frac{s}{N^{\gamma}};\mu\bigg)\bigg).
	\end{align*}
	By conditions expressed in (\ref{proprieta.B}) and the dominated convergence theorem the statement follows.
\qed\\
\section{Example: the Curie-Weiss model}
Now we describe the Curie-Weiss model, that is a model defined by Hamiltonian (\ref{Hamiltoniana.1}) and distribution (\ref{misura.BG.1}) where $\rho$ is given by 
\begin{equation*}
\rho(x)=\frac{1}{2}\Big(\delta(x-1)+\delta(x+1)\Big).
\end{equation*}
For further arguments related to this model see \cite{ellis2005entropy}.
The definition of $\rho$ implies that the space of all configuration is $\Omega_{N}=\{1,-1\}^{N}$.
The function $f$ given by (\ref{funzionale.pressione}) becomes
\begin{equation}\label{funzione.f.curie}
f(x)=-\frac{ J}{2}x^{2}+\ln\cosh(Jx +h)
\end{equation}
whose estremality condition is given by the so called mean-field equation
\begin{equation}\label{campo.medio}
\mu=\tanh(J\mu+h).
\end{equation}
The solutions of this equation are the intersections between the hyperbolic tangent $y=\tanh(J\mu+h)$ and the line $y=\mu$. 
%The hyperbolic tangent is an odd function of $\mu$ convex for $\mu>0$. 
As $h\neq 0$, for any positive value of $ J$ the equation (\ref{campo.medio}) admits a unique solution $\mu_{h}$ different from zero that has the same sign as the field $h$. This solution is the unique maximum point of the function $f$. On the other hand as $h=0$ the number of solutions of equation (\ref{campo.medio}) depends on the slope $ J$ of the hyperbolic tangent. If $ J\leq 1$ there is a unique solution, the zero, which is the unique maximum point of the function $f$. If $ J>1$, the equation (\ref{campo.medio}) admits other two solutions $\pm\mu_{0}$. In this case the function $f$ reaches its maximum in $\pm\mu_{0}$. To determinate the type and the strength of the maximum points of $f$ as parameters $J$ and $h$ change, we compute the even derivatives of $f$ in the points until we obtain a value different from zero. We obtain
\begin{enumerate}
\item if $h\neq 0$ and $J>0$ the maximum point $\mu_{h}$ is of type $k=1$ and strength $\lambda=- J(1- J(1-\mu_{h}^{2})$;
\item if $h=0$ and $ J<1$ the maximum point $0$ is of type $k=1$ and strength $\lambda=- J(1- J)$;
\item if $h=0$ and $ J>1$ maximum points $\pm\mu_{0}$ are of type $k=1$ and strength $\lambda=- J(1- J(1-\mu_{0}^{2})$;
\item if $h=0$ and $ J=1$ the maximum point $0$ is of type $k=2$ and strength $\lambda=-2$.
\end{enumerate}
By theorem \ref{teorema.1} we get the distribution in the thermodynamic limit of the magnetization:
\begin{equation*}
 m_{N}(\boldsymbol{\sigma})\overset{\mathscr{D}}{\rightarrow}
\begin{cases}
\delta(x-\mu_{h})&h\neq 0, J>0\\
\delta(x)&h=0, J\leq 1\\
\frac{1}{2}\delta(x-\mu_{0})+\frac{1}{2}\delta(x+\mu_{0})&h=0, J>1.  
\end{cases}
\end{equation*}
Defined the susceptibility of the model as $\chi=\partial\mu/\partial h$, by the main field equation (\ref{campo.medio}) we obtain
\begin{equation*}
\chi=\frac{(1-\mu^{2})}{1-J(1-\mu^{2})}.
\end{equation*}
By theorem \ref{teorema.2} it is easy to check that in the thermodynamic limit\\
\begin{align*}
\dfrac{S_{N}(\boldsymbol{\sigma})-N\mu}{\sqrt{N}}\overset{\mathscr{D}}{\rightarrow} N(0,\chi) &\quad\;\text{as $J>0$ and $h\neq 0$} \\\\
\dfrac{S_{N}(\boldsymbol{\sigma})}{\sqrt{N}}\overset{\mathscr{D}}{\rightarrow} N(0,\chi)&\quad\;\text{as $0< J<1$ and $h=0$}\\\\
\dfrac{S_{N}(\boldsymbol{\sigma})}{N^{3/4}}\overset{\mathscr{D}}{\rightarrow} \dfrac{\exp\Big(-\dfrac{x^{4}}{12}\Big)dx}{\displaystyle{\int_{\mathbb{R}}}\exp\Big(-\dfrac{x^{4}}{12}\Big)dx}&\quad\;\text{as $J=1$ and $h=0$}.
\end{align*}\\
If $ J>1$ and $h=0$ the function $f$ admits two global maximum points $\pm\mu_{0}$. Considering the point $\mu_{0}$, by theorem \ref{teorema.3} there exists $A>0$ such that for all $a\in(0,A)$ if $m_{N}(\boldsymbol{\sigma})\in[\mu_{0}-a,\mu_{0}+a]$
$$\frac{S_{N}(\boldsymbol{\sigma})- N\mu_{0}}{\sqrt{N}}\overset{\mathscr{D}}{\rightarrow} N(0,\chi)$$ 
An analogous result holds for the point $-\mu_{0}$.\\

To complete the description of the Curie-Weiss model we analyze its phase transition. A phase transition point is any point of non-analyticity of the thermodynamic limit of the pressure occurring for real $h$ and/or real positive $J$. If $h\neq 0$ it is easy to show that there is not any phase transition.
The situation is totally different as $h=0$. In absence of the field $h$ we have:
\begin{equation*}
\lim_{N\rightarrow\infty}p_{N}(J,0)=\begin{cases}
0 & \text{when $ J\leq 1$} \\
-\frac{ J}{2}\mu_{0}^{2}+\ln\cosh( J\mu_{0}) & \text{when $ J>1$}.
\end{cases}
\end{equation*}
As $ J\rightarrow 1^{+}$ the spontaneous magnetization $\mu_{0}$ tends to zero, thus the limit of the pressure is continuous for every values of $J$. Differentiating this limit with respect to $J$ we obtain:
\begin{align*}
\frac{\partial}{\partial J}\Big(\lim_{N\rightarrow\infty}p_{N}(J,h)\Big) &=-\frac{J}{2}\mu^{2}- J\mu\frac{\partial\mu}{\partial J}+\tanh(J\mu+h)\Big(\mu+ J\frac{\partial\mu}{\partial J}\Big)\\
&=\frac{1}{2}\mu^{2}.
\end{align*}
Thus in zero magnetic field 
\begin{equation*}
\frac{\partial}{\partial J}\Big(\lim_{N\rightarrow\infty}p_{N}(J,0)\Big)=\begin{cases}
0 & \text{when $ J\leq 1$} \\
\frac{1}{2}\mu_{0}^{2} & \text{when $ J>1$}.
\end{cases}
\end{equation*}\\
Also this function is continuous in $J$. If we differentiate another time the limit of the pressure we get:
\begin{equation*}
\frac{\partial^{2}}{\partial J^{2}}\Big(\lim_{N\rightarrow\infty}p_{N}(J,h)\Big)=\mu\frac{\partial \mu}{\partial J}.
\end{equation*}\\
Since
\begin{equation}
\mu\frac{\partial \mu}{\partial J}=\frac{1}{2}\frac{\partial \mu^{2}}{\partial J}
\end{equation}\\
in zero field we have:
\begin{equation}\label{transizione.fase}
\frac{\partial^{2}}{\partial J^{2}}\Big(\lim_{N\rightarrow\infty}p_{N}(J,0)\Big)=\begin{cases}
0 & \text{when $ J\leq 1$} \\\\
\dfrac{1}{2}\dfrac{d\mu_{0}^{2}}{dJ} & \text{when $ J>1$}.
\end{cases}
\end{equation}
Just below $J=1$ the value of $\mu_{0}$ is small, thus we can expand the hyperbolic tangent of the mean field equation (\ref{campo.medio}):
%and so in order to better understand its behaviour when $ J\rightarrow 1^{-}$ 
\begin{equation}\label{espansione.tangente.iperbolica}
\mu_{0}= J\mu_{0}-\frac{( J\mu_{0})^{3}}{3}+O(\mu_{0}^{5}) \quad\quad\text{as}\;\; J\rightarrow 1^{+}.
\end{equation}
Since $\mu_{0}$ is different from zero as $ J>1$, we can divide by $ J\mu_{0}$ the equation (\ref{espansione.tangente.iperbolica}). We obtain 
\begin{equation*}
\frac{1}{ J}=1-\frac{( J\mu_{0})^{2}}{3}+O(\mu_{0}^{4}) \quad\quad\text{as}\;\;  J\rightarrow 1^{+}.
\end{equation*}
Thus: 
\begin{equation*}
\mu_{0} \sim\bigg( \frac{3}{( J)^{2}}\Big( 1-\frac{1}{ J}\Big)\bigg)^{\frac{1}{2}}\sim\bigg(3\Big( 1-\frac{1}{ J}\Big)\bigg)^{\frac{1}{2}}\quad \quad\text{as}\;\;  J\rightarrow 1^{+}
\end{equation*}
and the second line of (\ref{transizione.fase}) can be approximate in the following way:
\begin{equation}\label{approssimazione.transizione}
\frac{1}{2}\frac{d\mu_{0}^{2}}{dJ} \sim\frac{1}{2}\frac{d}{dJ}\bigg(3\Big( 1-\frac{1}{ J}\Big)\bigg)=\frac{3}{2J^{2}} \quad\quad \text{as}\;\;  J\rightarrow\frac{1}{J}^{+}.
\end{equation}\\
By (\ref{approssimazione.transizione}) it follows that the second derivative of the thermodynamic limit (\ref{transizione.fase}) is discontinuous. The model exhibits a phase transition of the second order for $h=0$ and $J=1$. We claim that for this choice of the parameters the normalize sum of spins does not converge to a Gaussian distribution in the thermodynamic limit. Thus the theorems \ref{teorema.1} and \ref{teorema.2} are potent tools to obtain information about the criticality of a phase.
\clearpage{\pagestyle{empty}\cleardoublepage}

\chapter[The Multi-Species Mean-Field Model]{The Multi-Species Mean-Field Model}
\lhead[\fancyplain{}{\bfseries\thepage}]{\fancyplain{}{\bfseries\rightmark}}
In this chapter we deal with the multi-species generalization of the Curie-Weiss model. After computing the exact solution of the thermodynamic limit of the pressure we analyze the asymptotic behaviour of the normalize random vector whose components are the sums of the spins of each species.
\section{The model}
We consider a system of $N$ particles that can be divided into $n$ subsets $P_{1},\dots , P_{n}$ with $P_{l}\cap P_{s}=\emptyset$, for $l\neq s$ and sizes $|P_{l}|=N_{l}$, where $\sum_{l=1}^{n}N_{l}=N$. Particles interact 
with each other and with an extern field according to the mean field Hamiltonian:
\begin{equation}\label{hamiltoniana.multi.1}
H_{N}(\boldsymbol{\sigma})=-\frac{1}{2N}\sum_{i,j=1}^{N}J_{ij}\sigma_{i}\sigma_{j}-\sum_{i=1}^{N}h_{i}\sigma_{i} \; .
\end{equation}
The $\sigma_{i}$ represents the spin of the particle $i$, while $J_{ij}$ is the parameter that tunes the mutual 
interaction between the particle $i$ and the particle $j$ and takes values according to the following symmetric matrix:

\begin{displaymath}
         \begin{array}{ll}
                \\
                N_1 \left\{ \begin{array}{ll||}
                                      \\
                                   \end{array}  \right.
                                        \\
                N_2 \left\{ \begin{array}{ll||}
                                        \\
                                   \end{array}  \right.
                                          \\
                                         \\
                                         \\
                      
                  N_n \left\{ \begin{array}{ll||}
                     \\
            \\
                \\
                                  \end{array}  \right.
         \end{array}
          \!\!\!\!\!\!\!\!
         \begin{array}{ll||}
                \quad
                 \overbrace{\qquad }^{\textrm{$N_1$}}\;
                 \overbrace{\qquad }^{\textrm{$N_2$}}\qquad\quad\;
                 \overbrace{\qquad\qquad\quad }^{\textrm{$N_n$}}
                  \\
                 \left(\begin{array}{c|c|cc|ccc}
                               \mathbf{ J}_{11}  &  \mathbf{ J}_{12} & &\;\dots\; & &\;\;\mathbf{ J}_{1n}\;\;&
                                \\
                                 \hline
                              \mathbf{ J}_{12} & \mathbf{ J}_{22} & & & & &\\
                             \hline
                             & & & & & &\\
                             \vdots & & & & & &\\
                             \hline
                             & & & & & &\\
                             \mathbf{ J}_{1n} & \mathbf{ J}_{2n} & &\;\dots\; & &\;\;\mathbf{ J}_{nn}\;\; &\\
                             & & & & & &
                      \end{array}\right)
               \end{array}
\end{displaymath}\\
\noindent where each block $\mathbf{J}_{ls}$ has constant elements $J_{ls}$. For $l=s$, $\mathbf{J}_{ll}$ is a square matrix, whereas the matrix $\mathbf{ J}_{ls}$ is rectangular. We assume $J_{11}, J_{22},\dots , J_{nn}$ to be positive, whereas $J_{ls}$ with $l\neq s$ can be either positive or negative allowing both ferromagnetic and antiferromagnetic interactions. The vector field takes also different values depending on the subset the particles belong to as specified by: 

\begin{displaymath}
         \begin{array}{ll}
                N_1 \left\{ \begin{array}{ll}
                                      \\
                                   \end{array}  \right.
                                        \\
                N_2 \left\{ \begin{array}{ll}
                                        \\
                                   \end{array}  \right.
                                          \\
                                         \\
                                         \\
                      
                  N_n \left\{ \begin{array}{ll}
                     \\
            \\
                \\
                                  \end{array}  \right.

           \!\!\!\!\!\!
    \end{array}
    \!\!\!\!\!\!
    \left(\begin{array}{ccc|c}
                \mathbf{h}_{1}
            \\
            \hline
            
            \mathbf{h}_{2}
            \\
            \hline
            \\
            \vdots
            \\
            \hline
            \\
            \mathbf{h}_{n}
            \\
            \\
        \end{array}\right)
\end{displaymath}\\
\noindent where each $\mathbf{h}_{l}$ is a vector of costant elements $h_{l}$.

The joint distribution of a spin configuration $\boldsymbol{\sigma}=(\sigma_{1},\dots ,\sigma_{N})$ is given by the Boltzmann-Gibbs measure:
\begin{equation}\label{misura.BG.multi}
P_{N,\mathbf{J},\mathbf{h}}\{\boldsymbol{\sigma}\}=\frac{\exp(-H_{N}(\boldsymbol{\sigma}))}{Z_{N}(\mathbf{J},\mathbf{h})}\prod\limits_{i=1}^{N}d\rho(\sigma_{i})
\end{equation}
where $Z_{N}(\mathbf{J},\mathbf{h})$ is the partition function
\begin{equation}\label{partizione.multi}
Z_{N}(\mathbf{J},\mathbf{h})=\int_{\mathbb{R}^{N}}\exp(-H_{N}(\boldsymbol{\sigma}))\prod\limits_{i=1}^{N}d\rho(\sigma_{i})\end{equation}
and $\rho$ is the measure:
\begin{equation}\label{misura.ro.multi}
\rho(x)=\frac{1}{2}\Big(\delta(x-1)+\delta(x+1)\Big)
\end{equation}
where $\delta(x-x_{0})$ with $x_{0}\in\mathbb{R}$ denotes the unit point mass with support at $x_{0}$. The definition of $\rho$ implies that each spin can take only the values $\pm 1$. The inverse temperature $\beta$ isn't explicitly written because it is included in the parameters of the Hamiltonian.
%Indicating by $S_{N_{l}}(\sigma)$ 
%\begin{equation}
%S_{N_{l}}(\sigma)=\sum_{i \in P_{l}}\sigma_{i}
%\end{equation}

%\noindent the sum of the spins of the subset $P_{l}$, and by $\alpha_{l}=\frac{N_{l}}{N}$ the relative size of the subset $P_{l}$ on the whole, we may express the Hamiltonian (\ref{Hamiltoniana2pop}) as:

%\begin{equation}\label{Hami2pop}
%H(\sigma)=-\sum_{l=1}^{n}\frac{\alpha_{l}J_{l}}{2N_{l}}(S_{N_{l}})^{2}-\sum_{1\leq l<s\leq n} I_{ls}\sqrt{\frac{\alpha_{l}\alpha_{s}}{N_{l}N_{s}}}S_{N_{l}}S_{N_{s}}-\sum_{l=1}^{n}h_{l}S_{N_{l}}
%\end{equation}\\

By introducing the magnetization of a set of spins $A$ as:
\begin{equation*}
m_{A}(\boldsymbol{\sigma})=\frac{1}{|A|}\sum_{i \in A}\sigma_{i}
\end{equation*}
\noindent and indicating by $m_{l}(\boldsymbol{\sigma})$ the magnetization of the set $P_{l}$, and by $\alpha_{l}=N_{l}/N$ the relative size of the set $P_{l}$, we may easily express the Hamiltonian (\ref{hamiltoniana.multi.1}) as:
\begin{equation}\label{hamiltoniana.multi.2}
 H_{N}(\boldsymbol{\sigma})=-Ng\Big(m_{1}(\boldsymbol{\sigma}),\dots,m_{n}(\boldsymbol{\sigma})\Big)
\end{equation}
\noindent where the function $g$ is:
\begin{equation}\label{funzione.g}
g(x_{1},\dots,x_{n})=\frac{1}{2}\sum\limits_{l, s=1}^{n}\alpha_{l}\alpha_{s}J_{ls}x_{l}x_{s}+\sum\limits_{l=1}^{n}\alpha_{l}h_{l}x_{l}  .
\end{equation}
Defined
\begin{equation}\label{interazioneridotta}
\mathbf{J}=\begin{pmatrix}
J_{11}  & J_{12} & \dots & J_{1n}\\
J_{12}  & J_{22} & \dots & J_{2n}\\
\vdots&\vdots&&\vdots \\
J_{1n}  & J_{2n} & \dots & J_{nn}
\end{pmatrix}
\quad\quad\quad
\mathbf{h}=\begin{pmatrix}
h_{1}  \\
h_{2} \\
\vdots \\
h_{n}
\end{pmatrix}
\end{equation}\\
we can write the function $g$ in a compact way as
\begin{equation}\label{funzione.g.compatta}
 g(\mathbf{x})=\frac{1}{2}\langle\widetilde{\mathbf{J}}\mathbf{x},\mathbf{x}\rangle+\langle\widetilde{\mathbf{h}},\mathbf{x}\rangle
\end{equation}\\
\noindent where the matrix $\widetilde{\mathbf{J}}=\mathbf{D}_{\boldsymbol{\alpha}}\mathbf{D}_{\boldsymbol{\alpha}}\mathbf{J}\mathbf{D}_{\boldsymbol{\alpha}}\mathbf{D}_{\boldsymbol{\alpha}}$, the vector $\widetilde{\mathbf{h}}=\mathbf{D}_{\boldsymbol{\alpha}}\mathbf{D}_{\boldsymbol{\alpha}}\mathbf{h}$ and the matrix $\mathbf{D}_{\boldsymbol{\alpha}}=diag\{\sqrt{\alpha_{1}},\dots,\sqrt{\alpha_{n}}\}$.
The matrix $\mathbf{J}$ is called reduce interaction matrix.

\section{Thermodynamic limit}
The existence of the thermodynamic limit of the pressure 
\begin{equation*}
 p_{N}(\mathbf{J},\mathbf{h})=\frac{1}{N}\ln Z_{N}(\mathbf{J},\mathbf{h})
\end{equation*}
associated to the model defined by Hamiltonian (\ref{hamiltoniana.multi.1}) and distribution (\ref{misura.BG.multi}) is proved in the paper \cite{gallo2008bipartite}. 
In this section we compute the exact solution of this limit exploiting a tails estimation on the number of configurations that share the same vector of the magnetizations. In this way we obtain a lower and an upper bound for the partition function that converge to a same value as $N\rightarrow\infty$. This technique is used by Talagrand to compute the thermodynamic limit for the Curie-Weiss model \cite{talagrand2003spin}. 

Since the spins variable can take only two values we can write the partition function as
\begin{equation*}
Z_{N}(\mathbf{J},\mathbf{h})=\frac{1}{2^{N}}\sum_{\boldsymbol{\sigma}\in\Omega_{N}}\exp(-H_{N}(\boldsymbol{\sigma}))
\end{equation*}
where $\Omega_{N}=\{-1,1\}^{N}$ is the space of all possible configuration $\boldsymbol{\sigma}$.
 
Denote with $\boldsymbol{\sigma}_{l}$ the configuration of the spins of the set $P_{l}$ and define
\begin{equation}\label{cardinalita}
 A_{\mu_{l}}=\card\bigg\{\boldsymbol{\sigma}_{l}\in\Omega_{N_{l}}\Big| m_{l}(\boldsymbol{\sigma})=\mu_{l}\bigg\}
\end{equation}
using the Hamiltonian expressed as function of the magnetizations (\ref{hamiltoniana.multi.2}) where the function $g$ is given by (\ref{funzione.g.compatta}) we can write:
\begin{equation*}
Z_{N}(\mathbf{J},\mathbf{h})=\frac{1}{2^{N}}\sum_{\boldsymbol{\mu}}\prod_{l=1}^{n}A_{\mu_{l}}\exp\bigg(N\Big(\frac{1}{2}\langle\widetilde{\mathbf{J}}\boldsymbol{\mu},\boldsymbol{\mu}\rangle+\langle\widetilde{\mathbf{h}},\boldsymbol{\mu}\rangle\Big)\bigg)
\end{equation*}
where the sum extends over all the possible values of the random vector $(m_{1}(\boldsymbol{\sigma}),\dots,m_{n}(\boldsymbol{\sigma}))$. 
\begin{lemma}\label{lemma.talagrand}
	Consider the set $\Omega_{N_{l}}=\{-1,1\}^{N_{l}}$ of all possible configuration $\boldsymbol{\sigma}_{l}$. Let $A_{\mu_{l}}$ be a positive number defined by (\ref{cardinalita}). Then the following inequality holds
	\begin{equation}\label{bound}
	\frac{1}{C}\frac{2^{N_{l}}}{\sqrt{N_{l}}}\exp(-N_{l}\mathscr{I}(\mu_{l}))\leq A_{\mu_{l}}\leq 2^{N_{l}}\exp(-N_{l}\mathscr{I}(\mu_{l}))
	\end{equation}
	where $C$ is a constant and
	\begin{equation}\label{entropia}
	\mathscr{I}(x)=\frac{1}{2}\Big((1+x)\ln(1+x)+(1-x)\ln(1-x)\Big)
	\end{equation}
\end{lemma}\vspace{0.5cm}
\begin{proof}
 As $m_{l}(\boldsymbol{\sigma})=\mu_{l}$, the configuration $\boldsymbol{\sigma}_{l}$ contains $N_{l}(1+\mu_{l})/2$ times the value $1$ and $N_{l}(1-\mu_{l})/2$ times the value $-1$, thus we have:
\begin{equation*}
 A_{\mu_{l}}=
\begin{pmatrix}
 N_{l}\\
\frac{N_{l}(1+\mu_{l})}{2}
\end{pmatrix}
\end{equation*}
Using Stirling's formula, $n!\sim n^{n}e^{-n}\sqrt{2\pi n}$, we get\\
\begin{align}\label{lower.bound}
 A_{\mu_{l}}&\geq\sqrt{\frac{2}{\pi}}\dfrac{1}{\sqrt{N_{l}(1-\mu_{l}^{2})}}\dfrac{N_{l}^{N_{l}}}{\Big(\frac{N_{l}(1+\mu_{l})}{2}\Big)^{N_{l}(1+\mu_{l})/2}\Big(\frac{N_{l}(1-\mu_{l})}{2}\Big)^{N_{l}(1-\mu_{l})/2}}\nonumber\\\nonumber\\
&\geq\frac{1}{C}\frac{2^{N_{l}}}{\sqrt{N_{l}}}\dfrac{1}{(1+\mu_{l})^{N_{l}(1+\mu_{l})/2}(1-\mu_{l})^{N_{l}(1-\mu_{l})/2}}\nonumber\\\nonumber\\
&=\frac{1}{C}\frac{2^{N_{l}}}{\sqrt{N_{l}}}\exp(-N_{l}\mathscr{I}(\mu_{l}))
\end{align}
The (\ref{lower.bound}) gives a lower bound of $A_{\mu_{l}}$. To obtain an upper bound for $A_{\mu_{l}}$ we suppose the spins $\sigma_{i}$ independent. In this case all configurations $\boldsymbol{\sigma}_{l}$ have same probability, hence
\begin{equation*}
A_{\mu_{l}}=2^{N_{l}}P\bigg\{m_{l}(\boldsymbol{\sigma})=\mu_{l}\bigg\}\leq 2^{N_{l}}P\bigg\{m_{l}(\boldsymbol{\sigma})\geq \mu_{l}\bigg\}
\end{equation*}
where by definition of the magnetization
\begin{equation}\label{talagrand.passaggio.3}
 P\bigg\{m_{l}(\boldsymbol{\sigma})\geq \mu_{l}\bigg\}=P\bigg\{S_{l}(\boldsymbol{\sigma})\geq \mu_{l} N_{l}\bigg\}.
\end{equation}
Take $\lambda>0$ by Chebyshev's inequality we can bound the probability (\ref{talagrand.passaggio.3})
\begin{align}\label{talagrand.passaggio.4}
P\bigg\{S_{l}(\boldsymbol{\sigma})\geq \mu_{l} N_{l}\bigg\}&\leq e^{-\lambda\mu_{l} N_{l}}\prod_{i=1}^{N_{l}}E_{\rho}[\exp(\lambda\sigma_{i})]\nonumber\\
&=\exp(N_{l}(-\lambda\mu_{l} +\ln\cosh\lambda))\nonumber\\
&\leq\min_{\lambda}\{\exp(N_{l}(-\lambda\mu_{l}+\ln\cosh\lambda))\}
\end{align}
where $E_{\rho}[\cdot]$ denotes the expectation value with respect to the measure $\rho$ given by (\ref{misura.ro.multi}). If $|\mu_{l}|<1$, the exponent in the last line of (\ref{talagrand.passaggio.4}) is minimized for
\begin{equation}\label{talagrand.passaggio.5}
 \lambda=\tanh^{-1}(\mu_{l})=\frac{1}{2}\ln\bigg(\dfrac{1+\mu_{l}}{1-\mu_{l}}\bigg)
\end{equation}
Since $1/(\cosh^{2}y)=1-\tanh^{2}y$ the following equality holds
\begin{equation}\label{talagrand.passaggio.6}
 \ln\cosh\lambda=-\frac{1}{2}\ln(1-\mu_{l}^{2})
\end{equation}
Thus by (\ref{talagrand.passaggio.5}) and (\ref{talagrand.passaggio.6}) 
\begin{equation*}
\min_{\lambda}\{\exp(N_{l}(-\lambda\mu_{l}+\ln\cosh\lambda))\}=\exp(-N_{l}\mathscr{I}(\mu_{l})).
\end{equation*}
Hence we obtain the following upper bound for $A_{\mu_{l}}$
\begin{equation}\label{upper.bound}
A_{\mu_{l}}\leq 2^{N_{l}}\exp(-N_{l}\mathscr{I}(\mu_{l})).
\end{equation}
The statement (\ref{bound}) follows by (\ref{lower.bound}) and (\ref{upper.bound}).
\end{proof}
\vspace{0.5cm}
The lemma \ref{lemma.talagrand} allows to bound the partition function in the following way:
\begin{equation*}
\frac{1}{C}\prod_{l=1}^{n}\frac{1}{\sqrt{N_{l}}}\exp\Big(N\max_{\boldsymbol{\mu}}\bar{f}(\boldsymbol{\mu})\Big)\leq Z_{N}(\mathbf{J},\mathbf{h})\leq \prod_{l=1}^{n}(N_{l}+1)\exp\Big(N\max_{\boldsymbol{\mu}}\bar{f}(\boldsymbol{\mu})\Big)
\end{equation*}
where
\begin{equation}\label{funzione.f.bar}
 \bar{f}(x_{1},\dots,x_{n})=\frac{1}{2}\sum_{l,s=1}^{n}\alpha_{l}\alpha_{s}J_{ls}x_{l}x_{s}+\sum_{l=1}^{n}\alpha_{l}h_ {l}x_{l}-\sum_{l=1}^{n}\alpha_{l}\mathscr{I}(x_{l}).
\end{equation}
and the function $\mathscr{I}$ is defined in (\ref{entropia}).
Hence for the pressure we have:
\begin{multline*}
-\frac{1}{N}\bigg(\ln C+\frac{1}{2}\sum_{l=1}^{n}\ln N_{l}\bigg)+\max_{\boldsymbol{\mu}}\bar{f}(\boldsymbol{\mu})\leq p_{N}(\mathbf{J},\mathbf{h})\\\leq \frac{1}{N}\bigg(\sum_{l=1}^{n}\ln(N_{l}+1)\bigg)+\max_{\boldsymbol{\mu}}\bar{f}(\boldsymbol{\mu}).
\end{multline*}
Thus the limit as $N\rightarrow\infty$ of the pressure is obtained by maximazing the function $\bar{f}$ defined by (\ref{funzione.f.bar}). Differentiating $\bar{f}$ with respect to $x_{1},\dots,x_{n}$ we obtain the mean field equations of the model:
\begin{equation}\label{campomedio.multi}
\begin{cases}
x_{1} &\!\!\!\!= \tanh\Big(\sum\limits_{l=1}^{n}\;\alpha_{l}J_{1l}\;x_{l}+h_{1}\Big) \\
x_{2} &\!\!\!\!=\tanh\Big(\sum\limits_{l=1}^{n}\;\alpha_{l}J_{2l}\;x_{l}+h_{2}\Big)\\
\;\vdots\\
x_{n} &\!\!\!\!=\tanh\Big(\sum\limits_{l=1}^{n}\;\alpha_{l}J_{ln}\;x_{l}+h_{n}\Big) \; .
\end{cases} 
\end{equation}

When the reduce interaction matrix $\mathbf{J}$ is positive define we can compute the thermodynamic limit also using the Thompson method \cite{thompson1988classical}.  
Considering the Hamiltonian written as function of the magnetizations (\ref{hamiltoniana.multi.2}) where the function $g$ is given by (\ref{funzione.g.compatta}), the partition function (\ref{partizione.multi}) can be expressed in the form: 
\begin{equation*}
 Z_{N}(\mathbf{J},\mathbf{h})=\int_{\mathbb{R}^{n}}\exp\bigg(N\Big(\frac{1}{2}\langle\widetilde{\mathbf{J}}\mathbf{m},\mathbf{m}\rangle+\langle\widetilde{\mathbf{h}},\mathbf{m}\rangle\Big)\bigg)d\nu_{M}(\mathbf{m})
\end{equation*}
where $\nu_{M}$ denotes the distribution of the random vector $(m_{1}(\boldsymbol{\sigma}),\dots,m_{n}(\boldsymbol{\sigma}))$ on $(\mathbb{R}^{N},\prod_{i=1}^{N}\rho(\sigma_{i}))$. 

Since $\mathbf{J}$ is a positive define matrix the following identity holds
\begin{equation}\label{trasformazione.gaussiana.multi}
\exp\bigg(\frac{N}{2}\langle\widetilde{\mathbf{J}}\mathbf{m},\mathbf{m}\rangle\bigg)=\bigg(\dfrac{N\det\widetilde{\mathbf{J}}}{(2\pi)^{n}}\bigg)^{\frac{1}{2}}\!\!\!\int_{\mathbb{R}^{n}}\!\!\!\exp\bigg(\!\!-\frac{N}{2}\langle\widetilde{\mathbf{J}}\mathbf{x},\mathbf{x}\rangle+N\langle\widetilde{\mathbf{J}}\mathbf{x},\mathbf{m}\rangle\bigg)d\mathbf{x}.
\end{equation}
Using (\ref{trasformazione.gaussiana.multi}) the partition function becomes
\begin{align*}
Z_{N}(\mathbf{J},\mathbf{h})&=\bigg(\frac{N\det\widetilde{\mathbf{J}}}{(2\pi)^{n}}\bigg)^{\frac{1}{2}}\\&\quad\times\iint_{\mathbb{R}^{2n}}\exp\bigg(N\Big(-\frac{1}{2}\langle\widetilde{\mathbf{J}}\mathbf{x},\mathbf{x}\rangle+\langle\widetilde{\mathbf{h}},\mathbf{m}\rangle+\langle\widetilde{\mathbf{J}}\mathbf{x},\mathbf{m}\rangle\Big)\bigg)d\nu_{M}(\mathbf{m})d\mathbf{x}\\
&=\bigg(\frac{N\det\widetilde{\mathbf{J}}}{(2\pi)^{n}}\bigg)^{\frac{1}{2}}\\
&\quad\times\!\!\!\int_{\mathbb{R}^{n}}\!\!\exp\Big(N\Big(\!\!-\frac{1}{2}\langle\widetilde{\mathbf{J}}\mathbf{x},\mathbf{x}\rangle\Big)\Big)\!\!\int_{\mathbb{R}^{n}}\!\!\!\exp\Big(\!N\Big(\langle\widetilde{\mathbf{J}}\mathbf{x}+\widetilde{\mathbf{J}},\mathbf{m}\rangle\Big)\Big)d\nu_{M}(\mathbf{m})d\mathbf{x}.
\end{align*}\\
Since
\begin{align*}
 \int_{\mathbb{R}^{n}}\!\!\!\exp\bigg(N\Big(&\langle\widetilde{\mathbf{J}}\mathbf{x}+\widetilde{\mathbf{J}},\mathbf{m}\rangle\Big)\bigg)d\nu_{M}(\mathbf{m})\\&=\int_{\mathbb{R}^{N}}
\!\!\!\exp\bigg(\sum_{l=1}^{n}\alpha_{l}\sum_{i\in P_{l}}\sigma_{i}\Big(\sum_{s=1}^{n}\alpha_{s}J_{ls}x_{s}+h_{l}\Big)\bigg)\prod_{i\in P_{l}}d\rho(\sigma_{i})\nonumber\\
&=\prod_{l=1}^{n}\prod_{i\in P_{l}}\int_{\mathbb{R}}\exp\bigg(\alpha_{l}\sigma_{i}\Big(\sum_{s=1}^{n}\alpha_{s}J_{ls}x_{s}+h_{l}\Big)\bigg)d\rho(\sigma_{i})
\end{align*}\\
summing over the spins we obtain
\begin{equation*}
Z_{N}(\mathbf{J},\mathbf{h})=\bigg(\frac{N\det\widetilde{\mathbf{J}}}{(2\pi)^{n}}\bigg)^{\frac{1}{2}}\int_{\mathbb{R}^{n}}\exp(Nf(\mathbf{x}))d\mathbf{x}
\end{equation*}
where
\begin{equation}\label{funzionale.pressione.multi}
 f(x_{1},\dots,x_{n})=-\frac{1}{2}\sum_{l, s=1}^{n}\alpha_{l}\alpha_{s}J_{ls}x_{l}x_{s}+\!\sum_{l=1}^{n}\alpha_{l}\ln\bigg(\!\cosh\bigg(\sum_{s=1}^{n}\alpha_{s}J_{ls}x_{s}+h_{l}\bigg)\bigg).
\end{equation}

We can state the following:
\begin{prop}\label{prop.multi}
	Let $f$ be the function defined in (\ref{funzionale.pressione.multi}) associated to a model defined by the Hamiltonian (\ref{hamiltoniana.multi.2}). If the reduced interaction matrix $\mathbf{J}$ is positive define, then
	\begin{enumerate}
	\item $f$  has a finite number (different from zero) of global maximum points.
	\item for any positive $N\in\mathbb{R}$
	\begin{equation}\label{proprieta.f.multi.1}
	\int_{\mathbb{R}^{n}}\exp(Nf(\mathbf{x}))d\mathbf{x}<\infty
	\end{equation}
	\item if $\boldsymbol{\mu}$ is a global maximum point of $f$
	\begin{equation}\label{proprieta.f.multi.2}
	\lim_{N\rightarrow\infty}\frac{1}{N}\ln\int_{\mathbb{R}^{n}}\exp(Nf(\mathbf{x}))d\mathbf{x}=f(\boldsymbol{\mu}).
	\end{equation}
	\end{enumerate}
\end{prop}
\vspace{0.3cm}
\begin{proof}
	Since $\cosh y\leq e^{|y|}$ for all $y\in\mathbb{R}$ the following inequality holds:\\
	\begin{equation}\label{convessitaG}
	f(x_{1},\dots,x_{n}) \leq -\frac{1}{2}\sum_{l, s=1}^{n}\alpha_{l}\alpha_{s}J_{ls}x_{l}x_{s}+\sum_{l=1}^{n}\alpha_{l}\bigg|\sum_{s=1}^{n}\alpha_{s}J_{ls}\;x_{s}+h_{l}\bigg|\; .
	\end{equation}\\
	The function on the right-hand side of (\ref{convessitaG}) goes to $-\infty$ as $|\mathbf{x}|\rightarrow\infty$. Thus by (\ref{convessitaG}) the function $f$ have to show the same limiting behaviour. This property together with the analyticity assures that the function $f$ admits a finite number (different from zero) of global maxima points.

	We prove the second statement by induction. As $N=1$, since $\alpha_{l}<1$ for each $l=1,\dots,n$, we have:\\
	\begin{align*}
	\int_{\mathbb{R}^{n}}&e^{f(\mathbf{x})}d\mathbf{x}\nonumber\\
	&\leq\int_{\mathbb{R}^{n}}\exp\Big(\!-\frac{1}{2}\sum_{l, s=1}^{n}\alpha_{l}\alpha_{s}J_{ls}x_{l}	x_{s}\Big)
	\prod_{l=1}^{n}\cosh\Big(\sum_{s=1}^{n}\alpha_{s}J_{ls}x_{s}+h_{l}\Big)dx_{1}\dots dx_{n}\nonumber\\
	&\leq\exp\Big(n\max_{i}\{h_{i}\}+\max_{\mathbf{b}_{\boldsymbol{\sigma}}}\Big\{\frac{1}{2}\langle\widetilde{\mathbf{J}}^{-1}\mathbf{b}_{\boldsymbol{\sigma}},\mathbf{b}_{\boldsymbol{\sigma}}\rangle\Big\}\Big)
	\end{align*}\\
	\noindent where $\mathbf{b}_{\boldsymbol{\sigma}}=(\alpha_{1}\sum_{i=1}^{n}\sigma_{i}J_{1i},\dots,\alpha_{n}\sum_{i=1}^{n}\sigma_{i}J_{in})$. This proves the statement for $N=1$. Defined $F=\max \{f(\mathbf{x})|\mathbf{x}\in\mathbb{R}^{n}\}$ and supposed true the inductive hypothesis:
	\begin{equation*}
	\int_{\mathbb{R}^{n}}e^{(N-1)f(\mathbf{x})}d\mathbf{x}<\infty
	\end{equation*}
	we have:
	\begin{equation*}
	%\begin{split}
	\int_{\mathbb{R}^{n}}e^{Nf(\mathbf{x})}d\mathbf{x}=\int_{\mathbb{R}^{n}}e^{(N-1)f(\mathbf{x})}e^{f(\mathbf{x})}d\mathbf{x}
	\leq e^{F}\int_{\mathbb{R}^{n}}e^{(N-1)f(\mathbf{x})}d\mathbf{x}.
	%\end{split}
	\end{equation*}
	This proves the statement (\ref{proprieta.f.multi.1}).\vspace{0.3cm}

	As in the one-species case, to prove the statement (\ref{proprieta.f.multi.1}) we write\\
	\begin{equation*}
	\int_{\mathbb{R}^{n}}\exp(Nf(\mathbf{x}))d\mathbf{x}=e^{Nf(\boldsymbol{\mu})}I_{N}  
	\end{equation*}
	where 
	\begin{equation*}
	I_{N}=\int_{\mathbb{R}^{n}}\exp(N(f(\mathbf{x})-f(\boldsymbol{\mu})))d\mathbf{x}.
	\end{equation*}\\
	Since $f(\mathbf{x})-f(\boldsymbol{\mu})\leq 0$, the integral $I_{N}$ is a decreasing function of $N$.\\ 
	Thus
	\begin{equation*}
	\ln\int_{\mathbb{R}^{n}}\exp(Nf(\mathbf{x}))d\mathbf{x}\leq Nf(\boldsymbol{\mu})+\ln I_{1}.
	\end{equation*}\\
	Hence we obtain 
	\begin{equation}\label{prop.multi.passaggio.6}
	\lim_{N\rightarrow\infty}\frac{1}{N}\ln\int_{\mathbb{R}^{n}}\exp(Nf(\mathbf{x}))d\mathbf{x}\leq f(\boldsymbol{\mu}).
	\end{equation}
	By continuity of $f$, given any $\epsilon>0$, there exists $\delta_{\epsilon}>0$ such that as $\mathbf{x}$ is inside the ball $B(\boldsymbol{\mu},\delta_{\epsilon})$ we have $f(\mathbf{x})-f(\boldsymbol{\mu})>-\epsilon$. We can write:\\
	\begin{equation*}
	I_{N}\geq\int_{B(\boldsymbol{\mu},\delta_{\epsilon})}\exp(N(f(\mathbf{x})-f(\boldsymbol{\mu})))d\mathbf{x}>\frac{\pi^{\frac{n}{2}}}{\Gamma(\frac{n}{2}+1)}\delta_{\epsilon}^{n}e^{-N\epsilon}
	\end{equation*}\\
	where $\Gamma(x)$ is the function Gamma of Euler. Thus
	\begin{equation}\label{prop.multi.passaggio.7}
	\lim_{N\rightarrow\infty}\frac{1}{N}\ln\int_{\mathbb{R}^{n}}\exp(Nf(\mathbf{x}))d\mathbf{x}\geq f(\boldsymbol{\mu})-\epsilon.
	\end{equation}
	Since $\epsilon$ is arbitrary the statement (\ref{proprieta.f.multi.2}) follows from the inequalities (\ref{prop.multi.passaggio.6}) and (\ref{prop.multi.passaggio.7}).
\end{proof}
\vspace{0.5cm}
\noindent The proposition \ref{prop.multi} infers that in the thermodynamic limit
\begin{align*}
\lim_{N\rightarrow\infty}p_{N}(\mathbf{J},\mathbf{h})&=\lim_{N\rightarrow\infty}\bigg(\frac{1}{N}\ln\bigg(\frac{N\det\widetilde{\mathbf{J}}}{(2\pi)^{n}}\bigg)+\frac{1}{N}\ln\int_{\mathbb{R}}\exp(Nf(\mathbf{x}))d\mathbf{x}\bigg)\nonumber\\
&=\max_{\mathbf{x}}f(\mathbf{x}).
\end{align*}
The extremal conditions of $f$ give again the same mean-field equations (\ref{campomedio.multi}).

\section{Asymptotic behaviour of the sums of spins}
In this section we generalize the results obtained by Ellis, Newman an Rosen in \cite{ellis1978limit} and \cite{ellis1980limit}. In particular denoted by
\begin{equation*}
 S_{l}(\boldsymbol{\sigma})=\sum_{i\in P_{l}}\sigma_{i}
\end{equation*}
the sum of the spins belong to the set $P_{l}$ we determinate the suitable normalization of the random vector $(S_{1}(\boldsymbol{\sigma}),\dots,S_{n}(\boldsymbol{\sigma}))$ such that in the thermodynamic limit it converges to a well define $n$-dimensional random variable.

Before to go on we introduce some notations.
Considering $\mathbf{x},\mathbf{y}\in\mathbb{R}^{n}$ and $\gamma\in\mathbb{R}$ we define
\begin{itemize}
\item $\mathbf{x}^{\gamma}=(x_{1}^{\gamma},\dots,x_{n}^{\gamma})$;
\item $\mathbf{x}\mathbf{y}=(x_{1}y_{1},\dots,x_{n}y_{n})$;
\item $\dfrac{\mathbf{x}}{\mathbf{y}}=\Big(\dfrac{x_{1}}{y_{1}},\dots,\dfrac{x_{n}}{y_{n}}\Big)$ where $y_{l}\neq 0$ for $l=1,\dots,n$.
\end{itemize}

We shall see that the behaviour of the limiting distribution of the random vector of sums of spins depends crucially on the number and the type of the maximum points of the function $f$ exploited in the Thompson method to compute the thermodynamic limit. We recall it:
\begin{equation}\label{funzionale.pressione.multi.replica}
 f(x_{1},\dots,x_{n})=-\frac{1}{2}\sum_{l, s=1}^{n}\alpha_{l}\alpha_{s}J_{ls}x_{l}x_{s}+\!\sum_{l=1}^{n}\alpha_{l}\ln\bigg(\!\cosh\bigg(\sum_{s=1}^{n}\alpha_{s}J_{ls}x_{s}+h_{l}\bigg)\bigg).
\end{equation}

Let $\boldsymbol{\mu}^{1},\dots,\boldsymbol{\mu}^{P}$ be global maxima points of the function $f$ defined in (\ref{funzionale.pressione.multi.replica}). For each $p$ there exist the functions $f_{2j}^{\boldsymbol{\mu}^{p}}(\mathbf{x})\leq 0$, such that around $\boldsymbol{\mu}^{p}$ we can write $f$ as:
\begin{equation}\label{espansione.f.multi}
f(\mathbf{x}) =f(\boldsymbol{\mu}^{p})+\sum_{j=0}^{d}f_{2j}^{\boldsymbol{\mu}^{p}}(\mathbf{x}-\boldsymbol{\mu}^{p})+o\bigg(\Big(|\mathbf{x}'-\boldsymbol{\mu}^{p'}|^{2}+|\mathbf{x}''-\boldsymbol{\mu}^{p''}|^{2/q}\Big)^{d}\bigg)
\end{equation}
where $(\mathbf{x}',\mathbf{x}'')$ is a partition of the coordinate $\mathbf{x}$ and $q$ is a positive rational number such that $1/q\in\mathbb{N}$ and
\begin{equation*}
 f_{2j}^{\boldsymbol{\mu}^{p}}(t\mathbf{x}',t^{q}\mathbf{x}'')=t^{2j}f_{2j}^{\boldsymbol{\mu}^{p}}(\mathbf{x}',\mathbf{x}'')\quad\quad\text{all}\;\;t>0.
\end{equation*}

We define the type $k_{p}$ of the maximum point $\boldsymbol{\mu}^{p}$ as the minimum $d$ such that $f_{2d}^{\boldsymbol{\mu}^{p}}(\mathbf{x}-\boldsymbol{\mu}^{p})\neq 0$ as $\mathbf{x}\neq 0$ and $f_{2j}^{\boldsymbol{\mu}^{p}}(\mathbf{x}-\boldsymbol{\mu}^{p})=0$ for $j=1,\dots,d-1$.

We observe that when $q=1$ the expression (\ref{espansione.f.multi}) is the Taylor expansion of the function $f$. In this case $k_{p}$ is called the homogeneous type of the maximum point $\boldsymbol{\mu}^{p}$. In particular if a maximum points $\boldsymbol{\mu}^{p}$ has homogeneous type equal to $1$, around $\boldsymbol{\mu}$ we have:\\
%\vspace{0.5cm}
\begin{equation*}
f(\mathbf{x})=f(\boldsymbol{\mu}^{p})+\frac{1}{2}\langle\boldsymbol{\mathcal{H}}_{f}(\boldsymbol{\mu}^{p})(\mathbf{x}-\boldsymbol{\mu}^{p}),(\mathbf{x}-\boldsymbol{\mu}^{p})\rangle+o\Big(||(\mathbf{x}-\boldsymbol{\mu}^{p})^{2}||\Big)
\end{equation*}\\
\noindent where $\boldsymbol{\mathcal{H}}_{f}(\boldsymbol{\mu}^{p})$ is the Hessian matrix of $f$ computed in the maximum point $\boldsymbol{\mu}^{p}$. We define the maximal type $k^{*}$ of the function $f$ as the largest of the $k_{p}$. 

Define the function
\begin{equation}\label{funzione.B.multi}
B(\mathbf{x};\mathbf{y})=f(\mathbf{x}+\mathbf{y})-f(\mathbf{y}).
\end{equation} 
For each $p=1,\dots,P$ if $\boldsymbol{\mu}^{p}$ is a maximum point of homogeneous type $k_{p}$ there exists $\delta_{p}>0$ sufficiently small so that, as $N\rightarrow\infty$ for $||\mathbf{x}/\mathbf{N}^{1/2k_{p}}||<\delta_{p}$\\
\begin{align}\label{proprieta.B.multi}
N\cdot B\bigg(\dfrac{\mathbf{x}}{\mathbf{N}^{1/2k_{p}}},\boldsymbol{\mu}^{p}\bigg)&=f_{2k}^{\boldsymbol{\mu}^{p}}\bigg(\frac{\mathbf{x}}{\boldsymbol{\alpha}^{1/2k_{p}}}\bigg)+o(1)P_{2k_{p}}(\mathbf{x})\nonumber\\\\
N\cdot B\bigg(\dfrac{\mathbf{x}}{\mathbf{N}^{1/2k_{p}}},\boldsymbol{\mu}^{p}\bigg)&\leq\frac{1}{2}f_{2k}^{\boldsymbol{\mu}^{p}}\bigg(\frac{\mathbf{x}}{\boldsymbol{\alpha}^{1/2k_{p}}}\bigg)+o(1)P_{2k_{p}+1}(\mathbf{x})\nonumber
\end{align}
\noindent  where $\mathbf{N}=(N_{1},\dots,N_{n})$, $\boldsymbol{\alpha}=(\alpha_{1},\dots,\alpha_{n})$, $P_{2k_{p}}(\mathbf{x})$ is a polynomial of $2k_{p}$ degree and $P_{2k_{p}+1}(\mathbf{x})$ is a polynomial of $2k_{p}+1$ degree.\\

Now we can state our main results. The following describes the thermodynamic limiting distribution of the random vector of magnetizations.
\begin{teorema}\label{teo.multi.0}
 	Consider a system described by the mean-field Hamiltonian $H_{N}=-Ng(m_{1}(\boldsymbol{\sigma}),\dots,m_{n}(\boldsymbol{\sigma}))$ where $g$ is the convex function defined in (\ref{funzione.g}). Let $\boldsymbol{\mu^{1}},\dots,\boldsymbol{\mu^{P}}$ be the global maxima points of homogeneous maximal type $k^{*}$ of the function $f$ given by (\ref{funzionale.pressione.multi.replica}). 
        Then as $N\rightarrow\infty$
	\begin{equation*}
	(m_{1}(\boldsymbol{\sigma}),\dots,m_{n}(\boldsymbol{\sigma}))\overset{\mathscr{D}}{\rightarrow}\dfrac{\sum\limits_{p=1}^{P}b_{p}\delta(\mathbf{x}-\boldsymbol{\mu}^{p})}{\sum\limits_{p=1}^{P}b_{p}}
	\end{equation*}
	where $b_{p}=\displaystyle{\int_{\mathbb{R}}}\exp\bigg(f_{2k^{*}}^{\boldsymbol{\mu}^{p}}\Big(\dfrac{\mathbf{x}}{\boldsymbol{\alpha}^{1/2k^{*}}}\Big)\bigg)d\mathbf{x}$.
\end{teorema}
The following theorem solves the problem of the correct normalization of the random vector of the sums of spins whenever the function $f$ admits a unique maximum point.
\begin{teorema}\label{teo.multi.1}
	Consider a system described by the mean-field Hamiltonian $H_{N}=-Ng(m_{1}(\boldsymbol{\sigma}),\dots,m_{n}(\boldsymbol{\sigma}))$ where $g$ is the convex function defined in (\ref{funzione.g}). Let $\boldsymbol{\mu}=(\mu_{1},\dots,\mu_{n})$ be the unique global maximum point of the function $f$ given from (\ref{funzionale.pressione.multi.replica}). Let $k$ be the homogeneous type of the maximum point. 
	\begin{enumerate}
	\item If $k=1$ the asymptotic behaviour of the random vector\\
	\begin{equation}\label{vettore.S.1}
	\bar{\mathbf{S}}^{1}(\boldsymbol{\sigma})=\bigg(\dfrac{S_{1}(\boldsymbol{\sigma})-N_{1}\mu_{1}}{\sqrt{N_{1}}}, \dots,\dfrac{S_{n}(\boldsymbol{\sigma})-N_{n}\mu_{n}}{\sqrt{N_{n}}}\bigg)
	\end{equation}\\
	as $N_{1}\rightarrow\infty$, $\dots$, $N_{n}\rightarrow\infty$, for fixed values of $\alpha_{1},\dots, \alpha_{n}$, is given by a normal multivariate distribution whose covariance matrix is:\newpage
	\begin{equation}\label{covarianza}
	\widetilde{\boldsymbol{\chi}}=\begin{pmatrix}
	\dfrac{\partial\mu_{1}}{\partial h_{1}} & \sqrt{\dfrac{\partial \mu_{1}}{\partial h_{2}}\;\dfrac{\partial \mu_{2}}{\partial h_{1}}} & \dots & \sqrt{\dfrac{\partial \mu_{1}}{\partial h_{n}}\;\dfrac{\partial \mu_{n}}{\partial h_{1}}}\\\\
	\sqrt{\dfrac{\partial \mu_{1}}{\partial h_{2}}\;\dfrac{\partial \mu_{2}}{\partial h_{1}}} & \dfrac{\partial \mu_{2}}{\partial h_{2}} & \dots & \sqrt{\dfrac{\partial \mu_{2}}{\partial h_{n}}\;\dfrac{\partial \mu_{n}}{\partial h_{2}}}\\\\
	\vdots&\vdots&&\vdots \\\\
	\sqrt{\dfrac{\partial \mu_{1}}{\partial h_{n}}\;\dfrac{\partial \mu_{n}}{\partial h_{1}}} & \sqrt{\dfrac{\partial \mu_{2}}{\partial h_{n}}\;\dfrac{\partial \mu_{n}}{\partial h_{2}}} & \dots & \dfrac{\partial \mu_{n}}{\partial h_{n}}
	\end{pmatrix}\begin{matrix}
	\\\\\\\\\\\\\\.
	\end{matrix}
	\end{equation}
	%where $(\mu_{1},\dots,\mu_{n})$ is the solution of the Mean Field Equations (\ref{campomedio}) corresponding to the minimum.
	\item If $k>1$ the asymptotic behaviour of the random vector\\
	\begin{equation}\label{vettore.S.k}
	\bar{\mathbf{S}}^{k}(\boldsymbol{\sigma})=\bigg(\dfrac{S_{1}(\boldsymbol{\sigma})-N_{1}\mu_{1}}{(N_{1})^{1-1/2 k}}, \dots,\dfrac{S_{n}(\boldsymbol{\sigma})-N_{n}\mu_{n}}{(N_{n})^{1-1/2 k}}\bigg)
	\end{equation}\\
	\noindent as $N_{1}\rightarrow\infty$, $\dots$, $N_{n}\rightarrow\infty$, for fixed values of $\alpha_{1}$, $\dots$, $\alpha_{n}$, has density proportional to:
	\begin{equation*}
	\exp\bigg(f_{2k}^{\boldsymbol{\mu}}\Big(\dfrac{\mathbf{x}}{\boldsymbol{\alpha}^{1/2k}}\Big)\bigg)
	\end{equation*}
	\noindent where $\boldsymbol{\alpha}=(\alpha_{1},\dots ,\alpha_{n})$.
	\end{enumerate}
\end{teorema}

The following last theorem handles the case in which the function reaches the maximum in more than one point.
\begin{teorema}\label{teo.multi.2}
	Consider a system described by the mean-field Hamiltonian $H_{N}=-Ng(m_{1}(\boldsymbol{\sigma}),\dots,m_{n}(\boldsymbol{\sigma}))$ where $g$ is the convex function defined in (\ref{funzione.g}). Let $\boldsymbol{\mu}=(\mu_{1},\dots,\mu_{n})$ be a nonunique global maximum point of the function $f$ given from (\ref{funzionale.pressione.multi.replica}). Let $k$ be the homogeneous type of the maximum point. Define $\bar{\delta}$ to be the minimum distance between all distinct pair of global maximum points of the function $f$. 
	
	Then for any $d\in(0,\bar{\delta})$ when the random vector of the magnetizations $(m_{1}(\boldsymbol{\sigma}),\dots,m_{n}(\boldsymbol{\sigma}))$ is inside the ball $B(\boldsymbol{\mu},d)$ centered in $\boldsymbol{\mu}$ of radius $d$
	\begin{enumerate}
	\item  if $k=1$ the asymptotic behaviour of the random vector $\bar{\mathbf{S}}^{1}(\boldsymbol{\sigma})$ defined in (\ref{vettore.S.1}) as $N_{1}\rightarrow\infty$, $\dots$, $N_{n}\rightarrow\infty$, for fixed values of $\alpha_{1}$, $\dots$, $\alpha_{n}$, is given by a normal multivariate distribution whose covariance matrix is given by (\ref{covarianza});
	%and $(\mu_{1},\dots,\mu_{n})$ is the solution of the Mean Field Equations (\ref{campomedio}) corresponding to the given minimum.
	\item if $k>1$ the asymptotic behaviour of the random vector $\bar{{S}}_{k}(\boldsymbol{\sigma})$ defined in (\ref{vettore.S.k}) as $N_{1}\rightarrow\infty$, $\dots$, $N_{n}\rightarrow\infty$, for fixed values of $\alpha_{1}$, $\dots$, $\alpha_{n}$, has density proportional to:
	\begin{equation*}
	\exp\bigg(f_{2k}^{\boldsymbol{\mu}}\Big(\dfrac{\mathbf{x}}{\boldsymbol{\alpha}^{1/2k}}\Big)\bigg)
	\end{equation*}
	where $\boldsymbol{\alpha}=(\alpha_{1},\dots ,\alpha_{n})$.
	\end{enumerate}
\end{teorema}
The proofs of theorems \ref{teo.multi.0}, \ref{teo.multi.1} and \ref{teo.multi.2} need the following results.
\begin{lemma}\label{lemma.multi.1}
	Suppose that for each $N$, $\mathbf{X}_{N}=(X_{N_{1}},\dots, X_{N_{1}})$ and $\mathbf{Y}_{N}=(Y_{N_{1}},\dots, Y_{N_{n}})$ are independent random vectors. Suppose that $\mathbf{X}_{N}$ weakly converges to a distribution $\nu$ such that
	\begin{equation*}
	\int_{\mathbb{R}^{n}} e^{i\langle\mathbf{r},\mathbf{x}\rangle}d\nu(\mathbf{x})\neq 0 \quad\quad\text{for all}\;\;\mathbf{r}\in\mathbb{R}^{n}\;.
	\end{equation*}
	Then $\mathbf{Y}_{N}$ weakly converges to $\mu$ if and only if $\mathbf{X}_{N}+\mathbf{Y}_{N}$ weakly converges to the convolution $\nu *\mu$ of the distributions $\nu$ and $\mu$.
\end{lemma}
\begin{proof}
	Weak convergence of measures is equivalent to pointwise convergence of characteristic functions.
\end{proof}

\begin{lemma}\label{lemma.multi.2}
	Let $\mathbf{A}=\mathbf{D}_{\boldsymbol{\alpha}}\mathbf{J}\mathbf{D}_{\boldsymbol{\alpha}}$ be a positive defined matrix where the matrix $\mathbf{D}_{\boldsymbol{\alpha}}=diag\{\sqrt{\alpha_{1}},\dots,\sqrt{\alpha_{n}}\}$ and the matrix  $\mathbf{J}$ is defined in (\ref{interazioneridotta}). Given the random vector $(W_{1},\dots, W_{n})$ whose joint distribution is the normal multivariate
	\begin{equation}\label{normalebivariata}
	\bigg(\frac{\det \mathbf{A}}{(2\pi)^{n}}\bigg)^{\frac{1}{2}}\; \exp\bigg(\!-\frac{1}{2}\langle\mathbf{A}\mathbf{w},\mathbf{w}\rangle\bigg)
	\end{equation}\\
	\noindent if $(W_{1},\dots, W_{n})$ is independent of $(S_{1}(\boldsymbol{\sigma}),\dots, S_{n}(\boldsymbol{\sigma}))$ then for $(m_{1},\dots,m_{n}) \in \mathbb{R}^{n}$ and $\gamma\in \mathbb{R}$ the joint distribution of\\
	\begin{equation}\label{sommavettori}
	\bigg(\dfrac{W_{1}}{(N_{1})^{1/2-\gamma}},\dots,\dfrac{W_{n}}{(N_{n})^{1/2-\gamma}}\bigg)+\bigg(\dfrac{S_{1}(\boldsymbol{\sigma})-N_{1}m_{1}}{(N_{1})^{1-\gamma}},\dots, \dfrac{S_{n}(\boldsymbol{\sigma})-N_{n}m_{n}}{(N_{n})^{1-\gamma}}\bigg)
	\end{equation}\newpage
	\noindent is given by
	\begin{equation}\label{distribuzione.tesi}
	\dfrac{\exp\bigg(\!Nf\Big(\dfrac{x_{1}}{N_{1}^{\;\gamma}}+m_{1},\dots,\dfrac{x_{n}}{N_{n}^{\;\gamma}}+m_{n}\Big)\bigg)dx_{1}\dots dx_{n}}{\displaystyle{\int_{\mathbb{R}^{n}}} \exp\bigg(\!Nf\Big(\dfrac{x_{1}}{N_{1}^{\;\gamma}}+m_{1},\dots,\dfrac{x_{n}}{N_{n}^{\;\gamma}}+m_{n}\Big)\bigg)dx_{1}\dots dx_{n}}
	\end{equation}\\
	\noindent where $f(x_{1},\dots,x_{n})$ is the function defined in (\ref{funzionale.pressione.multi.replica}).
\end{lemma}\vspace{0.3cm}
\begin{proof}
Given $\theta_{1},\dots, \theta_{n}$ real\\
\begin{multline*}
\!\!P\bigg\{\!\dfrac{W_{1}}{(N_{1})^{1/2-\gamma}}+\dfrac{S_{1}(\boldsymbol{\sigma})\!-\!N_{1}m_{1}}{(N_{1})^{1-\gamma}}\!\leq\theta_{1},\!\dots ,\!\dfrac{W_{n}}{(N_{n})^{1/2-\gamma}}+\dfrac{S_{n}(\boldsymbol{\sigma})\!-\!N_{n}m_{n}}{(N_{n})^{1-\gamma}}\leq\theta_{n}\!\bigg\}\\\\\quad\quad=
P\;\Big\{\sqrt{N_{1}}W_{1}+S_{1}(\boldsymbol{\sigma})\in E_{1},\;\dots,\;\sqrt{N_{n}}W_{n}+S_{n}(\boldsymbol{\sigma})\in E_{n}\Big\}
\end{multline*}\\
where $E_{l}=(-\infty,\;(N_{l})^{1-\gamma}\theta_{l}+N_{l}m_{l}]$. The distribution of the random vector $(\sqrt{N_{1}}\;W_{1},\dots ,\sqrt{N_{n}}\;W_{n})$ is 
\begin{equation}\label{normale.bivariata.2}
\bigg(\frac{\det\widetilde{\mathbf{A}}}{(2\pi)^{n}}\bigg)^{\frac{1}{2}} \;\exp\Big(\!-\frac{1}{2}\langle\widetilde{\mathbf{A}}\mathbf{x},\mathbf{x}\rangle\Big)
\end{equation}
\noindent where it is easy to verify that $\widetilde{\mathbf{A}}=1/N\mathbf{J}$. We claim that since the matrix $\mathbf{A}$ is positive define also $\widetilde{\mathbf{A}}$ has this property. 
%Expressed the Hamiltonian (\ref{hamiltoniana}) as function of the sums of spins
%\begin{equation}\label{Hamiltonianasomme}
%H_{N}(S_{1},\dots, S_{n})=-\frac{1}{2}\sum_{l, s=1}^{n} J_{ls}\sqrt{\frac{\alpha_{l}\alpha_{s}}{N_{l}N_{s}}}S_{l}S_{s}-\sum_{l=1}^{n}h_{l}S_{l}\;.
%\end{equation}
The joint distribution of the random vector $(S_{1}(\boldsymbol{\sigma}),\dots ,S_{n}(\boldsymbol{\sigma}))$ is:\\
\begin{equation}\label{distribuzione.somme}
\frac{1}{Z_{N}(\mathbf{J},\mathbf{h})}\exp\bigg(\frac{1}{2N}\langle\mathbf{J}\mathbf{s},\mathbf{s}\rangle+\langle\mathbf{h},\mathbf{s}\rangle\bigg)d\nu_{S}(\mathbf{s})
\end{equation}\\
\noindent where $\nu_{S}(\mathbf{s})$ is the distribution of $(S_{1}(\boldsymbol{\sigma}),\dots ,S_{n}(\boldsymbol{\sigma}))$ on $(\mathbb{R}^{N},\prod_{i=1}^{N}\rho(\sigma_{i}))$. 
%In particular $\nu(\mathbf{s})=\rho^{*N_{1}}(s_{1})\dots\rho^{*N_{n}}(s_{n})$ where $\rho^{*N_{i}}(s_{i})$ denotes the $N_{i}$-fold convolution of the measure $\rho$ given by (\ref{misura.ro}) with itself. 
The distribution of the random vector (\ref{sommavettori}) is given by the convolution of the distribution (\ref{normale.bivariata.2}) with the distribution (\ref{distribuzione.somme}).\newpage 
\noindent Thus:
\begin{align*}
P&\Big\{\sqrt{N_{1}}\;W_{1}+S_{1}\in E_{1},\dots ,\sqrt{N_{n}}\;W_{n}+S_{n}\in E_{n}\Big\} \nonumber\\\nonumber\\
&=\frac{1}{Z_{N}(\mathbf{J},\mathbf{h})}\bigg(\frac{\det\widetilde{\mathbf{A}}}{(2\pi)^{n}}\bigg)^{\frac{1}{2}}\nonumber\\
&\quad\times\!\!\iint_{\bigotimes\limits_{l=1}^{n} E_{l}\times\mathbb{R}^{n}}\!\!\!\!\!\!\exp\bigg(\!\frac{1}{2N}\Big(\!-\langle\mathbf{J}(\mathbf{w}-\mathbf{s}),(\mathbf{w}-\mathbf{s})\rangle+\langle\mathbf{J}\mathbf{s},\mathbf{s}\rangle\Big)+\langle\mathbf{h},\mathbf{s}\rangle\bigg)d\nu_{S}(\mathbf{s})d\mathbf{x}\nonumber\\\nonumber\\
&=\frac{1}{Z_{N}(\mathbf{J},\mathbf{h})}\bigg(\frac{\det\widetilde{\mathbf{A}}}{(2\pi)^{n}}\bigg)^{\frac{1}{2}}\nonumber\\
&\quad\times\!\!\int_{\bigotimes\limits_{l=1}^{n} E_{l}}\!\!\!\!\!\exp\bigg(\!\!-\frac{1}{2N}\langle\mathbf{J}\mathbf{w},\mathbf{w}\rangle\bigg)\!\!\int_{\mathbb{R}^{n}}\exp\bigg(\frac{1}{N}\langle\mathbf{J}\mathbf{w},\mathbf{s}\rangle+\langle\mathbf{h},\mathbf{s}\rangle\bigg)d\nu_{S}(\mathbf{s})d\mathbf{w}
\end{align*}
where $\bigotimes\limits_{l=1}^{n} E_{l}$ denotes the Cartesian product of the sets $E_{l}$. 

\noindent Since
\begin{align*}
\int_{\mathbb{R}^{n}}\exp\bigg(\frac{1}{N}\langle\mathbf{J}&\mathbf{w},\mathbf{s}\rangle+\langle\mathbf{h},\mathbf{s}\rangle\bigg)d\nu_{S}(\mathbf{s})\\&=\prod_{l=1}^{n}\int_{\mathbb{R}^{N_{1}}}\!\!\!\exp\bigg(\sum_{i\in P_{l}}\sigma_{i}\bigg(h_{l}+\frac{1}{N}\sum_{p=1}^{n}J_{lp}w_{p}\bigg)\bigg)\prod_{i\in P_{l}}d\rho(\sigma_{i})\nonumber\\
&=\prod_{l=1}^{n}\prod_{i\in P_{l}}\int_{\mathbb{R}}\exp\bigg(\sigma_{i}\bigg(h_{l}+\frac{1}{N}\sum_{p=1}^{n}J_{lp}w_{p}\bigg)\bigg)d\rho(\sigma_{i})
\end{align*}
making the following change of variables
\begin{equation*}
x_{l}=\dfrac{w_{l}-N_{l}m_{l}}{(N_{l})^{1-\gamma}}\quad\quad l=1,\dots, n
\end{equation*}
\noindent and integrating over $\mathbf{s}$, we obtain:

\begin{align}\label{calcoloneproseguo}
P&\Big\{\sqrt{N_{1}}\;W_{1}+S_{1}\in E_{1},\dots ,\sqrt{N_{n}}\;W_{n}+S_{n}\in E_{n}\Big\}\nonumber\\\nonumber\\
&=\dfrac{\prod\limits_{l=1}^{n}(N_{l})^{1-\gamma}}{Z_{N}(\mathbf{J},\mathbf{h})}\bigg(\frac{\det\widetilde{\mathbf{A}}}{(2\pi)^{n}}\bigg)^{\frac{1}{2}}\nonumber\\
&\quad\times\int_{-\infty}^{\theta_{1}}\dots\int\limits_{-\infty}^{\theta_{n}}\exp\bigg(\!-\frac{N}{2}\sum_{l, p=1}^{n}\alpha_{l}\alpha_{p}J_{lp}\Big(\dfrac{x_{l}}{N_{l}^{\;\gamma}}+m_{l}\Big)\Big(\dfrac{x_{p}}{N_{p}^{\;\gamma}}+m_{p}\Big)+\nonumber\\
&\quad+\sum_{l=1}^{n}N_{l}\ln\Big(\cosh\Big(h_{l}\sum_{p=1}^{n}\alpha_{p}J_{lp}\Big(\dfrac{x_{p}}{N_{p}^{\;\gamma}}+m_{p}\Big)\Big)\Big)\bigg)dx_{1}\dots dx_{n}\nonumber\\\nonumber\\
&=\dfrac{\prod\limits_{l=1}^{n}(N_{l})^{1-\gamma}}{Z_{N}(\mathbf{J},\mathbf{h})}\bigg(\frac{\det\widetilde{\mathbf{A}}}{(2\pi)^{n}}\bigg)^{\frac{1}{2}}\nonumber\\
&\quad\times\int_{-\infty}^{\theta_{1}}\dots\int\limits_{-\infty}^{\theta_{n}}\exp\bigg(\!Nf\Big(\dfrac{x_{1}}{N_{1}^{\;\gamma}}+m_{1},\dots,\dfrac{x_{n}}{N_{n}^{\;\gamma}}+m_{n}\Big)\bigg)dx_{1}\dots dx_{n}.
\end{align}\\
Taking $\theta_{1}\rightarrow\infty,\dots,\theta_{n}\rightarrow\infty$ the (\ref{calcoloneproseguo}) gives an equation for $Z_{N}(\mathbf{J},\mathbf{h})$ which when substituted back yields the result (\ref{distribuzione.tesi}). The integral in the last expression is finite by (\ref{proprieta.f.multi.1}).
\end{proof}
We remark that as $\gamma\!<\!1/2$, the random vector $(W_{1},\dots, W_{n})$ does not contribute to the limit of (\ref{distribuzione.tesi}) as $N_{1}\rightarrow\infty$, $\dots$, $N_{n}\rightarrow\infty$.
\begin{lemma}\label{lemma.multi.3}
	Defined $F=\max\{f(\mathbf{x})|\mathbf{x}\in\mathbb{R}^{n}\}$, let $V$ be any closed (possibly unbounded) subset of $\mathbb{R}^{n}$ which contains no global maxima of the function $f$ defined in $\ref{funzionale.pressione.multi.replica}$. Then for any $\mathbf{t}\in\mathbb{R}^{n}$ there exists $\epsilon>0$ such that\\
	\begin{equation}\label{tesi.lemma.3.multi}
	e^{-NF}\int_{V} \exp(Nf(\mathbf{x})+\langle \mathbf{t},\mathbf{N}^{\gamma}\mathbf{x}\rangle)d\mathbf{x}=O(e^{-N\epsilon}) \quad\quad\quad N\rightarrow\infty.
	\end{equation}
\end{lemma}
\begin{proof}
	$V$ contains no global maxima of $f(\mathbf{x})$, thus there exists $\bar{\epsilon}>0$ such that:
	\begin{equation*}
	\sup_{\mathbf{x}\in V}f(\mathbf{x})\leq\sup_{\mathbf{x}\in\mathbb{R}}f(\mathbf{x})-\bar{\epsilon}=F-\bar{\epsilon}.
	\end{equation*}\\
	Pick $M>0$ so large that for every $N$, whenever $||\mathbf{x}||\geq M$
	\begin{equation*}
	Nf(\mathbf{x})+\langle \mathbf{t},\mathbf{N}^{\gamma}\mathbf{x}\rangle\leq -\frac{1}{5}N\langle\widetilde{\mathbf{J}}\mathbf{x},\mathbf{x}\rangle\leq -\frac{1}{5}\langle\widetilde{\mathbf{J}}\mathbf{x},\mathbf{x}\rangle
	\end{equation*}
	we can write
	\begin{align}\label{dim.3}
	e^{-NF}&\int_{V}e^{Nf(\mathbf{x})+\langle \mathbf{t},\mathbf{N}^{\gamma}\mathbf{x}\rangle}d\mathbf{x}\nonumber\\ &=e^{-NF}\bigg(\int_{V\cap\{||\mathbf{x}||\geq M\}}\!\!\!\!\!\!\!e^{Nf(\mathbf{x})+\langle \mathbf{t},\mathbf{N}^{\gamma}\mathbf{x}\rangle}d\mathbf{x}+\!\!\int_{V\cap\{||\mathbf{x}||<M\}}\!\!\!\!\!\!\!e^{Nf(\mathbf{x})+\langle \mathbf{t},\mathbf{N}^{\gamma}\mathbf{x}\rangle}d\mathbf{x}\bigg)\nonumber\\
	&\leq e^{-NF}\int_{V\cap\{||\mathbf{x}||\geq M\}}\!\!\!\!\!\!e^{-\frac{1}{5}\langle\widetilde{\mathbf{J}}\mathbf{x},\mathbf{x}\rangle}d\mathbf{x}+e^{-NF}e^{-N\bar{\epsilon}}e^{\mathbf{N}^{\gamma}M||\mathbf{t}||}\int_{V\cap\{||\mathbf{x}||<M\}}\!\!\!d\mathbf{x}.
	\end{align}
	Since $F\geq f(0,\dots,0)\geq 0$ and $\widetilde{\mathbf{J}}$ is positive define the latter passage of (\ref{dim.3}) as $N\rightarrow\infty$ is $O(e^{-N\epsilon})$ where $\epsilon=\min\{F,\bar{\epsilon}\}$. This proves the (\ref{tesi.lemma.3.multi})
\end{proof}\vspace{0.3cm}
\begin{lemma}\label{lemma.multi.4}
Let $\boldsymbol{\mu}$ be a global maximum point of the function $f$ given by (\ref{funzionale.pressione.multi.replica}). Let $k$ be the homogeneous type of $\boldsymbol{\mu}$. Define $F=\max\{f(\mathbf{x})|\mathbf{x}\in\mathbb{R}^{n}\}$. Then there exists a positive number $\delta_{\boldsymbol{\mu}}$ such that for any $\mathbf{t}\in\mathbb{R}^{n}$, any $k\in\mathbb{N}$, any $\delta\in(0,\delta_{\boldsymbol{\mu}}]$ and any bounded continuous function $\phi:\mathbb{R}^{n}\rightarrow\mathbb{R}$
\begin{multline}\label{tesi.lemma.multi.4}
\lim_{N\rightarrow\infty}e^{-\langle\mathbf{t},\mathbf{N}^{1/2k}\boldsymbol{\mu}\rangle}\Big(\prod_{l=1}^{n}N_{l}\Big)^{1/2k}e^{-NF}\int_{B(\boldsymbol{\mu},\delta)}e^{Nf(\mathbf{x})+\langle\mathbf{t},\mathbf{N}^{1/2k}\mathbf{x}}\rangle\phi(\mathbf{x})d\mathbf{x}\\=
\phi(\boldsymbol{\mu})\int_{\mathbb{R}^{n}}\exp\bigg(f^{\boldsymbol{\mu}}_{2k}\Big(\dfrac{\mathbf{x}}{\boldsymbol{\alpha}^{1/2k}}\Big)+\langle\mathbf{t},\mathbf{x}\rangle\bigg)d\mathbf{x}.
\end{multline}
\end{lemma}
\vspace{0.3cm}
\begin{proof}
To make the notation easier we define $\gamma=1/2k$. Making the change of variable
	\begin{equation*}
	x_{l}=\mu_{l}+\dfrac{u_{l}}{N_{l}}\quad l=1,\dots,n
	\end{equation*}
	the left-hand side of (\ref{tesi.lemma.multi.4}) becomes
	\begin{equation}\label{passaggio.lemma.multi.4}
	\lim_{N\rightarrow\infty}\int_{||\frac{\mathbf{u}}{\mathbf{N}^{\gamma}}||\leq d}\phi\Big(\boldsymbol{\mu}+\frac{\mathbf{u}}{\mathbf{N}^{\gamma}}\Big)\exp\bigg(NB\Big(\boldsymbol{\mu}+\frac{\mathbf{u}}{\mathbf{N}^{\gamma}};\boldsymbol{\mu}\Big)+\langle\mathbf{t},\mathbf{x}\rangle\bigg)d\mathbf{x}
	\end{equation}
	where $B$ is the function defined in (\ref{funzione.B.multi}). By the conditions expressed in (\ref{proprieta.B.multi}) and the dominate convergence theorem the limit (\ref{passaggio.lemma.multi.4}) is equal to
	\begin{equation}\label{passaggio2.lemma.multi.4}
	\phi(\boldsymbol{\mu})\int_{\mathbb{R}^{n}}\exp\bigg(f^{\boldsymbol{\mu}}_{2k}\Big(\dfrac{\mathbf{u}}{\boldsymbol{\alpha}^{1/2k}}\Big)+\langle\mathbf{t},\mathbf{u}\rangle\bigg)d\mathbf{u}.
	\end{equation}
	Since $f^{\boldsymbol{\mu}^{p}}_{2k}(\frac{\mathbf{x}}{\boldsymbol{\alpha}^{1/2k}})<0$ for every $\mathbf{x}$ different from zero, the integral in (\ref{passaggio2.lemma.multi.4}) is finite. 
	This completes the proof of the lemma. 
\end{proof}\vspace{0.3cm}
\noindent {\it Proof of Theorem \ref{teo.multi.0}}.
	By definition $m_{l}(\boldsymbol{\sigma})=S_{l}(\boldsymbol{\sigma})/N$, thus by lemmas \ref{lemma.multi.1} and \ref{lemma.multi.2} (with $\gamma=0$ and $\mathbf{m}=\mathbf{0}$) we know that consider a random vector $(W_{1},\dots,W_{n})$ whose distribution is the multivariate Gaussian given by (\ref{normalebivariata}), if $(W_{1},\dots,W_{n})$ is independent from $(S_{1}(\boldsymbol{\sigma}),\dots,S_{n}(\boldsymbol{\sigma}))$ the distribution of
	\begin{equation*}
	\bigg(\dfrac{W_{1}}{\sqrt{N_{1}}},\dots,\dfrac{W_{1}}{\sqrt{N_{1}}}\bigg)+\bigg(\dfrac{S_{1}(\boldsymbol{\sigma})}{N_{1}},\dots,\dfrac{S_{n}(\boldsymbol{\sigma})}{N_{n}}\bigg)
	\end{equation*}\\
	is given by
	\begin{equation*}
	\dfrac{\exp(Nf(\mathbf{x}))d\mathbf{x}}{\displaystyle{\int_{\mathbb{R}^{n}}}\exp(Nf(\mathbf{x}))d\mathbf{x}}.
	\end{equation*}\\
	Thus to prove the theorem it is enough to show that for any bounded continuous function $\phi:\mathbb{R}^{n}\rightarrow\mathbb{R}$
	\begin{equation}\label{teo.0.passaggio.1}
	\dfrac{\displaystyle{\int_{\mathbb{R}}}e^{Nf(\mathbf{x})}\phi(\mathbf{x})d\mathbf{x}}{\displaystyle{\int_{\mathbb{R}}}e^{Nf(\mathbf{x})}d\mathbf{x}}\rightarrow\dfrac{\sum\limits_{p=1}^{P}\phi(\mu_{p})b_{p}}{\sum\limits_{p=1}^{P}b_{p}}.
	\end{equation}
	Consider $\delta_{1},\dots,\delta_{P}$ such that conditions given in (\ref{proprieta.B.multi}) are satisfied. We choose $\bar{\delta}=\min\{\delta_{p}\;|\;p=1,\dots,P\}$ decreasing it (if necessary) to assure that $0<\bar{\delta}<\min\{|\boldsymbol{\mu}^{p}-\boldsymbol{\mu}^{q}|:1\leq p\neq q\leq P\}$. Denoted by $V$ the closet set\\
	\begin{equation*}
	V=\mathbb{R}-\bigcup_{p=1}^{P}B(\boldsymbol{\mu}_{p},\bar{\delta})
	\end{equation*}\\
	where $B(\boldsymbol{\mu}_{p},\bar{\delta})$ is the ball centered in $\boldsymbol{\mu}^{p}$ of radius $\bar{\delta}$, by lemma \ref{lemma.multi.3} with $\mathbf{t}=\mathbf{0}$ there exists $\epsilon>0$ such that as $N\rightarrow\infty$\\
	\begin{equation}\label{teo.0.passaggio.2}
	e^{-NF}\int_{V}e^{Nf(\mathbf{x})}\phi(\mathbf{x})d\mathbf{x}=O(e^{-N\epsilon}). 
	\end{equation}\\
	Moreover, for each $p=1,\dots,P$ by lemma \ref{lemma.multi.4} we have\\
	\begin{multline}\label{teo.0.passaggio.3}
	\lim_{N\rightarrow\infty}\Big(\prod_{l=1}^{n} N_{l}\Big)^{\frac{1}{2k^{*}}}e^{-NF}\int_{B(\boldsymbol{\mu}^{p},\bar{\delta})}e^{Nf(\mathbf{x})}\phi(\mathbf{x})d\mathbf{x} \\=\phi(\boldsymbol{\mu}^{p})\int_{\mathbb{R}}\exp\bigg(f^{\boldsymbol{\mu}^{p}}_{2k^{*}}\Big(\dfrac{\mathbf{x}}{\boldsymbol{\alpha}^{1/2k^{*}}}\Big)\bigg)d\mathbf{x}.
	\end{multline}
	By (\ref{teo.0.passaggio.2}) and (\ref{teo.0.passaggio.3}) we get
	\begin{multline}\label{teo.0.passaggio.6}
	\lim_{N\rightarrow\infty}\Big(\prod_{l=1}^{n} N_{l}\Big)^{\frac{1}{2k^{*}}}e^{-NF}\int_{\mathbb{R}^{n}}e^{Nf(\mathbf{x})}\phi(\mathbf{x})d\mathbf{x} \\=\sum_{p=1}^{P}\phi(\boldsymbol{\mu}^{p})\int_{\mathbb{R}}\exp\bigg(f^{\boldsymbol{\mu}_{p}}_{2k^{*}}\Big(\frac{\mathbf{x}}{\boldsymbol{\alpha}^{1/2k^{*}}}\Big)\bigg)d\mathbf{x}.
	\end{multline}\\
	In a similar way for the denominator we have:
	\begin{multline}\label{teo.0.passaggio.7}
	\lim_{N\rightarrow\infty}\Big(\prod_{l=1}^{n} N_{l}\Big)^{\frac{1}{2k^{*}}}e^{-NF}\int_{\mathbb{R}^{n}}e^{Nf(\mathbf{x})}d\mathbf{x} \\=\sum_{p=1}^{P}\int_{\mathbb{R}}\exp\bigg(f^{\boldsymbol{\mu}^{p}}_{2k^{*}}\Big(\frac{\mathbf{x}}{\boldsymbol{\alpha}^{1/2k^{*}}}\Big)\bigg)d\mathbf{x}.
	\end{multline}\\
	Now the statement (\ref{teo.0.passaggio.1}) follows from (\ref{teo.0.passaggio.6}) and (\ref{teo.0.passaggio.7}).
	\qed
	\vspace{0.5cm}\\
	\noindent{\it Proof of Theorem \ref{teo.multi.1}}.
	For $k>1$, by lemmas \ref{lemma.multi.1} and \ref{lemma.multi.2} (with $\gamma=1/2k$ and $\mathbf{m}=\boldsymbol{\mu}$), we have to prove that, for any bounded continuous function $\psi:\mathbb{R}^{n}\rightarrow\mathbb{R}$
	\begin{equation*}
	\dfrac{\displaystyle{\int_{\mathbb{R}^{n}}} \exp\bigg(\!Nf\Big(\dfrac{\mathbf{x}}{\mathbf{N}^{1/2k}}+\boldsymbol{\mu}\Big)\bigg)\psi(\mathbf{x})d\mathbf{x}}{\displaystyle{\int_{\mathbb{R}^{n}}} \exp\bigg(\!Nf\Big(\dfrac{\mathbf{x}}{\mathbf{N}^{1/2k}}+\boldsymbol{\mu}\Big)\bigg)d\mathbf{x}}
	\rightarrow \dfrac{\displaystyle{\int_{\mathbb{R}^{n}}} \exp\bigg(f_{2k}^{\boldsymbol{\mu}}\bigg(\frac{\mathbf{x}}{\boldsymbol{\alpha}^{1/2k}}\bigg)\bigg)\psi(\mathbf{x})d\mathbf{x}}{\displaystyle{\int_{\mathbb{R}^{n}}} \exp\bigg(f_{2k}^{\boldsymbol{\mu}}\bigg(\frac{\mathbf{x}}{\boldsymbol{\alpha}^{1/2k}}\bigg)\bigg)d\mathbf{x}}.
	\end{equation*}
	%where to ease the notation we set $\mathbf{N}=(N_{1},\dots,N_{n})$.
	To make the notation easier we set $\gamma=1/2k$. 
	We pick $\delta >0$ such as conditions (\ref{proprieta.B.multi}) are verified. Defined $F=f(\mathbf{\boldsymbol{\mu}})$, by lemma \ref{lemma.multi.3} with $\mathbf{t}=\mathbf{0}$ there exists $\epsilon >0$ so that
	\begin{equation}\label{archelo}
	e^{-NF}\int_{||\frac{\mathbf{x}}{\mathbf{N}^{\gamma}}||\geq\delta}\exp\bigg(\!Nf\Big(\dfrac{\mathbf{x}}{\mathbf{N}^{\gamma}}+\boldsymbol{\mu}\Big)\bigg)\psi(\mathbf{x})d\mathbf{x}=O\Big(\Big(\prod_{l=1}^{n} N_{l}\Big)^{\gamma}e^{-N\epsilon}\Big).
	\end{equation}
	On the other hand
	\begin{multline}\label{chenometido}
	e^{-NF}\int_{||\frac{\mathbf{x}}{\mathbf{N}^{\gamma}}||<\delta} \exp\bigg(\!Nf\Big(\dfrac{\mathbf{x}}{\mathbf{N}^{\gamma}}+\boldsymbol{\mu}\Big)\bigg)\psi(\mathbf{x})d\mathbf{x}\\
	=e^{-N(F-f(\boldsymbol{\mu}))}\int_{||\frac{\mathbf{x}}{\mathbf{N}^{\gamma}}||<\delta} \exp\bigg(\!NB\Big(\dfrac{\mathbf{x}}{\mathbf{N}^{\gamma}},\boldsymbol{\mu}\Big)\bigg)\psi(\mathbf{x})d\mathbf{x}.	
	\end{multline}
	Thus by conditions (\ref{proprieta.B.multi}) and the dominated convergence theorem the right-hand side of (\ref{chenometido}) as $N\rightarrow\infty$ tends to
	\begin{equation*}
	\int_{||\frac{\mathbf{x}}{\mathbf{N}^{\gamma}}||<\delta} \exp\bigg(f_{2k}^{\boldsymbol{\mu}}\bigg(\frac{\mathbf{x}}{\boldsymbol{\alpha}^{1/2k}}\bigg)\bigg)\psi(\mathbf{x})d\mathbf{x}.
	\end{equation*}
	\noindent This with (\ref{archelo}) proves the statement $(2)$ of the theorem.

	We observe that for $k=1$ the following identity holds
	\begin{equation*}
	N\cdot f_{2k}^{\boldsymbol{\mu}}\bigg(\frac{\mathbf{x}}{\boldsymbol{\alpha}^{1/2k}}\bigg)=-\frac{1}{2}\langle\boldsymbol{-\widetilde{\mathcal{H}}}_{f}(\boldsymbol{\mu})\mathbf{x},\mathbf{x}\rangle
	\end{equation*}
	\noindent where $\boldsymbol{\widetilde{\mathcal{H}}}_{f}=\mathbf{D}_{\alpha}^{-1}\boldsymbol{\mathcal{H}}_{f}\mathbf{D}_{\alpha}^{-1}$ is a negative define matrix. Following a similar procedure, by (\ref{proprieta.B.multi}) and the dominate convergence theorem, we prove that for any bounded continuous function $\psi:\mathbb{R}^{n}\rightarrow\mathbb{R}$ the ratio\newpage
	\begin{equation*}
	\dfrac{\displaystyle{\int_{\mathbb{R}^{n}}} \exp\bigg(\!Nf\Big(\dfrac{\mathbf{x}}{\sqrt{\mathbf{N}}}+\boldsymbol{\mu}\Big)\bigg)\psi(\mathbf{x})d\mathbf{x}}{\displaystyle{\int_{\mathbb{R}^{n}}}\exp\bigg(\!Nf\Big(\dfrac{\mathbf{x}}{\sqrt{\mathbf{N}}}+\boldsymbol{\mu}\Big)\bigg)d\mathbf{x}}
	\end{equation*}
	converges to\\
	\begin{equation*}
	\bigg(\dfrac{-\det \boldsymbol{\widetilde{\mathcal{H}}}_{f}(\boldsymbol{\mu})}{(2\pi)^{n}}\bigg)^{\frac{1}{2}}\int_{\mathbb{R}^{n}} \exp\Big(-\frac{1}{2}\langle\boldsymbol{-\widetilde{\mathcal{H}}}_{f}(\boldsymbol{\mu})\mathbf{x},\mathbf{x}\rangle\Big)\psi(\mathbf{x})d\mathbf{x}.
	\end{equation*}\\
	\noindent  The multivariate Gaussian obtained
	\begin{equation*}
	\bigg(\dfrac{-\det \boldsymbol{\widetilde{\mathcal{H}}}_{f}(\boldsymbol{\mu})}{(2\pi)^{n}}\bigg)^{\frac{1}{2}} \exp\Big(\!-\frac{1}{2}\langle-\boldsymbol{\widetilde{\mathcal{H}}}_{f}(\boldsymbol{\mu})\mathbf{x},\mathbf{x}\rangle\Big)d\mathbf{x}
	\end{equation*}
	\noindent is the convolution of the distribution of the random vector $(W_{1},\dots,W_{n})$ with the distribution of the random vector $\bar{\mathbf{S}}^{1}(\boldsymbol{\sigma})$ given by (\ref{vettore.S.1}). 

	Indicated with $\phi_{\mathbf{W}}(\boldsymbol{\lambda})$, $\phi_{\bar{\mathbf{S}}^{1}}(\boldsymbol{\lambda})$ and $\phi_{\mathbf{W}+\bar{\mathbf{S}}^{1}}(\boldsymbol{\lambda})$ respectively the characteristic function of the random vectors $(W_{1},\dots,W_{n})$, of the random vector (\ref{vettore.S.1}) and of their sum the following equality holds:
	\begin{equation}\label{prodcaratteristiche}
	\phi_{\mathbf{W}+\bar{\mathbf{S}}^{1}}(\boldsymbol{\lambda})=\phi_{\mathbf{W}}(\boldsymbol{\lambda})\;\phi_{\bar{\mathbf{S}}^{1}}(\boldsymbol{\lambda})\;.
	\end{equation}
	We remember that the characteristic function of a random vector $(X_{1},\dots,X_{n})$ whose joint distribution is a multivariate Gaussian $$\sqrt{\frac{\det \mathbf{B}}{(2\pi)^{n}}}\; \exp\Big(\!-\frac{1}{2}\langle\mathbf{B}\mathbf{x},\mathbf{x}\rangle\Big)$$ is $\phi(\boldsymbol{\lambda})=\exp(-1/2\langle\mathbf{B}^{-1}\boldsymbol{\lambda},\boldsymbol{\lambda}\rangle)$ where $\mathbf{B}^{-1}$ is the covariance matrix of the vector $(X_{1},\dots,X_{n})$. 

	Thus the equality (\ref{prodcaratteristiche}) allows to determinate the covariance matrix of the vector $\bar{\mathbf{S}}^{1}(\boldsymbol{\sigma})$ taking off the matrix $\mathbf{A}^{-1}$ from $-\boldsymbol{\widetilde{\mathcal{H}}}_{f}^{-1}$. By calculus it is easy to verify that $-\boldsymbol{\widetilde{\mathcal{H}}}_{f}^{-1}-\mathbf{A}^{-1}=\boldsymbol{\widetilde{\chi}}$. 

	To complete the proof we have to show that the matrix $\boldsymbol{\widetilde{\chi}}$ is positive define. Considering the convex function
	\begin{align}\label{funzione.fi.multi}
	\Phi(x_{1},\dots,x_{n})&=\sum_{l=1}^{n}\alpha_{l}\ln\bigg(\!\cosh\bigg(\sum_{s=1}^{n}\alpha_{s}J_{ls}x_{s}+h_{l}\bigg)\bigg)\nonumber\\
	&=f(x_{1},\dots,x_{n})+\frac{1}{2}\sum_{l, s=1}^{n}\alpha_{l}\alpha_{s}J_{ls}x_{l}x_{s}
	\end{align}
	we can write $\boldsymbol{\mathcal{H}}_{f}=\boldsymbol{\mathcal{H}}_{\Phi}-\mathbf{D}_{\boldsymbol{\alpha}}\mathbf{D}_{\boldsymbol{\alpha}}\mathbf{J}\mathbf{D}_{\boldsymbol{\alpha}}\mathbf{D}_{\boldsymbol{\alpha}}$ where $\boldsymbol{\mathcal{H}}_{\Phi}$ is the Hessian matrix of the function $\Phi$. Since $\mathbf{A}=\mathbf{D}_{\boldsymbol{\alpha}}\mathbf{J}\mathbf{D}_{\boldsymbol{\alpha}}$ we get  $-\boldsymbol{\widetilde{\mathcal{H}}}_{f}=\mathbf{A}-\boldsymbol{\widetilde{\mathcal{H}}}_{\Phi}$ where $\boldsymbol{\widetilde{\mathcal{H}}}_{\Phi}=\mathbf{D}_{\alpha}^{-1}\boldsymbol{\mathcal{H}}_{\Phi}\mathbf{D}_{\alpha}^{-1}$. Multiplying $\boldsymbol{\widetilde{\chi}}$ by the positive define matrix $-\boldsymbol{\widetilde{\mathcal{H}}}_{f}$ we obtain
\begin{equation}\label{convessita.fi.multi}
 -\boldsymbol{\widetilde{\mathcal{H}}}_{f}(\boldsymbol{\mu})\boldsymbol{\widetilde{\chi}}=\Big(\mathbf{A}-\boldsymbol{\widetilde{\mathcal{H}}}_{\Phi}(\boldsymbol{\mu})\Big)\Big(\Big(\mathbf{A}-\boldsymbol{\widetilde{\mathcal{H}}}_{\Phi}(\boldsymbol{\mu})\Big)^{-1}-\mathbf{A}^{-1}\Big)=\boldsymbol{\widetilde{\mathcal{H}}}_{\Phi}(\boldsymbol{\mu})\mathbf{A}^{-1}.
\end{equation}
Since all matrices involved in (\ref{convessita.fi.multi}) are symmetric and $\boldsymbol{\widetilde{\mathcal{H}}}_{\Phi}(\boldsymbol{\mu})\mathbf{A}^{-1}$ is positive define it follows that also $\boldsymbol{\widetilde{\chi}}$ is positive define.
Hence the random vector (\ref{vettore.S.1}), converges to a multivariate Gaussian which covariance matrix is $\boldsymbol{\widetilde{\chi}}$.
\qed\\

In order to prove theorem \ref{teo.multi.2} we introduce the conditional joint distribution of a configuration $\boldsymbol{\sigma}$ conditioned on the event $(m_{1}(\boldsymbol{\sigma}),\dots,m_{n}(\boldsymbol{\sigma}))\in B(\boldsymbol{\mu},d)$
\begin{equation}\label{BG.condizionata}
 P_{N,\mathbf{J},\mathbf{h},d}\{\boldsymbol{\sigma}\}=\frac{1}{Z_{N}(\mathbf{J},\mathbf{h},d)}\exp\Big(\frac{1}{2N}\langle\mathbf{J}\mathbf{s},\mathbf{s}\rangle+\langle\mathbf{h},\mathbf{s}\rangle\Big)1_{B(\boldsymbol{\mu},d)}\Big(\frac{\mathbf{s}}{\mathbf{N}}\Big)d\nu_{S}(\mathbf{s})
\end{equation}
where $1_{B(\boldsymbol{\mu},d)}$ is the indicator function of the ball $B(\boldsymbol{\mu},d)$, 
$\nu_{S}$ denotes the distribution of the random vector $(S_{1}(\boldsymbol{\sigma}),\dots,S_{n}(\boldsymbol{\sigma}))$ on $(\mathbb{R}^{N},\prod_{i=1}^{N}\rho(\sigma_{i}))$ and $Z_{N}(\mathbf{J},\mathbf{h},d)$ is the normalizing constant.

We also need the following results:
\begin{lemma}\label{lemma.multi.5}
Let $P_{N,\mathbf{J},\mathbf{h},\delta}\{\boldsymbol{\sigma}\}$ be the joint distribution of $\boldsymbol{\sigma}=(\sigma_{1},\dots,\sigma_{N})$. Let $V_{\gamma}$ be the random vector
\begin{equation}\label{vettorone}
	\bigg(\dfrac{W_{1}}{(N_{1})^{1/2-\gamma}},\dots,\dfrac{W_{n}}{(N_{n})^{1/2-\gamma}}\bigg)+\bigg(\dfrac{S_{1}(\boldsymbol{\sigma})-N_{1}m_{1}}{(N_{1})^{1-\gamma}},\dots, \dfrac{S_{n}(\boldsymbol{\sigma})-N_{n}m_{n}}{(N_{n})^{1-\gamma}}\bigg)
	\end{equation}
where $\mathbf{W}\sim N(0,\mathbf{A}^{-1})$ and $\mathbf{A}=\mathbf{D}_{\boldsymbol{\alpha}}\mathbf{J}\mathbf{D}_{\boldsymbol{\alpha}}$ is a positive define matrix.\\ 
Then
\begin{equation}\label{tesi.lemma.multi.5}
\langle e^{\langle\mathbf{t},\mathbf{V}_{\gamma}\rangle}\rangle_{BG_{d}}=
\dfrac{e^{-\langle\mathbf{t},\mathbf{N}^{\gamma}\boldsymbol{\mu}\rangle}\displaystyle{\int_{\mathbb{R}^{n}}}\exp\Big(-\frac{N}{2}\langle\widetilde{\mathbf{J}}\mathbf{x},\mathbf{x}\rangle+\langle\mathbf{t},\mathbf{N}^{\gamma}\mathbf{x}\rangle\Big)I_{N}(\mathbf{x},\boldsymbol{\mu},d)d\mathbf{x}}{\displaystyle{\int_{\mathbb{R}^{n}}}\exp\Big(-\frac{N}{2}\langle\widetilde{\mathbf{J}}\mathbf{x},\mathbf{x}\rangle+\langle\mathbf{t},\mathbf{N}^{\gamma}\mathbf{x}\rangle\Big)I_{N}(\mathbf{x},\boldsymbol{\mu},d)d\mathbf{x}}\end{equation}
where $\langle\cdot\rangle_{BG_{d}}$ denotes the expectation value with respect to the conditional distribution (\ref{BG.condizionata}) and
\begin{equation}\label{I.grande}
 I_{N}(\mathbf{x},\boldsymbol{\mu},d)=\int_{\{\frac{\mathbf{s}}{\mathbf{N}}\in B(\boldsymbol{\mu},d)\}}\exp(\langle\mathbf{J}\boldsymbol{\alpha}\mathbf{x},\mathbf{s}\rangle+\langle\mathbf{h},\mathbf{s}\rangle)d\nu_{S}(\mathbf{s}).
\end{equation}
\end{lemma}

\begin{proof}
By following the same proof of lemma \ref{lemma.multi.2} we have that the distribution of the random vector $V_{\gamma}$ is given by
\begin{equation*}
 \dfrac{\exp\Big(-\frac{N}{2}\Big\langle\widetilde{\mathbf{J}}\Big(\boldsymbol{\mu}+\dfrac{\mathbf{x}}{\mathbf{N}^{\gamma}}\Big),\boldsymbol{\mu}+\dfrac{\mathbf{x}}{\mathbf{N}^{\gamma}}\Big\rangle\Big)I_{N}(\boldsymbol{\mu}+\frac{\mathbf{x}}{\mathbf{N}^{\gamma}},\boldsymbol{\mu},d)d\mathbf{x}}{\displaystyle{\int_{\mathbb{R}^{n}}}\exp\Big(-\frac{N}{2}\Big\langle\widetilde{\mathbf{J}}\Big(\boldsymbol{\mu}+\frac{\mathbf{x}}{\mathbf{N}^{\gamma}}\Big),\boldsymbol{\mu}+\frac{\mathbf{x}}{\mathbf{N}^{\gamma}}\Big\rangle\Big)d\mathbf{x}}
\end{equation*}
where
\begin{equation*}
I_{N}\Big(\boldsymbol{\mu}+\frac{\mathbf{x}}{\mathbf{N}^{\gamma}},\boldsymbol{\mu},d\Big)=\int_{\{\frac{\mathbf{s}}{\mathbf{N}}\in B(\boldsymbol{\mu},d)\}}\!\!\!\exp\Big(\Big\langle \mathbf{J}\boldsymbol{\alpha}\Big(\boldsymbol{\mu}+\frac{\mathbf{x}}{\mathbf{N}^{\gamma}}\Big),\mathbf{s}\Big\rangle+\langle \mathbf{h}, \mathbf{s}\rangle\Big)d\nu_{S}(\mathbf{s}).
\end{equation*}
Thus 
\begin{align*}
\langle e^{\langle\mathbf{t},\mathbf{V}_{\gamma}\rangle}&\rangle_{BG_{d}}\nonumber\\
&=\dfrac{\displaystyle{\int_{\mathbb{R}^{n}}}\!\!e^{\langle\mathbf{t},\mathbf{x}\rangle}\exp\Big(\!-\frac{N}{2}\Big\langle\widetilde{\mathbf{J}}\Big(\boldsymbol{\mu}+\frac{\mathbf{x}}{\mathbf{N}^{\gamma}}\Big),\boldsymbol{\mu}+\frac{\mathbf{x}}{\mathbf{N}^{\gamma}}\Big\rangle\Big)I_{N}\Big(\boldsymbol{\mu}+\frac{\mathbf{x}}{\mathbf{N}^{\gamma}},\boldsymbol{\mu},d\Big)d\mathbf{x}}{\displaystyle{\int_{\mathbb{R}^{n}}}\!\!\exp\Big(\!-\frac{N}{2}\Big\langle\widetilde{\mathbf{J}}\Big(\boldsymbol{\mu}+\frac{\mathbf{x}}{\mathbf{N}^{\gamma}}\Big),\boldsymbol{\mu}+\frac{\mathbf{x}}{\mathbf{N}^{\gamma}}\Big\rangle\Big)d\mathbf{x}}.
\end{align*}
Making the change of variable
\begin{equation*}
u_{l}=\mu_{l}+\dfrac{x_{l}}{N_{l}^{\gamma}}\quad l=1,\dots,n
\end{equation*}
the statement (\ref{tesi.lemma.multi.5}) holds.
\end{proof}
\begin{lemma}\label{ultimo.lemma}
Let $\boldsymbol{\mu}^{1},\dots,\boldsymbol{\mu}^{P}$ be global maxima points of the function $f$ given by (\ref{funzionale.pressione.multi.replica}). Let $k_{1},\dots,k_{p}$ be their homogeneous type. Define
\begin{multline}\label{K.grande}
K_{N}(\mathbf{t},\delta,\boldsymbol{\mu},\theta,k)=\Big(\prod_{l=1}^{n}N_{l}\Big)^{\frac{1}{2k}}\exp(-NF-\langle\mathbf{t},\mathbf{N}^{1/2k}\boldsymbol{\mu}\rangle)\\
\times\int_{B(\boldsymbol{\mu},\delta)}\!\!\!\exp\Big(-\frac{N}{2}\langle\widetilde{\mathbf{J}}\mathbf{x},\mathbf{x}\rangle+\langle\mathbf{t},\mathbf{N}^{\gamma}\mathbf{x}\rangle\Big)I_{N}^{c}(\mathbf{x},\boldsymbol{\mu},\theta)d\mathbf{x}
\end{multline}
where
\begin{equation}\label{I.grande.complementare}
I_{N}^{c}(\mathbf{x},\boldsymbol{\mu},\theta)=\int_{\{\frac{\mathbf{s}}{\mathbf{N}}\in B^{c}(\boldsymbol{\mu},\theta)\}}\exp(\langle\mathbf{J}\boldsymbol{\alpha}\mathbf{x},\mathbf{s}\rangle+\langle\mathbf{h},\mathbf{s}\rangle)d\nu_{S}(\mathbf{s}).
\end{equation}\\
For any $\theta>0$ there exists $\epsilon>0$ such that for each $p=1,\dots,P$
\begin{equation}\label{tesi.lemma.funzione.k}
 K_{N}(\mathbf{t},\delta_{p},\boldsymbol{\mu}^{p},\theta,k_{p})=O(e^{-N\epsilon}) \quad\text{as}\;N\rightarrow\infty.
\end{equation}
\end{lemma}
\begin{proof}
We prove the theorem for $p=1$. The proofs for other $p$ are similar. We observe that
\begin{align*}
 \bigg\{\bigg(\frac{S_{1}(\boldsymbol{\sigma})}{N_{1}},\dots,&\frac{S_{n}(\boldsymbol{\sigma})}{N_{n}}\bigg)\in B^{c}(\boldsymbol{\mu}^{1},\theta)\bigg\}\nonumber\\
&\subset\bigcup_{l=1}^{n}\bigg\{\bigg|\frac{S_{l}(\boldsymbol{\sigma})}{N_{l}}-\mu_{l}^{1}\bigg|\geq \bar{\theta}\bigg\}\nonumber\\
&=\bigcup_{l=1}^{n}\bigg(\bigg\{\bigg|\frac{S_{l}(\boldsymbol{\sigma})}{N_{l}}\bigg|\leq\mu_{l}^{1}-\bar{\theta}\bigg\}\cup\bigg\{\bigg|\frac{S_{l}(\boldsymbol{\sigma})}{N_{l}}\bigg|\geq\mu_{l}^{1}+\bar{\theta}\bigg\}\bigg)
\end{align*}
where $\bar{\theta}=\sqrt{n}\theta$. Thus
\begin{multline}\label{momento.te}
 I_{N}^{c}(\mathbf{x},\boldsymbol{\mu}^{1},\theta)\leq\sum_{l=1}^{n}\bigg(\int_{\{-\frac{s_{l}}{N_{l}}\geq-\mu_{l}^{1}+\bar{\theta}\}}\exp(\langle\mathbf{J}\boldsymbol{\alpha}\mathbf{x},\mathbf{s}\rangle+\langle\mathbf{h},\mathbf{s}\rangle)d\nu_{S}(\mathbf{s})\\
+\quad\int_{\{\frac{s_{l}}{N_{l}}\geq\mu_{l}^{1}+\bar{\theta}\}}\exp(\langle\mathbf{J}\boldsymbol{\alpha}\mathbf{x},\mathbf{s}\rangle+\langle\mathbf{h},\mathbf{s}\rangle)d\nu_{S}(\mathbf{s})\bigg).
\end{multline}
Consider one of the integrals of expression (\ref{momento.te})
\begin{align*}
\int_{\{\frac{s_{l}}{N_{l}}\geq\mu_{l}^{1}+\bar{\theta}\}}&\exp(\langle\mathbf{J}\boldsymbol{\alpha}\mathbf{x},\mathbf{s}\rangle+\langle\mathbf{h},\mathbf{s}\rangle)d\nu_{S}(\mathbf{s})\nonumber\\
&=\prod_{l=2}^{n}\cosh\Big(\sum_{q=1}^{n}\alpha_{q}J_{lq}x_{q}+h_{l}\Big)^{N_{l}}\nonumber\\
&\quad\times\int_{\{S_{1}(\boldsymbol{\sigma})\geq N_{1}(\mu_{1}^{1}+\bar{\theta})\}}\exp\Big(S_{1}(\boldsymbol{\sigma})\Big(\sum_{q=1}^{n}\alpha_{q}J_{1q}x_{q}+h_{1}\Big)\Big)\prod_{i\in P_{1}}d\rho(\sigma_{i}).
\end{align*}\\
By Chebishev's inequality for any $\tau>0$
\begin{align}\label{bamba1}
 &\int_{\{S_{1}(\boldsymbol{\sigma})\geq N_{1}(\mu_{1}^{1}+\bar{\theta})\}}\!\!\!\!\exp\Big(S_{1}(\boldsymbol{\sigma})\Big(\sum_{q=1}^{n}\alpha_{q}J_{1q}x_{q}+h_{1}\Big)\Big)\prod_{i\in P_{1}}d\rho(\sigma_{i})\nonumber\\
 &\leq\exp(-\alpha_{1}J_{11}\tau N_{1}(\mu_{1}^{1}+\bar{\theta}))\nonumber\\
 &\quad\times\int_{\mathbb{R}^{N_{1}}}\exp(\alpha_{1}J_{11}\tau\sum_{i\in P_{1}}\sigma_{i})\exp(\sum_{i\in P_{1}}\sigma_{i}(\alpha_{1}J_{11}x_{1}+\alpha_{2}J_{12}x_{2}+h_{1}))\nonumber\\
&=\exp(N_{1}(-\alpha_{1}J_{11}\tau(\mu_{1}^{1}+\bar{\theta})+\ln\cosh(\alpha_{1}J_{11}(x_{1}+\tau)+\alpha_{2}J_{12}x_{2}+h_{1}))).
\end{align}
By the mean field equations (\ref{campomedio.multi}) we have:
\begin{equation*}
 \frac{\partial}{\partial x_{1}}\Big(\ln\cosh(\alpha_{1}J_{11}x_{1}+\alpha_{2}J_{12}x_{2}+h_{1})\Big)(\boldsymbol{\mu}^{1})=\alpha_{1}J_{11}\mu_{1}^{1}.
\end{equation*}
Thus we can choose $\delta>0$ and $\tau>0$ sufficiently small such that $\delta<\delta_{1}$ and 
\begin{multline}\label{bamba2}
\ln\cosh(\alpha_{1}J_{11}(x_{1}+\tau)+\alpha_{2}J_{12}x_{2}+h_{1})\\\leq\ln\cosh(\alpha_{1}J_{11}x_{1}+\alpha_{2}J_{12}x_{2}+h_{1})+\alpha_{1}J_{11}\mu_{1}^{1}\tau+\frac{1}{2}\alpha_{1}J_{11}\tau\bar{\theta}
\end{multline}
for each $\mathbf{x}\in B(\boldsymbol{\mu}^{1},\delta)$. The other integrals in (\ref{momento.te}) are handled in a similar way. 
At last applying the bounds (\ref{bamba1}) and (\ref{bamba2}) to (\ref{momento.te}), for all $\mathbf{x}\in B(\boldsymbol{\mu}^{1},\delta)$ we obtain:
\begin{equation*}
I_{N}^{c}(\mathbf{x},\boldsymbol{\mu}^{1},\theta)\leq 2n\exp\Big(N\Big(-\frac{1}{2}\alpha_{1}^{2}J_{11}\tau\bar{\theta}+\Phi(\mathbf{x})\Big)\Big)
\end{equation*}
where $\Phi$ is given by (\ref{funzione.fi.multi}). Hence
\begin{align}\label{pera}
K_{N}(\mathbf{t},\delta,\boldsymbol{\mu}^{1},\theta,k_{1})&\leq 2n (N_{l})^{1/2k_{1}}\exp\Big(-NF-\langle\mathbf{t},\mathbf{N}^{1/2k_{1}}\boldsymbol{\mu}^{1}\rangle-\frac{N}{2}\alpha_{1}^{2}J_{11}\tau\bar{\theta}\Big)\nonumber\\
&\quad\times\int_{B(\boldsymbol{\mu}^{1},\delta)}\exp(Nf(\mathbf{x})+\langle\mathbf{t},\mathbf{N}^{1/2k_{1}}\mathbf{x}\rangle)d\mathbf{x}.
\end{align}
Applying the lemma \ref{lemma.multi.4} to the right-hand side of (\ref{pera}) we obtain 
\begin{equation}\label{stima1}
K_{N}(\mathbf{t},\delta,\boldsymbol{\mu}^{1},\theta,k_{1})=O(e^{-\frac{N}{2}\alpha_{1}^{2}J_{11}\tau\bar{\theta}})\quad\quad\text{as}\; N\rightarrow\infty.
\end{equation}
We now bound
\begin{align*}
K_{N}(\mathbf{t},&\delta_{1},\boldsymbol{\mu}^{1},\theta,k_{1})-K_{N}(\mathbf{t},\delta,\boldsymbol{\mu}^{1},\theta,k_{1})\nonumber\\
&= (N_{l})^{1/2k_{1}}\exp(-NF-\langle\mathbf{t},\mathbf{N}^{1/2k_{1}}\boldsymbol{\mu}^{1}\rangle)
\nonumber\\
&\quad\times\int_{B(\boldsymbol{\mu}^{1},\delta_{1})\smallsetminus B(\boldsymbol{\mu}^{1},d)}\!\!\!\!\exp\Big(-\frac{N}{2}\langle\widetilde{\mathbf{J}}\mathbf{x},\mathbf{x}\rangle+\langle\mathbf{t},\mathbf{N}^{1/2k_{1}}\mathbf{x}\rangle\Big)I_{N}^{c}(\mathbf{x},\boldsymbol{\mu}^{1},\theta)d\mathbf{x}.
\end{align*}\\
Considering the functions $I_{N}$ and $I_{N}^{c}$ given respectively by (\ref{I.grande}) and (\ref{I.grande.complementare}), it is easy to verify that:
\begin{align}\label{tortainarrivo}
I_{N}(\mathbf{x},\boldsymbol{\mu}^{1},\theta)+I_{N}^{c}(\mathbf{x},\boldsymbol{\mu}^{1},\theta)&=\exp(N\Phi(\mathbf{x})).
\end{align}
Thus:
\begin{equation*}
I_{N}^{c}(\mathbf{x},\boldsymbol{\mu}^{1},\theta)\leq\exp(N\Phi(\mathbf{x})).
\end{equation*}\\
By definition of the function $\Phi$ we get
\begin{align*}
K_{N}(\mathbf{t},\delta_{1},&\boldsymbol{\mu}^{1},\theta,k_{1})-K_{N}(\mathbf{t},\delta,\boldsymbol{\mu}^{1},\theta,k_{1})\nonumber\\
&\leq (N_{l})^{1/2k_{1}}e^{-NF-\langle\mathbf{t},\mathbf{N}^{1/2k_{1}}\boldsymbol{\mu}^{1}\rangle}\nonumber\\
&\quad\times\int_{B(\boldsymbol{\mu}^{1},\delta_{1})\smallsetminus B(\boldsymbol{\mu}^{1},d)}\!\!\!\!\exp(Nf(\mathbf{x})+\langle\mathbf{t},\mathbf{N}^{1/2k_{1}}\mathbf{x}\rangle)d\mathbf{x}.
\end{align*}\\
Making the change of variable
\begin{equation*}
u_{l}=\mu_{l}+\dfrac{x_{l}}{N_{l}^{1/2k_{1}}}\quad l=1,\dots,n
\end{equation*}
we obtain
\begin{multline*}
K_{N}(\mathbf{t},\delta_{1},\boldsymbol{\mu}^{1},\theta,k_{1})-K_{N}(\mathbf{t},\delta,\boldsymbol{\mu}^{1},\theta,k_{1})\\\leq\int_{E}\exp\bigg(N\bigg(B\Big(\frac{\mathbf{u}}{\mathbf{N}^{1/2k_{1}}};\boldsymbol{\mu}^{1}\Big)\bigg)+\langle\mathbf{t},\mathbf{u}\rangle\bigg)d\mathbf{u}
\end{multline*}
where\\
\begin{equation*}
E=\Big\{\Big|\Big|\frac{\mathbf{u}}{\mathbf{N}^{1/2k_{1}}}\Big|\Big|<\delta_{1}\Big\}\setminus\Big\{\Big|\Big|\frac{\mathbf{u}}{\mathbf{N}^{1/2k_{1}}}\Big|\Big|<\delta\Big\}.
\end{equation*}\\
Observing that
\begin{equation*}
E\subset\Big\{\Big|\Big|\frac{\mathbf{u}}{\mathbf{N}^{1/2k_{1}}}\Big|\Big|<\delta\Big\}^{c}
\end{equation*}\\
it follows that as $N\rightarrow\infty$, for some $\epsilon_{0}>0$
\begin{equation}\label{stima2}
 K_{N}(\mathbf{t},\delta_{1},\boldsymbol{\mu}^{1},\theta,k_{1})-K_{N}(\mathbf{t},\delta,\boldsymbol{\mu}^{1},\theta,k_{1})=O(e^{N\epsilon_{0}})
\end{equation}
The statement (\ref{tesi.lemma.funzione.k}) follows by (\ref{stima1}) and (\ref{stima2}). 
\end{proof}
\vspace{0.5cm}
\noindent {\it Proof of Theorem \ref{teo.multi.2}}. We give the proof for $\boldsymbol{\mu}=\boldsymbol{\mu}^{1}$. The other global maximum points are handled identically. Fix $\mathbf{t}\in\mathbb{R}^{n}$ we choose the number $\delta_{p}>0$, $p=1,\dots,P$ according to lemma \ref{lemma.multi.4}. For all $\delta\in(0,\delta_{p}]$
\begin{multline}\label{iniziamente}
\lim_{N\rightarrow\infty}e^{-\langle\mathbf{t},\mathbf{N}^{1/2k_{p}}\boldsymbol{\mu}^{p}\rangle}\Big(\prod_{l=1}^{n}N_{l}\Big)^{1/2k_{p}}e^{-NF}\int_{B(\boldsymbol{\mu}^{p},\delta)}e^{Nf(\mathbf{x})+\langle\mathbf{t},\mathbf{N}^{1/2k_{p}}\mathbf{x}\rangle}d\mathbf{x}\\
=\int_{\mathbb{R}^{n}}\exp\bigg(f^{\boldsymbol{\mu}^{p}}_{2k_{p}}\Big(\frac{\mathbf{x}}{\boldsymbol{\alpha}^{1/2k_{p}}}\Big)+\langle\mathbf{t},\mathbf{x}\rangle\bigg)d\mathbf{x}
\end{multline}\\
Since $\mathbb{R}^{n}=B(\boldsymbol{\mu}^{1},\delta_{1})\cup B^{c}(\boldsymbol{\mu}^{1},\delta_{1})$ and $I_{N}(\mathbf{x},\boldsymbol{\mu}^{1},\delta)=e^{N\Phi(\mathbf{x})}-I_{N}^{c}(\mathbf{x},\boldsymbol{\mu}^{1},\delta)$ by lemma \ref{lemma.multi.5} after have
multiplied numerator and denominator by the term $(\prod_{l=1}^{n}N_{l})^{1/2k_{1}}e^{-NF}$ we obtain
\begin{multline*}
\langle e^{\langle\mathbf{t},\mathbf{V}_{1/2k_{1}}\rangle}\rangle_{BG_{d}}\\=
\dfrac{L_{N}(\mathbf{t},\delta_{1},\boldsymbol{\mu}^{1},k_{1})-K_{N}(\mathbf{t},\delta_{1},\boldsymbol{\mu}^{1},d,k_{1})+M_{N}(\mathbf{t},\delta_{1},\boldsymbol{\mu}^{1},d,k_{1})}{L_{N}(\mathbf{0},\delta_{1},\boldsymbol{\mu}^{1},k_{1})-K_{N}(\mathbf{0},\delta_{1},\boldsymbol{\mu}^{1},d,k_{1})+M_{N}(\mathbf{0},\delta_{1},\boldsymbol{\mu}^{1},d,k_{1})}
\end{multline*}
where the random vector $\mathbf{V}_{1/2k_{1}}$ is defined by (\ref{vettorone}), the function $K_{N}$ is defined by (\ref{K.grande}),
\begin{align*}
L_{N}(\mathbf{t},\delta_{1},\boldsymbol{\mu}^{1},k_{1})&=\Big(\prod_{l=1}^{n}N_{l}\Big)^{1/2k_{1}}\exp(-NF-\langle\mathbf{t},\mathbf{N}^{1/2k_{1}}\boldsymbol{\mu}^{1}\rangle)\nonumber\\
&\quad\times\int_{B(\boldsymbol{\mu}^{1},\delta_{1})}\!\!\!\!\exp\Big(\!-\frac{N}{2}\langle\widetilde{\mathbf{J}}\mathbf{x},\mathbf{x}\rangle+\langle\mathbf{t},\mathbf{N}^{1/2k_{1}}\mathbf{x}\rangle+N\Phi(\mathbf{x})\Big)d\mathbf{x}\nonumber\\
&=\Big(\prod_{l=1}^{n}N_{l}\Big)^{1/2k_{1}}\exp(-NF-\langle\mathbf{t},\mathbf{N}^{1/2k_{1}}\boldsymbol{\mu}^{1}\rangle)\nonumber\\
&\quad\times\int_{B^{c}(\boldsymbol{\mu}^{1},\delta_{1})}\exp(Nf(\mathbf{x})+\langle\mathbf{t},\mathbf{N}^{1/2k_{1}}\mathbf{x}\rangle)d\mathbf{x}
\end{align*}
and 
\begin{align*}
M_{N}(\mathbf{t},&\delta_{1},\boldsymbol{\mu}^{1},d,k_{1})\\
&=\Big(\prod_{l=1}^{n}N_{l}\Big)^{1/2k_{1}}\exp(-NF-\langle\mathbf{t},\mathbf{N}^{1/2k_{1}}\boldsymbol{\mu}^{1}\rangle)\nonumber\\
&\quad\times\int_{B(\boldsymbol{\mu}^{1},\delta_{1})}\!\!\!\!\!\exp\Big(-\frac{N}{2}\langle\widetilde{\mathbf{J}}\mathbf{x},\mathbf{x}\rangle+\langle\mathbf{t},\mathbf{N}^{1/2k_{1}}\mathbf{x}\rangle\Big)I_{N}(\mathbf{x},\boldsymbol{\mu}^{1},d)d\mathbf{x}
\end{align*}
By lemma \ref{ultimo.lemma} (with $p=1$ and $\theta=d$) there exists $\epsilon_{1}>0$ such that
\begin{equation}\label{cavolini}
 K_{N}(\mathbf{t},\delta_{1},\boldsymbol{\mu}^{1},d,k_{1})=O(e^{-N\epsilon_{1}}).
\end{equation}
Now we prove that there exists also $\epsilon_{2}>0$ such that
\begin{equation}\label{checavoli}
 M_{N}(\mathbf{t},\delta_{1},\boldsymbol{\mu}^{1},d,k_{1})=O(e^{-N\epsilon_{2}}).
\end{equation}
Define the closet set 
\begin{equation*}
 V=\mathbb{R}^{n}-\bigcup_{p=1}^{P}B(\boldsymbol{\mu}^{p},\delta_{p}).
\end{equation*}
Then 
\begin{equation*}
 B^{c}(\boldsymbol{\mu}^{1},\delta_{1})\subset V\cup\bigcup_{p=2}^{P}B(\boldsymbol{\mu}^{p},\delta_{p}).
\end{equation*}
Thus we can write
\begin{align}\label{micaela}
M_{N}(\mathbf{t},&\delta_{1},\boldsymbol{\mu}^{1},d,k_{1})\nonumber\\
&=\Big(\prod_{l=1}^{n}N_{l}\Big)^{1/2k_{1}}\exp(-NF-\langle\mathbf{t},\mathbf{N}^{1/2k_{1}}\boldsymbol{\mu}^{1}\rangle)\nonumber\\
&\quad\times\int_{V\cup\bigcup_{p=2}^{P}B(\boldsymbol{\mu}^{p},\delta_{p})}\!\!\!\!\!\exp\Big(-\frac{N}{2}\langle\widetilde{\mathbf{J}}\mathbf{x},\mathbf{x}\rangle+\langle\mathbf{t},\mathbf{N}^{1/2k_{1}}\mathbf{x}\rangle\Big)I_{N}(\mathbf{x},\boldsymbol{\mu}^{1},d)d\mathbf{x}.
\end{align}
Since $d<\bar{\delta}$ and $||\boldsymbol{\mu}^{1}-\boldsymbol{\mu}^{p}||\geq\bar{\delta}$ for $p=2,\dots,P$we have $B(\boldsymbol{\mu}^{1},\delta_{1})\subset B^{c}(\boldsymbol{\mu}^{p},\bar{\delta}-d)$ hence for each $\mathbf{x}\in\mathbb{R}^{n}$ and $p=2,\dots,P$
\begin{equation}\label{quasiletto}
 I_{N}(\mathbf{x},\boldsymbol{\mu}^{1},d)\leq I_{N}^{c}(\mathbf{x},\boldsymbol{\mu}^{p},\bar{\delta}-d).
\end{equation}
Moreover by (\ref{tortainarrivo})
\begin{equation}\label{letto}
I_{N}(\mathbf{x},\boldsymbol{\mu}^{1},d)\leq\exp(N\Phi(\mathbf{x})).
\end{equation}
Using (\ref{quasiletto}) in the integrals over $B(\boldsymbol{\mu}^{p},\delta_{p})$, $p=2,\dots,P$ and (\ref{letto}) in the integral over $V$ of the expression (\ref{micaela}), we obtain
\begin{align}\label{ultimissime cose}
M_{N}(\mathbf{t},&\delta_{1},\boldsymbol{\mu}^{1},d,k_{1})\nonumber\\
&=\Big(\prod_{l=1}^{n}N_{l}\Big)^{1/2k_{1}}\exp(-NF-\langle\mathbf{t},\mathbf{N}^{1/2k_{1}}\boldsymbol{\mu}^{1}\rangle)\nonumber\\
&\quad\times\int_{V}\exp(Nf(\mathbf{x})+\langle\mathbf{t},\mathbf{N}^{1/2k_{1}}\mathbf{x}\rangle)d\mathbf{x}\nonumber\\
&\quad+\sum_{p=2}^{P}\exp\Big(\langle\mathbf{t},\mathbf{N}^{1/2k_{1}}(\boldsymbol{\mu}^{p}-\boldsymbol{\mu}^{1})\rangle\Big)K_{N}(\mathbf{t},\delta_{p},\boldsymbol{\mu}^{p},\bar{\delta}-d,k_{1}).
\end{align}
Applying the lemma \ref{lemma.multi.3} to the first term of the right-hand side of (\ref{ultimissime cose}) and lemma \ref{ultimo.lemma} to each term of the sum, the result (\ref{checavoli}) holds.

By (\ref{cavolini}) and (\ref{checavoli}) it follows
\begin{align*}
 \lim_{N\rightarrow\infty}\langle e^{\langle\mathbf{t},\mathbf{V}_{1/2k_{1}}\rangle}\rangle_{BG_{d}}&=\dfrac{L_{N}(\mathbf{t},\delta_{1},\boldsymbol{\mu}^{1},k_{1})}{L_{N}(\mathbf{0},\delta_{1},\boldsymbol{\mu}^{1},k_{1})}\nonumber\\\nonumber\\
&=\dfrac{\displaystyle{\int_{\mathbb{R}^{n}}}\exp\bigg(f_{2k_{1}}^{\boldsymbol{\mu}^{1}}\Big(\frac{\mathbf{x}}{\boldsymbol{\alpha}^{1/2k_{p}}}\Big)+\langle\mathbf{t},\mathbf{x}\rangle\bigg)d\mathbf{x}}{\displaystyle{\int_{\mathbb{R}^{n}}}\exp\bigg(f_{2k_{p}}^{\boldsymbol{\mu}^{1}}\Big(\frac{\mathbf{x}}{\boldsymbol{\alpha}^{1/2k_{p}}}\Big)\bigg)d\mathbf{x}}
\end{align*}
where in the last identity we use (\ref{iniziamente}). By the assumption on the random vector $W$ the theorem \ref{teo.multi.2} is proved. 
\clearpage{\pagestyle{empty}\cleardoublepage}
\chapter[The Inverse Problem]{The Inverse Problem}
\lhead[\fancyplain{}{\bfseries\thepage}]{\fancyplain{}{\bfseries\rightmark}}
So far we have seen that once the probability distribution of a configuration of spins is assigned it is possible to compute the Gibbs state in the thermodynamic limit of any observable of interest and to related them to the parameters of the model.

Now we want to analyze the inverse problem, that is the determination of the parameters of a Boltzmann-Gibbs's distribution knowing the average value and the correlations of the spins at the equilibrium.\\
 
We solve the inverse problem for the Curie-Weiss model as well as for its multi-species generalization with naive-mean-field method. This technique is used in \cite{tanaka1998mean,roudi2009ising,roudi2009statistical} to solve the inverse problem for the Ising model.
\section{Inverse problem for the curie-Weiss model}
We recall that the model of Curie-Weiss is defined by the Hamiltonian
\begin{equation*}
H_{N}(\boldsymbol{\sigma})=-\frac{J}{2N}\sum_{i,j=1}^{N}\sigma_{i}\sigma_{j}-h\sum_{i=1}^{N}\sigma_{i}
\end{equation*}
and the distribution 
\begin{equation}\label{misura.BG.replica}
P_{N,J,h}\{\boldsymbol{\sigma}\}=\frac{\exp(- H_{N}(\boldsymbol{\sigma}))}{Z_{N}(J,h)}\prod_{i=1}^{N}d\rho(\sigma_{i})
\end{equation}
where
\begin{equation}\label{ro.curie.replica}
\rho(x)=\frac{1}{2}\Big(\delta(x-1)+\delta(x+1)\Big).
\end{equation}\\
By (\ref{ro.curie.replica}) and the definition of magnetization the distribution (\ref{misura.BG.replica}) can be written as
\begin{align}\label{BG.curie}
P_{N,J,h}\{\boldsymbol{\sigma}\}&=\dfrac{\exp(-H_{N}(\boldsymbol{\sigma}))}{\sum\limits_{\boldsymbol{\sigma}\in\Omega_{N}}\exp(-H_{N}(\boldsymbol{\sigma}))}\nonumber\\
&=\dfrac{\exp\bigg(N\Big(\frac{J}{2}m_{N}^{2}(\boldsymbol{\sigma})+hm_{N}(\boldsymbol{\sigma})\Big)\bigg)}{\sum\limits_{\boldsymbol{\sigma}\in\Omega_{N}}\exp\bigg(N\Big(\frac{J}{2}m_{N}^{2}(\boldsymbol{\sigma})+hm_{N}(\boldsymbol{\sigma})\Big)\bigg)}
\end{align}\\
where $\Omega_{N}=\{-1,1\}^{N}$.

Since $\langle m_{N}(\boldsymbol{\sigma})\rangle_{BG}=\langle\sigma_{i}\rangle_{BG}$ for each $i=1,\dots,N$ to solve the inverse problem we have to know the average value and the variance of the magnetization with respect to the Boltzman-Gibbs measure (\ref{BG.curie}).\\ 

As $h\neq 0$ and $J>0$ or $h=0$ and $J<1$, by theorem \ref{teorema.1} the following equality holds
\begin{equation}\label{identita}
\lim_{N\rightarrow\infty}\langle m_{N}(\boldsymbol{\sigma})\rangle_{BG}=\mu
\end{equation}\\
where $\mu$ is a solution of the mean-field equation (\ref{campo.medio}) that maximaze the function $f$ given by (\ref{funzione.f.curie}). Differentiating the identity (\ref{identita}) with respect to $h$ we obtain:
\begin{equation*}
\lim_{N\rightarrow\infty}\frac{\partial}{\partial h}\langle m_{N}(\boldsymbol{\sigma})\rangle_{BG}=\chi
\end{equation*}\newpage
\noindent where by definition $\chi=\partial \mu/\partial h$ and\\
\begin{align}\label{chi.1}
\frac{\partial}{\partial h}\langle m_{N}(\boldsymbol{\sigma})\rangle_{BG} &=\frac{\partial}{\partial h}\bigg(\dfrac{\sum_{\boldsymbol{\sigma}\in\Omega_{N}}m_{N}(\boldsymbol{\sigma})\exp(-H_{N}(\boldsymbol{\sigma}))}{\sum_{\boldsymbol{\sigma}\in\Omega_{N}}\exp(- H_{N}(\boldsymbol{\sigma}))}\bigg)\nonumber\\\nonumber\\
&=N\dfrac{\sum_{\boldsymbol{\sigma}\in\Omega_{N}}m_{N}^{2}(\boldsymbol{\sigma})\exp(-H_{N}(\boldsymbol{\sigma}))}{\sum_{\boldsymbol{\sigma}\in\Omega_{N}}\exp(- H_{N}(\boldsymbol{\sigma}))}\nonumber\\ &\quad -N\bigg(\dfrac{\sum_{\boldsymbol{\sigma}\in\Omega_{N}}m_{N}(\boldsymbol{\sigma})\exp(-H_{N}(\boldsymbol{\sigma}))}{\sum_{\boldsymbol{\sigma}\in\Omega_{N}}\exp(- H_{N}(\boldsymbol{\sigma}))}\bigg)^{2}\nonumber\\\nonumber\\
&=N\Big(\langle m_{N}^{2}(\boldsymbol{\sigma})\rangle_{BG}-\langle m_{N}(\boldsymbol{\sigma})\rangle_{BG}^{2}\Big).
\end{align}\\
On the other hand we have seen that 
\begin{equation}\label{chi.2}
 \chi=\dfrac{1-\mu^{2}}{1-J(1-\mu^{2})}.
\end{equation}
Thus by (\ref{chi.1}) and (\ref{chi.2}) in the thermodynamic limit we can compute the parameter $J$ from the average value and the variance of the magnetization:\\
\begin{equation}\label{J.naive}
 J=\frac{1}{1-\langle m_{N}(\boldsymbol{\sigma})\rangle_{BG}^{2}}-\frac{1}{N\Big(\langle m_{N}^{2}(\boldsymbol{\sigma})\rangle_{BG}-\langle m_{N}(\boldsymbol{\sigma})\rangle_{BG}^{2}\Big)}.
\end{equation}\\
Determinate $J$, we obtain an espression for the field $h$, in the thermodynamic limit, inverting the mean-field equation 
\begin{equation}\label{h.naive}
h=\tanh^{-1}(\langle m_{N}(\boldsymbol{\sigma})\rangle_{BG})-J\langle m_{N}(\boldsymbol{\sigma})\rangle_{BG}
\end{equation}
where $J$ is given by (\ref{J.naive}).
This solve the inverse problem for the Curie-Weiss model as $h\neq 0$ and $J>0$ or $h=0$ and $J<1$.\\

We have seen that if $h=0$ and $J>1$ the function $f$ reaches the maximum in two different points $\pm\mu_{0}$ both of maximal type. Thus the identity (\ref{identita}) is not verified. In spite of it we can solve the inverse problem observing that there exixts $\epsilon>0$ such that whenever $m_{N}(\boldsymbol{\sigma})\in(\pm\mu_{0}-\epsilon,\pm\mu_{0}+\epsilon)$ the following holds
\begin{equation}\label{identita.2}
\lim_{N\rightarrow\infty}\langle m_{N}(\boldsymbol{\sigma})\rangle_{BG}=\pm\mu_{0}
\end{equation}\\
and then applying to (\ref{identita.2}) the same procedure we have shown.\\

The inverse problem is very useful when we deal with the phenomenological test of the Curie-Weiss model. In this case we would like to estimate the unknown parameters $J$ and $h$ of the distribution of Boltzmann-Gibbs from a sample of empirical data. One possibility to reach this aim is to use the maximum likelihood estimation. This method determinates the parameters that maximize the probability to obtain the given sample under the condition that the sample is the realization of random variables indipendent and identically distributed.

Since the spins of a single configuration $\boldsymbol{\sigma}$ are dependent, to apply the maximum likelihood estimation we have to consider a sample constituted by $M$ configurations of spins $\boldsymbol{\sigma}^{(1)},\dots,\boldsymbol{\sigma}^{(M)}$ indipendent and identically distribuited according to the Boltzmann-Gibbs measure (\ref{BG.curie}). Our aim is to maximize with respect to $J$ and $h$ the joint probability of the sample\\
\begin{equation*}
 L(J,h)=P_{N,J,h}\Big\{\boldsymbol{\sigma}^{(1)},\dots,\boldsymbol{\sigma}^{(M)}\Big\}.
\end{equation*}\\
The function $L(J,h)$ is called maximum likelihood function. Since the sample is indipendent and identically distributed we have\\
\begin{align*}
L(J,h)&=\prod_{m=1}^{M}P_{N,J,h}\Big\{\boldsymbol{\sigma}^{(m)}\Big\}\nonumber\\
&=\prod_{m=1}^{M}\frac{\exp(- H_{N}(\boldsymbol{\sigma}^{(m)}))}{\sum_{\boldsymbol{\sigma}\in\Omega_{N}}\exp(- H_{N}(\boldsymbol{\sigma}))}.
\end{align*}\\
To maximize the function $L(J,h)$ we should compute the derivative of a product. Since a function and its logarithm reach the maximum in the same point we consider the logarithm of the maximum likelihood function
\begin{equation*}
\ln L(J,h)=\sum_{m=1}^{M}\bigg(-H_{N}(\boldsymbol{\sigma}^{(m)})-\sum_{\boldsymbol{\sigma}\in\Omega_{N}}\exp(- H_{N}(\boldsymbol{\sigma})\bigg)
\end{equation*}
and differentiate it with respect to $h$ and $J$:
\begin{align*}
\frac{\partial\ln L(J,h)}{\partial h}&=\sum_{m=1}^{M}\bigg(Nm_{N}(\boldsymbol{\sigma}^{(m)})-N\dfrac{\sum_{\boldsymbol{\sigma}\in\Omega_{N}}\exp(-H_{N}(\boldsymbol{\sigma}))m_{N}(\boldsymbol{\sigma})}{\sum_{\boldsymbol{\sigma}\in\Omega_{N}}\exp(- H_{N}(\boldsymbol{\sigma}))}\bigg)\nonumber\\
&=N\sum_{m=1}^{M}\bigg(m_{N}(\boldsymbol{\sigma}^{(m)})-\langle m_{N}(\boldsymbol{\sigma})\rangle_{BG}\bigg)\\\nonumber\\
\frac{\partial\ln L(J,h)}{\partial J}&=\sum_{m=1}^{M}\bigg(\frac{N}{2}m_{N}^{2}(\boldsymbol{\sigma}^{(m)})-\frac{N}{2}\dfrac{\sum_{\boldsymbol{\sigma}}\exp(-H_{N}(\boldsymbol{\sigma}))m_{N}^{2}(\boldsymbol{\sigma})}{\sum_{\boldsymbol{\sigma}}\exp(- H_{N}(\boldsymbol{\sigma}))}\bigg)\nonumber\\
&=\sum_{m=1}^{M}\bigg(\frac{N}{2}m_{N}^{2}(\boldsymbol{\sigma}^{(m)})- \frac{N}{2}\langle m_{N}^{2}(\boldsymbol{\sigma})\rangle_{BG}\bigg).
\end{align*}
These derivatives are equal to zero as the following equalities hold
\begin{equation}\label{risultati.verosimiglianza}
\begin{cases}
\langle m_{N}(\boldsymbol{\sigma})\rangle_{BG}=\dfrac{1}{M}\sum\limits_{m=1}^{M}m_{N}(\boldsymbol{\sigma}^{(m)})
\\\\
\langle m_{N}^{2}(\boldsymbol{\sigma})\rangle_{BG} =\dfrac{1}{M}\sum\limits_{m=1}^{M}m_{N}^{2}(\boldsymbol{\sigma}^{(m)}).
\end{cases}
\end{equation}
Thus the function $L(J,h)$ reaches its maximum when the average value and the variance of the magnetization are calculated from the data according to (\ref{risultati.verosimiglianza}). Once we have compute these values, appling the procedure explained above to solve the inverse problem we obtain that the maximum likelihood estimators for $J$ and $h$ are given respectively by (\ref{J.naive}) and (\ref{h.naive}) where the average value and the variance of the magnetization are computed from the data according to (\ref{risultati.verosimiglianza}). 

\section{Inverse problem for the multi-species mean-field model}
The inverse problem for the multi-species mean field model is handled in a similar way.
Suppose to know the average value and the correlations of the magnetizations of a model defined by Hamiltonian
\begin{equation*}
H_{N}(\boldsymbol{\sigma})=-N\Big(\frac{1}{2}\sum\limits_{l, s=1}^{n}\alpha_{l}\alpha_{s}J_{ls}m_{l}(\boldsymbol{\sigma})m_{s}(\boldsymbol{\sigma})+\sum\limits_{l=1}^{n}\alpha_{l}h_{l}m_{l}(\boldsymbol{\sigma})\Big)
\end{equation*}
and distribution 
\begin{equation}\label{misura.BG.multi.replica}
P_{N,\mathbf{J},\mathbf{h}}\{\boldsymbol{\sigma}\}=\frac{\exp(-H_{N}(\boldsymbol{\sigma}))}{Z_{N}(\mathbf{J},\mathbf{h})}\prod\limits_{i=1}^{N}d\rho(\sigma_{i})
\end{equation}
where $\rho$ is given by (\ref{ro.curie.replica}). We recall that by the definition of the measure $\rho$ the distribution (\ref{misura.BG.multi.replica}) can be write as:
\begin{equation}
P_{N,\mathbf{J},\mathbf{h}}\{\boldsymbol{\sigma}\}=\dfrac{\exp(-H_{N}(\boldsymbol{\sigma}))}{\sum\limits_{\boldsymbol{\sigma}\in\Omega_{N}}\exp(-H_{N}(\boldsymbol{\sigma}))}.
\end{equation}
If the function $f$ defined in (\ref{funzionale.pressione.multi}) has a unique maximum point $\boldsymbol{\mu}=(\mu_{1},\dots,\mu_{n})$ of maximal type, by theorem \ref{teo.multi.1} the following identities hold:
\begin{equation}\label{identita.multi}
\lim_{N\rightarrow\infty}\langle m_{l}(\boldsymbol{\sigma})\rangle_{BG}=\mu_{l}\quad l=1,\dots,n.
\end{equation}
Differentiating the identities (\ref{identita.multi}) with respect to $h_{s}$, $s=1,\dots,n$ we obtain
\begin{equation*}
 \lim_{N\rightarrow\infty}\frac{\partial }{\partial h_{s}}\langle m_{l}(\boldsymbol{\sigma})\rangle_{BG}=\chi_{ls}\quad l,s=1,\dots,n
\end{equation*}
where $\chi_{ls}$ are the elements of susceptibility matrix. In particular
\begin{align*}
\chi_{ls}=\frac{\partial\mu_{l}}{\partial h_{s}}&=\frac{\partial}{\partial h_{s}}\Big(\tanh\Big(h_{l}+\sum\limits_{p=1}^{n}\alpha_{p}J_{lp}\mu_{p}\Big)\Big)\nonumber\\
&=\Big(1-\tanh^{2}\Big(h_{l}+\sum\limits_{p=1}^{n}\alpha_{p}J_{lp}\mu_{p}\Big)\Big)\Big(\delta_{ls}+\sum\limits_{p=1}^{n}\alpha_{p}J_{lp}\frac{\partial\mu_{p}}{\partial h_{s}}\Big)\nonumber\\
&=(1-\mu_{l}^{2})\Big(\delta_{ls}+\sum\limits_{p=1}^{n}\alpha_{p}J_{lp}\chi_{ps}\Big)
\end{align*}
where $\delta_{ls}$ denots the delta of Dirac picked in $l=s$. Thus the susceptibility matrix can be write as:
\begin{equation}\label{chi.matrice}
\boldsymbol{\chi}=\mathbf{P}(\mathbf{I}+\mathbf{J}\mathbf{D}_{\boldsymbol{\alpha}}\mathbf{D}_{\boldsymbol{\alpha}}\boldsymbol{\chi})
\end{equation}\\
 where $\mathbf{P}=diag\{1-\mu_{1}^{2},\dots,1-\mu_{n}^{2}\}$, $\mathbf{D}_{\boldsymbol{\alpha}}=diag\{\sqrt{\alpha_{1}},\dots,\sqrt{\alpha_{n}}\}$, $\mathbf{I}$ is the identity matrix and $\mathbf{J}$ is the reduce interaction matrix.
Moreover, for each $l,s=1,\dots,n$
\begin{align}\label{chi.multi.2}
\frac{\partial}{\partial h_{s}}\langle &m_{l}(\boldsymbol{\sigma})\rangle_{BG}\nonumber\\ 
&=\frac{\partial}{\partial h_{s}}\bigg(\dfrac{\sum_{\boldsymbol{\sigma}\in\Omega_{N}}m_{l}(\boldsymbol{\sigma})e^{-H_{N}(\boldsymbol{\sigma})}}{\sum_{\boldsymbol{\sigma}\in\Omega_{N}}e^{-H_{N}(\boldsymbol{\sigma})}}\bigg)\nonumber\\\nonumber\\
&=N_{s}\dfrac{\sum_{\boldsymbol{\sigma}\in\Omega_{N}}m_{l}(\boldsymbol{\sigma})m_{s}(\boldsymbol{\sigma})e^{-H_{N}(\boldsymbol{\sigma})}}{\sum_{\boldsymbol{\sigma}\in\Omega_{N}}e^{-H_{N}(\boldsymbol{\sigma})}}\nonumber\\
&\quad-N_{s}\bigg(\dfrac{\sum_{\boldsymbol{\sigma}\in\Omega_{N}}m_{l}(\boldsymbol{\sigma})e^{-H_{N}(\boldsymbol{\sigma})}}{\sum_{\boldsymbol{\sigma}\in\Omega_{N}}e^{-H_{N}(\boldsymbol{\sigma})}}\bigg)\bigg(\dfrac{\sum_{\boldsymbol{\sigma}\in\Omega_{N}}m_{s}(\boldsymbol{\sigma})e^{-H_{N}(\boldsymbol{\sigma})}}{\sum_{\boldsymbol{\sigma}\in\Omega_{N}}e^{-H_{N}(\boldsymbol{\sigma})}}\bigg)\nonumber\\\nonumber\\
&=N_{s}\Big(\langle m_{l}(\boldsymbol{\sigma})m_{s}(\boldsymbol{\sigma})\rangle_{BG}-\langle m_{l}(\boldsymbol{\sigma})\rangle_{BG}\langle m_{s}(\boldsymbol{\sigma})\rangle_{BG}\Big).
\end{align}\\
Computing the elements of $\boldsymbol{\chi}$ according to (\ref{chi.multi.2}) by (\ref{chi.matrice}) we get an expression of the reduced intaraction matrix $\mathbf{J}$ related to the average value and the correlations of the magnetizations:
\begin{equation}\label{J.inverso.multi}
 \mathbf{J}=(\mathbf{P}^{-1}-\boldsymbol{\chi}^{-1})\mathbf{D}_{\boldsymbol{\alpha}}^{-1}\mathbf{D}_{\boldsymbol{\alpha}}^{-1}.
\end{equation}
Once the matrix $\mathbf{J}$ is determinated, the elements of the vector $\mathbf{h}$ are obtained inverting the mean field equations (\ref{campomedio.multi})
\begin{equation}\label{h.inverso.multi}
h_{l}=\tanh^{-1}(\mu_{l})-\sum\limits_{s=1}^{n}\;\alpha_{s}J_{ls}\mu_{s}\quad l=1,\dots,n.
\end{equation}

In the previous section we have seen that the inverse problem together with the maximum likelihood estimation avoids to measure the parameters of a Boltzmann-Gibbs's distribution from a suitable sample of empirical data. 

Consider a sample of $M$ configurations of spins $\boldsymbol{\sigma}^{(1)},\dots,\boldsymbol{\sigma}^{(M)}$ indipendent and identically distribuited according to the Boltzmann-Gibbs measure (\ref{misura.BG.multi.replica}). 
The maximum likelihood function related to the sample is
\begin{align}\label{verosimiglianza.multi}
L(\mathbf{J},\mathbf{h})&=P_{N,\mathbf{J},\mathbf{h}}\Big\{\boldsymbol{\sigma}^{(1)},\dots,\boldsymbol{\sigma}^{(M)}\Big\}\nonumber\\
&=\prod_{m=1}^{M}P_{N,\mathbf{J},\mathbf{h}}\Big\{\boldsymbol{\sigma}^{(m)}\Big\}\nonumber\\
&=\prod_{m=1}^{M}\frac{\exp\Big(-H_{N}(\boldsymbol{\sigma}^{(m)})\Big)}{Z_{N}(\mathbf{J},\mathbf{h})}.
\end{align}
Differentiating the logarithm of the likelihood function (\ref{verosimiglianza.multi})
\begin{equation*}
\ln L(\mathbf{J},\mathbf{h})=\sum_{m=1}^{M}\bigg(-H_{N}(\boldsymbol{\sigma}^{(m)})-\ln Z_{N}(\mathbf{J},\mathbf{h})\bigg)
\end{equation*}
with respect to $h_{l}$ and $J_{ls}$, $l,s=1,\dots,n$ we obtain
\begin{align*}
\frac{\partial\ln L(\mathbf{J},\mathbf{h})}{\partial h_{l}}&=N_{l}\sum_{m=1}^{M}\bigg(m_{l}(\boldsymbol{\sigma}^{(m)})-\dfrac{\sum_{\boldsymbol{\sigma}}\exp(-H_{N}(\boldsymbol{\sigma}))m_{l}(\boldsymbol{\sigma})}{Z_{N}(\mathbf{J},\mathbf{h})}\bigg)\\
&=N_{l}\sum_{m=1}^{M}\bigg(m_{l}(\boldsymbol{\sigma}^{(m)})-\langle m_{l}(\boldsymbol{\sigma})\rangle_{BG}\bigg)\\\\
\frac{\partial\ln L(\mathbf{J},\mathbf{h})}{\partial J_{ls}}&=\frac{N\alpha_{l}\alpha_{s}}{2}\sum_{m=1}^{M}\bigg(m_{l}(\boldsymbol{\sigma}^{(m)})m_{s}(\boldsymbol{\sigma}^{(m)})\\
&\;\;\qquad\qquad\qquad\qquad-\dfrac{\sum_{\boldsymbol{\sigma}}\exp(-H_{N}(\boldsymbol{\sigma}))m_{l}(\boldsymbol{\sigma})m_{s}(\boldsymbol{\sigma})}{Z_{N}(\mathbf{J},\mathbf{h})}\bigg)\\
&=\frac{N\alpha_{l}\alpha_{s}}{2}\sum_{m=1}^{M}\bigg(m_{l}(\boldsymbol{\sigma}^{(m)})m_{s}(\boldsymbol{\sigma}^{(m)})-\langle m_{l}(\boldsymbol{\sigma})m_{s}(\boldsymbol{\sigma})\rangle_{BG}\bigg).
\end{align*}\newpage
These derivatives are equal to zero as the following equalities hold\\
\begin{align}\label{ultimo.cap.5}
\begin{cases}
\langle m_{l}(\boldsymbol{\sigma})\rangle_{BG}=\dfrac{1}{M}\sum\limits_{m=1}^{M}m_{l}(\boldsymbol{\sigma}^{(m)})\quad l=1,\dots,n\\\\
\langle m_{l}(\boldsymbol{\sigma})m_{s}(\boldsymbol{\sigma})\rangle_{BG}=\dfrac{1}{M}\sum\limits_{m=1}^{M}m_{l}(\boldsymbol{\sigma}^{(m)})m_{s}(\boldsymbol{\sigma}^{(m)})\quad l,s=1,\dots,n.
\end{cases}
\end{align}\\
Now applying the inverse problem we obtain that the maximum likelihood estimator for the reduce interaction matrix $\mathbf{J}$ and the vector field $\mathbf{h}$ are given respectively by (\ref{J.inverso.multi}) and (\ref{h.inverso.multi}) where the average value and the correlations of the magnetizations are compute from the empirical data according to (\ref{ultimo.cap.5}).
\clearpage{\pagestyle{empty}\cleardoublepage}

\chapter*{Acknowledgements}
%\rhead[\fancyplain{}{\bfseries
%Acknowledgements}]{\fancyplain{}{\bfseries\thepage}}
%\lhead[\fancyplain{}{\bfseries\thepage}]{\fancyplain{}{\bfseries
%Acknowledgements}} 
\addcontentsline{toc}{chapter}{Acknowledgements} 
I thank A. Parmeggiani for his important help. I also thank P. Contucci, F. Unguendoli, C. Giardina, C. Giberti, C.Vernia, R. Burioni, E. Agliari, S. Loreti and A. Bianchi for useful discussions.
\clearpage{\pagestyle{empty}\cleardoublepage}
\clearpage{\pagestyle{empty}\cleardoublepage}
\nocite{*}
\addcontentsline{toc}{chapter}{Bibliography}
\rhead[\fancyplain{}{\bfseries
Bibliography}]{\fancyplain{}{\bfseries\thepage}}
\lhead[\fancyplain{}{\bfseries\thepage}]{\fancyplain{}{\bfseries
Bibliography}} 
\bibliography{bibliodott}{}
\bibliographystyle{amsalpha}

\end{document}